%% file: main-arxiv.tex
\documentclass[a4paper,UKenglish,cleveref,autoref,thm-restate]{lipics-v2021}

\pdfoutput=1
\hideLIPIcs

\input{./extras/metadata}

\input{./extras/macros}

\newtheoremstyle{empty}
  {\topsep}   %
  {\topsep}   %
  {\itshape}  %
  {}       %
  {\bfseries} %
  {}         %
  {5pt plus 1pt minus 1pt} %
  {\textcolor{darkgray}{$\blacktriangleright$}}          %

\newtheorem*{theorem*}{Theorem}

\theoremstyle{empty}
\newtheorem*{statement*}{}

\begin{document}
	
	\maketitle

\input{./extras/abstract}

	\section{Introduction}\label{sec:introduction}
	\input{./sections/introduction.tex}

\input{./sections/relatedwork.tex}

	\section{Preliminaries} \label{sec:preliminaries}

\input{./sections/preliminaries-new}

	\section{The automata disambiguation framework}
	\input{./sections/framework}\label{sec:framework}
	
	\section{Disambiguation schemes for different levels of ambiguity} \label{sec:disambiguation-conditions}
	\input{./sections/disambiguationConditions-new.tex}

	\section{Disambiguation of weighted automata} \label{sec:weighted}
	\input{./sections/weighted.tex}

	\section{Future work}

\input{./sections/conclusions-esa.tex}\label{sec:conclusion}
	
	\newpage
	
	\bibliography{./extras/bib2doi}

	\newpage
	\appendix
	
	\section{A counterexample to Mohri's disambiguation algorithm} \label{app:mohri}
	
	\input{./sections/app-mohri-counterexample.tex}

	\section{Proofs of Section~\ref{sec:framework}} \label{app:framework}

\input{./sections/app-framework.tex}

	\section{Proofs of Section~\ref{sec:disambiguation-conditions}} \label{app:disambiguation-conditions}

\input{./sections/app-disambiguationConditions.tex}

	\section{Proofs of Section~\ref{sec:weighted}} \label{app:weighted}

\input{./sections/app-weighted.tex}

\end{document}

%% file: extras/metadata.tex
\hideLIPIcs

\title{A simple algorithmic framework for disambiguation of finite automata} 
\titlerunning{A simple algorithmic framework for disambiguation of finite automata} 

\author{Mauricio Cari}{Pontificia Universidad Católica de Chile, Chile \\ 
	Millennium Institute for Foundational Research on Data, Chile}{mauricio.cari@uc.cl}{https://orcid.org/0009-0005-5059-8068}{}
\author{Martín Muñoz}{Univ. Artois, CNRS, UMR 8188, Centre de Recherche en Informatique de Lens (CRIL), F-62300 Lens, France; Pontificia Universidad Católica de Chile \& Millennium Institute for Foundational Research on Data (IMFD), Chile}{munoz@cril.fr}{https://orcid.org/0009-0003-3294-6159}{}
\author{Cristian Riveros}{Pontificia Universidad Católica de Chile, Chile \\ 
	Millennium Institute for Foundational Research on Data, Chile}{cristian.riveros@uc.cl}{https://orcid.org/0000-0003-0832-116X}{}

\authorrunning{Mauricio Cari, Martín Muñoz, and Cristian Riveros}
\Copyright{Mauricio Cari, Martín Muñoz, and Cristian Riveros}

\ccsdesc[500]{Theory of computation~Automata extensions}%

\keywords{Algorithmic automata theory, unambiguous automata models, degree of ambiguity, determinization, disambiguation, weighted automata.}

\nolinenumbers %

\funding{The work of Cari and Riveros was supported by ANID Fondecyt Regular project 1230935 and ANID – Millennium Science Initiative Program – Code ICN17\_002.}

%% file: extras/macros.tex
\usepackage[utf8]{inputenc}
\usepackage{xspace}
\usepackage{algorithm}
\usepackage[noend]{algpseudocode}
\usepackage{amsthm}
\usepackage{textcomp}
\usepackage{stmaryrd}
\usepackage{varwidth}
\usepackage{graphicx}
\usepackage{thmtools} %

\usepackage{tikz}
\usetikzlibrary{automata, graphs, fit, shapes, positioning,chains,arrows,snakes,decorations.pathmorphing, calc}

\makeatletter
\@fleqnfalse
\@mathmargin\@centering
\makeatother

\usepackage
[disable]
{todonotes}

\newcommand{\daa}{\mathsf{da}_{\cA}}
\newcommand{\da}{\mathsf{da}}
\newcommand{\degg}{\mathsf{deg}}

\newcommand{\dom}{\mathsf{dom}}
\newcommand{\detr}{\mathsf{det}}

\newcommand{\runs}{\mathsf{runs}}
\newcommand{\supp}{\mathsf{supp}}

\newcommand{\CF}{\mathsf{CF}}
\newcommand{\ICF}{\mathsf{ICF}}
\newcommand{\ECF}{\mathsf{ECF}}

\newcommand{\reach}{\operatorname{Reach}}

\newcommand{\oracleSingle}{\cO_{\mathsf{single}}}
\newcommand{\oracleFull}{\cO_{\mathsf{full}}}

\newcommand{\disDFA}{\cD_{\mathsf{det}}}
\newcommand{\disUFA}{\cD_{\CF}}

\newcommand{\cA}{\mathcal{A}}

\newcommand{\cL}{\mathcal{L}}
\newcommand{\cO}{\mathcal{O}}
\newcommand{\cD}{\mathcal{D}}
\newcommand{\cC}{\mathcal{C}}

\newcommand{\bbN}{\mathbb{N}}

\newcommand{\bbZ}{\mathbb{Z}}

\newcommand{\set}[1]{\{#1\}}

\newcommand{\bbM}{\mathbb{M}}
\newcommand{\bbS}{\mathbb{S}}
\newcommand{\zero}{\overline{0}}
\newcommand{\one}{\overline{1}}
\newcommand{\bbT}{\mathbb{T}}

\newcommand{\cW}{\mathcal{W}}
\newcommand{\vI}{\vec{I}}
\newcommand{\vF}{\vec{F}}
\newcommand{\cF}{\mathcal{F}}

\newcommand{\DeltaFA}{\Delta_{\operatorname{FA}}}
\newcommand{\IFA}{I_{\operatorname{FA}}}
\newcommand{\FFA}{F_{\operatorname{FA}}}

\newcommand{\res}{\mathsf{res}}
\newcommand{\fac}{\mathsf{fac}}

\newcommand{\resMohri}{\res_{\operatorname{M}}}
\newcommand{\facMohri}{\fac_{\operatorname{M}}}
\newcommand{\cFMohri}{\cF_{\operatorname{M}}}
\newcommand{\cFK}{\cF_{\operatorname{K}}}
\newcommand{\PiId}{\Pi_{\operatorname{id}}}

\renewcommand{\paragraph}[1]{\smallskip
	\noindent\textbf{#1}.}

\renewcommand{\operatorname}[1]{\mathsf{#1}}

\algrenewcommand\algorithmicindent{1.3em}%
\algnewcommand\algorithmicforeach{\textbf{for each}}
\algdef{S}[FOR]{ForEach}[1]{\algorithmicforeach\ #1\ \algorithmicdo}

\newcommand{\sem}[1]{\llbracket #1 \rrbracket}
\newcommand{\eps}{\varepsilon}

\makeatletter
\renewcommand\section{\@startsection{section}{1}{\z@}%
	{-0.8\baselineskip plus -0.2\baselineskip minus -0.1\baselineskip}%
	{0.35\baselineskip plus 0.1\baselineskip}%
	{\normalfont\bfseries\large\raggedright}}

\renewcommand\subsection{\@startsection{subsection}{2}{\z@}%
	{-0.6\baselineskip plus -0.2\baselineskip minus -0.1\baselineskip}%
	{0.25\baselineskip plus 0.1\baselineskip}%
	{\normalfont\bfseries\normalsize\raggedright}}

\def\thm@space@setup{%
	\thm@preskip=3pt plus 1pt minus 1pt  %
	\thm@postskip=3pt plus 1pt minus 1pt %
}

\makeatother

\AtBeginDocument{%
	\setlength{\abovedisplayskip}{0.5\abovedisplayskip}%
	\setlength{\belowdisplayskip}{0.5\belowdisplayskip}%
	\setlength{\abovedisplayshortskip}{0.5\abovedisplayshortskip}%
	\setlength{\belowdisplayshortskip}{0.5\belowdisplayshortskip}%
}

%% file: extras/abstract.tex
\begin{abstract}
We study the task of disambiguation of finite state automata, namely, converting an automaton into an equivalent, unambiguous one. We do this by developing a novel and simple algorithmic framework that generalizes the subset construction for determinization, and that satisfies some desirable properties: (1) it preserves the original automaton if it was already unambiguous, (2) it computes the successor states on-the-fly and (3) computes each new state in polynomial time---this last point is crucial as it guarantees that the running time is polynomial in the size of the output automaton. Then, we show how to apply this framework for {\em partial} disambiguation: by changing the criterion that builds the new states, we develop algorithms for different levels of ambiguity, namely, finitely ambiguous, and polynomially ambiguous automata. These algorithms also satisfy condition (1) for their respective levels, and also (2) and (3). Finally, we show that the disambiguation framework can easily be extended to other models of automata like weighted automata. %

\end{abstract}

%% file: sections/introduction.tex
Automata theory is today one of the main areas of theoretical computer science, studying computational models with restricted resources and having applications in different areas of computer science.
For these applications, finding automata models with good algorithmic properties is a crucial task, and \emph{determinism} is probably the most well-known condition for reaching them. For example, deterministic finite automata are efficiently closed under several operations (i.e., with a polynomial-size output); the evaluation problem can be efficiently computed; equivalence and containment problems are tractable (and intractable in general)~\cite{StockmeyerM73}; the number of states can be efficiently minimized~\cite{hopcroft1971n}; and learning can be efficiently performed in an active setting~\cite{angluin1987learning}, among other results~\cite{handbook-automata}. 

In practice, systems employ \emph{non-deterministic} automata models that are common in user applications. These are, in general, exponentially more succinct than their deterministic counterparts, but unfortunately do not possess such good algorithmic properties. 
To address this problem, systems typically rely on {\em determinizing} non-deterministic models; that is, converting a non-deterministic finite automaton into a deterministic one. 
The {\em subset construction}, introduced by Rabin and Scott~\cite{RabinS59}, is the most used and best-known determinization procedure, in which non-determinism is simulated by maintaining subsets of states. 
Although the output of determinizing the automata could have up to exponential size~\cite{meyer1971economy}, the subset construction has the advantage that it is simple and easy to implement, and can be computed on-the-fly (i.e., only the necessary part of the determinization is expanded).
Indeed, this is probably the reason why the procedure is used so extensively (see, e.g.,~\cite{holzmann1997model,cox2007regular,AllauzenRSSM07,riveros2023rematch,BucchiGQRV22}). As an example, in regular expression (regex) evaluation, a regex engine (without backtracking)~\cite{cox2007regular} usually evaluates a deterministic finite automaton on-the-fly: the running of the automaton is simulated by maintaining the current set of active states, caching the previously computed sets, and computing the next set for each new arriving character (if it was not already computed). This procedure allows for exploring only the sets of states reached by the data, and the next state can be easily computed from the previous state.

An alternative notion to determinism are \emph{unambiguous} automata models, which are non-deterministic machines with at most one successful run per input. Unambiguous finite automata generalize deterministic finite automata, offering a good balance between efficiency and succinctness~\cite{Colcombet15}.
Specifically, they usually admit efficient algorithms; for example, the equivalence and containment problems are tractable~\cite{stearns1985equivalence}, one can compute the number of accepted words (of a fixed length) in polynomial time~\cite{ArenasCJR21}, or efficiently enumerate them with constant delay~\cite{AmarilliBJM17}. %
Sometimes, these properties also extend to more general notions of ambiguity, like \emph{finitely ambiguous} models (i.e., each input is accepted by at most a finite number of accepting runs) or \emph{polynomially ambiguous} models (i.e., each input is accepted by at most a polynomial number of accepting runs in the size of the input) for which researchers have also found interesting results~\cite{weber1991degree,FijalkowRW22,Colcombet15}. 

The disambiguation of finite automata has been stated as a relevant algorithmic problem~in automata theory~\cite{colcombet2022unambiguity,Colcombet12}.
Unlike the determinization procedure for finite automata, there is no standard or well-known procedure that converts any given non-deterministic automata into an equivalent, unambiguous one (besides the determinization itself). 
Like determinization, one would also desire a disambiguation procedure that is simple and easy to implement, and can be efficiently computed on-the-fly (i.e., computing the next set of states in polynomial time). In the past, researchers have proposed procedures for disambiguating finite automata and other models~\cite{Sakarovitch98,mohri13,MohriWFAdisambiguation,Roche,HarjuKL92}; however, these procedures are often cumbersome, cannot be computed on-the-fly, and may even encounter technical difficulties (see Appendix~\ref{app:mohri}). Indeed, today there~is no well-known disambiguation procedure that is used in practice. Furthermore, none of these procedures can be extended to other levels of ambiguity (e.g., finitely ambiguous) or other automata models (e.g., weighted automata). An important exception to this rule is the \emph{unambiguous subset construction} of Weber and Klemm~\cite{WeberEconomy} which, as far as we know, has been largely overlooked by the community.

In this work, we generalize the approach of Weber and Klemm and propose an algorithmic framework for disambiguation of NFA that is arguably simple and can be computed on-the-fly (Section~\ref{sec:framework}). This framework allows for the presentation of efficient disambiguation algorithms for unambiguous, finitely ambiguous, and polynomially ambiguous automata. Furthermore, we show that these constructions preserve the original automata if it already had the target ambiguity, and are optimal in some precise sense (Section~\ref{sec:disambiguation-conditions}). Finally, we show that this framework can be extended to other automata models by presenting general disambiguation algorithms for the model of weighted automata (Section~\ref{sec:weighted}).

%% file: sections/relatedwork.tex
\paragraph{Related work}  Unambiguous, finitely ambiguous, and polynomially ambiguous finite state automata have been extensively studied in the literature for various automata models (see, e.g., \cite{Colcombet15,weber1991degree,handbook-automata}). Our work complements this literature, where we seek useful algorithms for constructing such models.

 Algorithms for disambiguating finite state automata, transducers or weighted automata have been proposed in the literature~\cite{Sakarovitch98,mohri13,MohriWFAdisambiguation,Roche,HarjuKL92}. In general, all of them rely on maintaining the active state and a subset of competing states (i.e., runs) plus some conditions to prune undesirable runs. 
Some of them even require some additional intermediate steps. For example, \cite{Sakarovitch98} needs to compute first the determinization of the automaton, and \cite{MohriWFAdisambiguation, Roche} require a pre-disambiguation step.
As discussed above, these conditions can be rather obscure and their correctness is hard to prove. %
Indeed, in Appendix~\ref{app:mohri} we show that the construction presented in~\cite{mohri13} does not work in general. As we mentioned, our approach is inspired by the construction of Weber and Klemm~\cite{WeberEconomy}. To the best of our knowledge, previous approaches have not studied the disambiguation of finite automata as a general framework, and their construction does not extend to other levels of~ambiguity. 

In the context of weighted automata, ambiguity can define a strict hierarchy of classes of functions, and researchers have studied the problem of deciding whether a given weighted automaton (e.g., over the tropical semiring) is equivalent to a deterministic or unambiguous one~\cite{KirstenOnDeterminizationWeighted,Kirsten08Clones,jecker2024determinisation, almagor2025WFATropicalDecidability, almagor2026WFATropicalUDecidability}.
These articles studied only the decision problem of determinization, and the constructions are usually non-trivial. In contrast, the present work aims for a simple and modular algorithm that
could lead to useful solutions in practice.

\paragraph{Outline of the paper} We start with some preliminaries in Section~\ref{sec:preliminaries}. We present the disambiguation framework in Section~\ref{sec:framework} and use it for finding disambiguation schemes for different levels of ambiguity in Section~\ref{sec:disambiguation-conditions}. We extend this framework to the disambiguation of weighted automata in Section~\ref{sec:weighted}. Finally, we discuss some future work in Section~\ref{sec:conclusion}. Due to space constraints, the proofs of the results can be found in the extended version~\cite{arxivversion}.

%% file: sections/preliminaries-new.tex
\paragraph{Non-deterministic finite automata} A \emph{non-deterministic finite automaton} (NFA) is a tuple $\cA = (Q, \Sigma, \Delta, I, F)$ where $Q$ is a finite set of states, $\Sigma$ is a finite alphabet, $\Delta \subseteq Q \times \Sigma \times Q$ is the transition relation, $I \subseteq Q$ is the set of initial states, and $F \subseteq Q$ is the set of final states. A \emph{partial run} of $\cA$ over a word $w = a_1\ldots a_n$ is a sequence:
\[
\rho \ := \ p_0 \, \xrightarrow{a_1} \, p_1 \, \xrightarrow{a_2} \ \cdots \ \xrightarrow{a_n} \, p_n \tag{$*$}
\] 
where $(p_{i-1}, a_i, p_i)\in \Delta$ for every $i\in [n]$. We say that $\rho$ is a \emph{run} of $\cA$ over $w$ if additionally $p_0 \in I$. Further, we say that a (partial) run $\rho$ is \emph{accepting} if $p_n\in F$.  
The language of $\cA$, denoted as $\cL(\cA)$, is the set of all $w \in \Sigma^*$ such that there exists an accepting run of $\cA$ over $w$.

One can see the transition relation $\Delta$ as a function $\Delta: Q \times \Sigma \rightarrow 2^Q$ where $\Delta(p, a) = \{q \mid (p,a,q) \in \Delta\}$. 
We extend $\Delta$ from letters to words (i.e., $\Delta: Q \times \Sigma^* \rightarrow 2^Q$) as usual: $\Delta(p,\varepsilon) = \{p\}$ and $\Delta(p, a\cdot w) = \bigcup_{q \in \Delta(p,a)} \Delta(q, w)$. 
Furthermore, we can extend $\Delta$ from single states to sets of states $\Delta\colon 2^Q \times \Sigma^* \rightarrow 2^Q$ defined as $\Delta(S, w) = \bigcup_{p\in S} \Delta(p, w)$. One can note that $w \in \cL(\cA)$ iff $\Delta(I, w) \cap F \neq \emptyset$. For the sake of presentation, we will use $\Delta(q, a)$, $\Delta(q, w)$, or $\Delta(S, w)$ with the same $\Delta$ when the input is clear from the context.

We define the size of $\cA$ as $|\cA| = |Q| + |\Delta|$. Further, we say that two NFAs $\cA_1$ and $\cA_2$ are \emph{isomorphic} if $\cA_1$ and $\cA_2$ are identical up to renaming the states.

\paragraph{Trimming} We say that $q \in Q$ is \emph{reachable from $p \in Q$} if there exists a partial run $\rho$ like ($*$) of $\cA$ over some word $w$ such that $p_0 = p$ and $p_n = q$ (i.e., $q \in \Delta(p,w)$ for some~$w$). A state $q \in Q$ is \emph{reachable} if $q$ can be reached from a state in $I$ (i.e., $q \in \Delta(I,w)$ for some~$w$) and is \emph{co-reachable} if some state in $F$ is reachable from $q$ (i.e., $\Delta(q,w) \cap F \neq \emptyset$ for some~$w$). In this work, we assume all NFAs are \emph{trimmed}, namely, that every $q \in Q$ is (co-)reachable. %

\paragraph{Ambiguity} We say that a finite automaton $\cA = (Q, \Sigma, \Delta, I, F)$ is \emph{deterministic} (DFA) iff $|I| = 1$ and for every $p\in Q$ and $a\in \Sigma$ there exists \emph{exactly one} state $q\in Q$ such that $(p,a,q)\in\Delta$. In other words, $\Delta$ forms a function $\Delta: Q \times \Sigma \rightarrow Q$. 
We say that $\cA$ is \emph{unambiguous} (UFA) iff for every $w \in \Sigma^*$ there exists at most one accepting run of $\cA$ over $w$. One can check that if $\cA$ is deterministic then $\cA$ is unambiguous, whereas the converse does not necessarily~hold.
In Figure~\ref{fig:FAexample} (b) and (c), we display a DFA and UFA, respectively.

Given an NFA $\cA$, the \emph{degree of ambiguity} $\daa(w)$ of a word $w$ is the number of different accepting runs of $\cA$ over $w$~\cite{weber1991degree}. Similarly, the \emph{degree of ambiguity} $\da(\cA)$ of $\cA$ is the maximum degree of ambiguity over all $w \in \cL(\cA)$, that is, $\da(\cA) = \max_{w \in \cL(\cA)} \daa(w)$. If no such maximum exists, then $\da(\cA) = \infty$.
We say that an NFA $\cA$ is \emph{finitely ambiguous} (finFA) if $\da(\cA) \leq k$ for some $k \in \bbN$. In this case, we say that $\cA$ is \emph{$k$-ambiguous} ($k$-ambFA). Instead, $\cA$ is called \emph{infinitely ambiguous} if $\da(\cA) = \infty$.
Note that $\cA$ is unambiguous iff $\da(\cA) = 1$ (i.e., $1$-ambiguous). Further, the level of ambiguity forms a hierarchy, namely, every $\ell$-ambFA is also $\ell+1$-ambFA, and so on. In Figure~\ref{fig:FAexample}~(a), we show a finFA $\cA$ with $\da(\cA) = 2$.
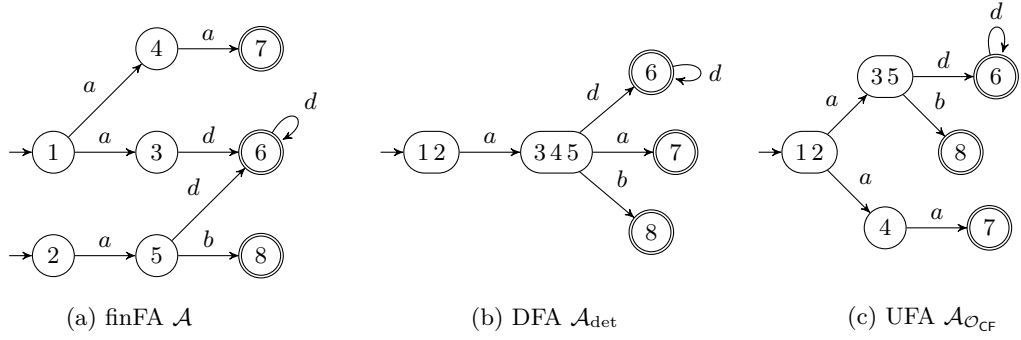
\begin{figure}[t]
    \centering
     \input{./figures/noUnambiguousExample-unified.tex}
    \caption{(a) An example of an NFA $\cA$ with $\da(\cA) = 2$. (b) The determinization $\cA_{\det}$ of $\cA$. (c) The disambiguation $\cA_{\cO_{\CF}}$ of $\cA$.}
    \label{fig:FAexample}
\end{figure}

For infinitely ambiguous automata, one can characterize how the ambiguity increases with the size of the input word~\cite{weber1991degree}.  
The \emph{degree of growth of ambiguity}, denoted as $\degg(\cA)$, is defined as the smallest degree of a polynomial $p:\bbN \rightarrow \bbN$ such that for all $w \in \Sigma^*$, we have $\daa(w) \leq p(|w|)$, if such a polynomial exists. We say that $\cA$ is \emph{polynomially ambiguous} (polyFA) if $\degg(\cA) = k$ for some $k \in \bbN$.
Otherwise, if no such polynomial exists, we define $\degg(\cA) = \infty$, and say that $\cA$ is \emph{exponentially ambiguous} (expFA). In Figure~\ref{fig:conditionExamples} (c) and (a), we show examples of a polyFA and expFA, respectively.

The work by Weber and Seidl~\cite{weber1991degree} defines criteria that characterize the degree and the degree of growth of ambiguity of an NFA. By these criteria, deciding whether an NFA $\cA$ is deterministic, unambiguous, finitely, or polynomially ambiguous can be done in polynomial time over the size of $\cA$.  Conversely, computing the exact degree of ambiguity~$\da(\cA)$~is~PSPACE-complete~\cite{ChanI83} and the degree of growth of ambiguity can be computed~in~polynomial~time~\cite{weber1991degree}.

We denote by DFA, UFA, finFA, and polyFA the class of all finite automata that are deterministic, unambiguous, finitely ambiguous, and polynomially ambiguous, respectively.
One can note that these classes form a strict hierarchy where:
\[
\text{DFA} \ \subsetneq \ \text{UFA} \ \subsetneq  \ \text{finFA} \ \subsetneq \ \text{polyFA} \ \subsetneq \ \text{NFA}
\]
Furthermore, for each pair of classes $\cC_1$ and $\cC_2$ with $\cC_1 \ \subsetneq \ \cC_2$ where $\cC_1, \cC_2 \in \{\text{DFA}, \text{UFA}, \text{finFA}\}$ there exists a family $\{\cA_n\}_{n\in \bbN}$ where each $\cA_n$ has $O(n)$ number of states and $\cA_n \in \cC_2$ such that for every $\cA \in \cC_1$ with $\cL(\cA) = \cL(\cA_n)$, it holds that $|\cA| \in \Omega(2^n)$~\cite{Leiss85a, Leung05}. %
Similarly, there exists a family $\{\cA_n\}_{n \in \mathbb{N}}$ in polyFA where each $\cA_n$ has $O(p(n))$ states for some polynomial $p$, such that for every equivalent automaton in finFA it holds that $|\cA| \in \Omega(2^{n^{1/3}})$~\cite{HromkovicS11}; also, NFA can be super-polynomially more succinct than polyFA (see~\cite{HromkovicSKKS02,Leung98}).

%% file: figures/noUnambiguousExample-unified.tex
\small
\begin{tikzpicture}[every label/.style={black!60}, initial distance= {3mm}, initial text=, ->,>=stealth', node distance=0.8, auto,
	mystate/.style={state, inner sep=3pt, minimum size=0mm},
	longstate/.style={rectangle,inner sep=5pt,draw,rounded corners=7.8pt},
	mystate2/.style={state, inner sep=1pt, minimum size=0mm}]
  \node[mystate,initial, initial where=left]  (1)                {$1$};
  \node[mystate,initial, initial where=left]  (2) [below=of 1]   {$2$};
  \node[mystate]          (3) [right=of 1] {$3$};
  \node[mystate]          (4) [above=of 3] {$4$};
  \node[mystate]          (5) [right=of 2] {$5$};
  \node[mystate,accepting](7) [right=of 4] {$7$};
  \node[mystate,accepting](6) [right=of 3] {$6$};
  \node[mystate,accepting](8) [right=of 5] {$8$};
  
  \node at (1, -2.2) {(a) finFA $\cA$};

  \path[->] (1) edge              node   {$a$} (3)
            (1) edge              node   {$a$} (4)
            (2) edge              node    {$a$} (5)
            (3) edge              node   {$d$} (6)
            (6) edge  [in=30,out=60,loop ]            node   {$d$} (6)
            (4) edge              node   {$a$} (7)
            (5) edge              node   {$d$} (6)
            (5) edge              node   {$b$} (8);
            
  \begin{scope}[xshift=5cm]
  	 \node[longstate,initial, initial where=left]  (1_2) {$1\, 2$};
  	\node[longstate]          (3_4_5) [right=of 1_2]    {$3 \, 4 \, 5$};
  	\node[mystate,accepting]          (3_6) [above right=of 3_4_5]     {$6$};
  	\node[mystate,accepting](7) [right=of 3_4_5]          {$7$};
  	\node[mystate,accepting](8) [below right=of 3_4_5]    {$8$};
  	
  	\path[->] (1_2) edge              node    {$a$} (3_4_5)
  	(3_4_5) edge[shorten <=-1.1pt]            node    {$d$} (3_6)
  	(3_6) edge [loop right]       node   {$d$} (3_6)
  	(3_4_5) edge            node   {$a$} (7)
  	(3_4_5) edge[shorten <=-1.1pt]            node   {$b$} (8);
  	
  	\node at (1.5, -2.2) {(b) DFA $\cA_{\det}$};
  \end{scope}
  
  \begin{scope}[xshift=10cm]
  			\node[longstate,initial, initial where=left]  (1_2) at (0,0)   {$1\, 2$};
  			\node[longstate]          (3_5) at (1,1) {$3\, 5$};
  			\node[mystate]          (4) at (1,-1)   {$4$};
  			\node[mystate,accepting](6) [right=of 3_5]    {$6$};
  			\node[mystate,accepting](8) at (2,0)   {$8$};
  			\node[mystate,accepting](7) [right=of 4]      {$7$};
  			
  			\path[->] (1_2) edge[shorten <=-1.8pt, shorten >=-1.8pt]              node    {$a$} (3_5)
  			(1_2) edge[shorten <=-1.8pt, shorten >=-0.1pt]              node   {$a$} (4)
  			(3_5) edge              node    {$d$} (6)
  			(6) edge  [loop above]            node   {$d$} (6)
  			(3_5) edge[shorten <=-1.8pt]              node    {$b$} (8)
  			(4)   edge              node    {$a$} (7);
  			
  			\node at (1.5, -2.2) {(c) UFA $\cA_{\cO_{\CF}}$};
  \end{scope}
            
\end{tikzpicture}

%% file: sections/framework.tex
The goal of this section is to define a framework suited for disambiguation of finite state automata, namely, algorithms that convert any NFA into a UFA, finFA, or polyFA. 
We start by introducing the motivation of the framework and then provide the formal~definitions. 

\paragraph{Determinization} Let us highlight the classical \emph{subset construction} for determinization, first described by Rabin and Scott~\cite{RabinS59}, 
that receives $\cA = (Q, \Sigma, \Delta, I, F)$ and produces a~DFA:
\[
\cA_\detr := (2^Q, \Sigma, \Delta_\detr, \{I\}, \{S \mid S \cap F \neq \emptyset\}),
\]
its states are the subsets of $Q$, its single initial state is $I$, and its final states are the sets that contain at least one state from $F$. Further,  $(S_1, a, S_2) \in \Delta_\detr$ iff $S_2 = \Delta(S_1, a)$ for every $S_1, S_2 \in 2^Q$ and $a \in \Sigma$.
The NFA $\cA_\detr$ is deterministic and $\cL(\cA) = \cL(\cA_\detr)$.

It is important to note that, since we assume that all our NFA are trimmed, we assume that $\cA_\detr$ is trimmed after the construction (i.e., we only keep the subsets in $2^Q$ that reach $F$ and are reachable from $I$). One can prove that, given that $\cA$ is initially trimmed, all the states reachable from $I$ are reachable and co-reachable in $\cA_{\det}$. In Figure~\ref{fig:FAexample} (b), we display the determinization $\cA_{\det}$ of $\cA$ where we only have the states reachable from $\{1,2\}$.

Although the automaton $\cA_\detr$ could be of size exponential with respect to $|Q|$, we claim that it has three algorithmic properties that makes it suitable to be used in practice.
\smallskip
\begin{enumerate} \itemsep1mm
	\item (\textbf{idempotent}) If $\cA$ is deterministic, then $\cA_\detr$ is \emph{isomorphic} to $\cA$. 
	\item (\textbf{on-the-fly}) $\Delta_\detr$ can be computed \emph{on-the-fly}, namely, for any $S \in Q_\detr$ one can compute $\Delta_\detr(S, a)$ only from $\Delta$ and $S$.
	\item (\textbf{efficient}) Given an $S\in Q_\detr$ and $a\in\Sigma$, one can compute $\Delta_\detr(S, a)$ in {\em polynomial time} in $|S|$ and $|\Delta|$. 
\end{enumerate}
\smallskip
The first property ensures that the procedure preserves the automata if it already possesses the desired properties. The second and third are used in~practice~(see,~e.g.,~\cite{holzmann1997model,cox2007regular,AllauzenRSSM07,riveros2023rematch,BucchiGQRV22}) since one can keep a cache of explored sets and efficiently compute the~next~one~from~a~given~set. 

We want to have the same properties for a disambiguation procedure that takes any~NFA $\cA$ and generates an equivalent NFA with the desired ambiguity. It is important to remark that all disambiguation procedures proposed in the literature~\cite{Sakarovitch98,mohri13,MohriWFAdisambiguation,Roche,HarjuKL92}~fail~to~achieve~at~least one of them (and most fail all of them). 
For instance, Mohri's procedure~\cite{mohri13}~and Schützenberger's construction~\cite{Sakarovitch98}  cannot be computed on-the-fly, since they require knowledge about previously computed states to discard transitions and disallow finality of new~states.~Moreover, all of them work for converting NFA to UFA, and none of them translates into~finFA~or~polyFA. 

\paragraph{The framework} Let $\cA = (Q, \Sigma, \Delta, I, F)$ be an NFA. We define a \emph{partition oracle}~$\cO$~for~$\cA$~as a function $\cO:2^Q \mapsto 2^{2^Q}$ that receives an $S\subseteq Q$, and outputs a partition of $S$, namely, $\cO(S) = \{S_1,\ldots,S_m\}$ where $S = S_1 \uplus \cdots \uplus S_m$.
Given a partition~oracle~$\cO$~for~$\cA$,~we~define~the~NFA:
\[
\cA_{\cO} \ := \ (2^Q, \Sigma, \Delta_{\cO}, \cO(I), \{S \mid S \cap F \neq \emptyset\}).
\]  
This is similar to Rabin-Scott's $\cA_{\det}$, except the transition relation and the set of initial states depend on $\cO$.
Specifically, we define the transition relation $\Delta_{\cO}$ as:
\[
\Delta_{\cO} \ := \  \big\{\, (S_1, a, S_2) \, \mid \, S_2 \in \cO(\Delta(S_1, a)) \, \big\},
\]
and more succinctly, $\Delta_{\cO}(S, a) = \cO(\Delta(S, a))$ for $S\subseteq Q$ and $a\in\Sigma$. In other words, finding the transitions $(S_1, a, S_2)\in\Delta_{\cO}$ for $S_1$ and $a$ can be done by first computing $\Delta(S_1, a)$ and then applying the partition oracle $\cO$ over it. We recall again that, similar to $\cA_{\det}$, we assume that $\cA_{\cO}$ is trimmed after the construction by keeping only the states reachable from $\cO(I)$.

A crucial property of $\cA_{\cO}$ is that it preserves equivalence for every partition oracle $\cO$.

\begin{restatable}{proposition}{frameworkEquivalence}\label{prop:frameworkEquivalence}
	For every partition oracle $\cO$ of $\cA$, it holds that $\cL(\cA) = \cL(\cA_{\cO})$.
\end{restatable} 

Given this construction, one could intuit that there exists a disambiguation  \emph{spectrum} from oracles that partition their sets more finely or more coarsely.
Indeed, the trivial partition oracles $\oracleSingle$ and $\oracleFull$, where $\oracleSingle(S) := \{\{q\} \mid q \in S\}$ and $\oracleFull(S) := \{S\}$, define two extremes: $\cA_{\oracleSingle}$ is isomorphic to $\cA$ and $\cA_{\oracleFull}$ is isomorphic to $\cA_{\detr}$.

Given an NFA $\cA$, our main goal is to find a partition oracle $\cO$ for which $\cA_{\cO}$ has a desired property (e.g., being unambiguous) and we want to compute $\cO(S)$ from $\cA$ and $S$ efficiently. For this purpose, we define a \emph{disambiguation scheme} as any function $\cD$ that, given an NFA~$\cA$, defines a partition oracle $\cD(\cA)$ for $\cA$. Furthermore, we say that a disambiguation scheme $\cD$ is \emph{efficient} 
if there is an algorithm that receives as input any NFA $\cA = (Q, \Sigma, \Delta, I, F)$ and $S \subseteq Q$ and outputs $[\cD(\cA)](S)$ 
in polynomial time in $|\cA|$ and $|S|$. 
As an example, there exists an efficient disambiguation scheme $\disDFA$ 
given by defining $\disDFA(\cA) = \oracleFull$ for every NFA.

Our disambiguation framework involves identifying efficient disambiguation 
schemes for various levels of ambiguity. %
Given a class $\cC$ of NFA (e.g., $\cC = \text{UFA}$), we say that a 
disambiguation scheme $\cD$ is a \emph{$\cC$-disambiguation scheme} 
iff $\cA_{\cO} \in \cC$ and $\cA_{\cO}$ is isomorphic to $\cA$ whenever 
$\cA \in \cC$, for every NFA $\cA$ and $\cO = \cD(\cA)$. 
Continuing the example, $\disDFA$ is a DFA-disambiguation scheme.
Notice that if we find an efficient $\cC$-disambiguation scheme $\cD$ then 
$\cD$ will satisfy our desirable algorithmic properties 1. to 3. for the class $\cC$ 
(i.e., similarly to the Rabin-Scott subset construction), namely, (1) if $\cA \in \cC$, then $\cA_{\cO}$ is 
isomorphic to $\cA$, (2) $\Delta_{\cO}$ can be computed on-the-fly, and (3) 
$\Delta_{\cO}$ can be computed efficiently. 

\paragraph{Discussion} The proposed framework 
for disambiguating NFA %
is {\em a} generic strategy for this task, but, of course, it is not the only one. 
One could also disambiguate NFA through other strategies 
(e.g.~\cite{Sakarovitch98, mohri13}) with their own desirable properties. 
The advantages of our approach are its simplicity and generality;
and, since each new state is created efficiently and on-the-fly,
it also offers runtime guarantees that depend on the size of the resulting automaton. 
Having a generic framework for disambiguation is also quite versatile: 
we can now define a wide array of algorithms to reach different
levels of ambiguity, and compare their properties.%

%% file: sections/disambiguationConditions-new.tex
We present disambiguation schemes for constructing UFA, 
finFA, and polyFA, and study their properties.
We start by presenting our general strategy for finding partition oracles, 
which is based on binary relations over states. 
Then, we instantiate this strategy for each~class.

\paragraph{Relational partition oracles} First, we introducing some useful notation. %
Let $R \subseteq A \times A$ be a symmetric relation (i.e., $(p,q) \in R$ implies $(q, p) \in R$) over some non-empty set $A$. 
We define by $\dom(R) = \{p \mid (p,q) \in R\}$ the active domain of $R$ (note that $\dom(R)$ might be a strict subset of $A$). 
We denote by $\kappa(R)$ the set of connected components of the undirected graph $(A, R)$ where $\kappa(R)$ always forms a partition~of~$\dom(R)$. 

Let $\cA = (Q, \Sigma, \Delta, I, F)$ be an NFA. 
Given a  symmetric and reflexive (i.e., $(p,p) \in R$ for every $p \in Q$) relation $R$ over $Q$, we define the oracle $\cO_R: 2^{Q} \mapsto 2^{2^Q}$ by $\cO_R(S) = \kappa(R \cap (S \times S))$. 
Since $R$ is reflexive, $\dom(R \cap (S \times S)) = S$ and then $\kappa(R \cap S \times S)$ forms a partition of~$S$.
Therefore, $\cO_R$ is well defined as a partition oracle of $\cA$. 
We say that $\cO$ is a \emph{relational partition oracle} of $\cA$ if there exists a relation $R$ over $Q$ with $\cO = \cO_R$. 
For example, for $R_1 = \{(p,p) \mid p \in Q\}$ and $R_2 = Q \times Q$, 
one can check that $\cO_{\mathsf{single}} = \cO_{R_1}$ and $\cO_{\mathsf{full}} = \cO_{R_{2}}$. 
Therefore, the partition oracles $\cO_{\mathsf{single}}$ and  $\cO_{\mathsf{full}}$ are relational as well. 

We say that $\cD$ is a \emph{relational disambiguation scheme} if $\cD(\cA)$ is relational for every~$\cA$.
Note that if for every $\cA$ we can compute a relation $R_{\cA}$ in polynomial time such that $\cD(\cA) = \cO_{R_\cA}$, then $\cD$ is efficient.
Thus, our strategy in the sequel is to provide relations $R_{\cA}$ that can be computed in polynomial time from $\cA$, leading to the desired disambiguation~schemes.

\paragraph{UFA disambiguation} Let $\cA = (Q, \Sigma, \Delta, I, F)$ be an NFA. 
We say that two states $p, q\in Q$ \emph{share a common future} in $\cA$ iff there exists $v \in \Sigma^*$ such that $\Delta(p, v) \cap F \neq \emptyset \neq \Delta(q,v) \cap F$. 
In other words, there exist two accepting partial runs $\rho_p$ and $\rho_q$ of $\cA$ over the same word $v$ starting from $p$ and $q$, respectively. 
We define the relation $\CF_\cA \subseteq Q \times Q$ where $(p,q) \in \CF_\cA$ iff $p$ and $q$ share a common future. 
By definition, $\CF_\cA$ is symmetric. 
Further, since $\cA$ is trimmed, $\CF_\cA$ is reflexive. 
Then, we denote by $\cO_{\CF}$ the relational partition oracle of $\cA$ defined~by~$\CF_\cA$.

Intuitively, if $p$ and $q$ share $v\in\Sigma^*$ as a common future, then, 
whenever a disambiguation scheme reaches $p$ and $q$ from the initial states (i.e., $p, q \in \Delta(I, u)$ for some $u$), 
$p$ and $q$ must be part of the same state in the resulting automata 
so that it stays unambiguous (otherwise, $u \cdot v$ will have two accepting runs). 
For this reason,  the partition oracle $\cO_{\CF}$ will keep $p$ and $q$ together along with all states in their corresponding connected component. 
Indeed, $\cO_{\CF}$ is a relational partition oracle for producing UFA and can be computed efficiently. 

\begin{restatable}{theorem}{UFAdisambiguation}\label{theo:UFA-disambiguation}
		For every NFA $\cA = (Q, \Sigma, \Delta, I, F)$, the following properties hold: (1) $\cA_{\cO_{\CF}}$ is always unambiguous; (2) if $\cA$ is unambiguous, then $\cA_{\cO_{\CF}}$ and $\cA$ are isomorphic; and (3) the relation $\CF_\cA$ can be computed in time $O(|\cA|^2)$. 
	That is, the disambiguation scheme $\cD_{\CF}$ that assigns the relational partition oracle $\cO_{\CF}$ is an efficient UFA-disambiguation scheme.
\end{restatable} 

The time to compute $\CF_\cA$ matches the classical algorithm for testing if an NFA is unambiguous, and the exponent cannot be lowered, under fine-grained assumptions~\cite{finegraineddisamb}.

Coming back to our example NFA $\cA$ in Figure~\ref{fig:FAexample}(a), one can  check that its common future relation is $\CF_{\cA} =\{(1, 2), (2, 1), (3, 5), (5, 3), (4, 4), \ldots\}$. We show its disambiguation $\cA_{\cO_{\CF}}$ by the oracle $\cO_{\CF}$ in Figure~\ref{fig:FAexample} (c).

While $\cD_{\CF}$ fulfils our desirable algorithmic properties, one can easily devise ``efficient UFA-disambiguation schemes'' by an ad-hoc partition oracle. For instance,~one~may~use~$\cO_{\mathsf{single}}$ when $\cA$ is unambiguous, and $\cO_{\mathsf{full}}$, otherwise (in other words, if $\cA$ is not unambiguous, determinize it). 
So, what makes $\cD_{\CF}$ so special as a disambiguation procedure?
We argue that $\cD_{\CF}$ is minimal in the following sense.
Let $\pi_1, \pi_2$ be two partitions of the same set $A$. We say that $\pi_1$ \emph{refines} $\pi_2$ if for every $P_1 \in \pi_1$ there is a $P_2 \in \pi_2$ such that $P_1 \subseteq P_2$. That is, $\pi_1$ breaks $A$ into smaller pieces than $\pi_2$. 
Let us define $\reach(\cA) = \bigcup_{w \in \Sigma^*} \Delta(I, w)$ as the set of all states reachable from $I$.
For two partition oracles $\cO_1$ and $\cO_2$ of $\cA$, we say that $\cO_1$ \emph{refines} $\cO_2$ over $\cA$ if $\cO_1(I)$ refines $\cO_2(I)$ and, for every $S \in \{\Delta(S',a) \mid S' \in \reach(\cA_{\cO_2}) \wedge a \in \Sigma\}$, $\cO_1(S)$ refines $\cO_2(S)$. 

\begin{restatable}{theorem}{UFAminimality}\label{theo:UFA-minimality}
		Let $\cO$ be any partition oracle of an NFA $\cA$. If
	$\cA_{\cO}$ is unambiguous, then $\cO_{\CF}$ refines $\cO$ over $\cA$. %
\end{restatable}  

In other words, the disambiguation scheme $\cD_{\CF}$ provides the \emph{finest} possible partitions~on state sets $S \subseteq Q$ while $\cA_{\cO}$ is still unambiguous. 
We argue that fineness is a reasonable criterion for minimality. 
Indeed, finer partitions, in a sense, better preserve the original automaton, and the finest partition is $\cO_{\mathsf{single}}$ which leaves the automaton unchanged.

\paragraph{Discussion} The reader may have noted that, although the partition oracle $\cO_{\CF}$ finds the finest possible partition, the determinization $\cA_{\det}$ of $\cA$ in Figure~\ref{fig:FAexample} (b) has less states compared to the disambiguation $\cA_{\cO_{\CF}}$ in Figure~\ref{fig:FAexample} (c) (i.e., one state less).
Therefore, it is essential to discuss that, regarding minimizing the size of the resulting automaton, 
the disambiguation schemes $\disUFA$ and $\disDFA$ (i.e., the determinization procedure) are incomparable. On the one hand, we know that there exists a family $\{\cA^n\}_{n\in \bbN}$ of UFA such that every DFA $\cA$ with $\cL(\cA) = \cL(\cA^n)$ has $|\cA| \in \Omega(2^n)$~\cite{Colcombet15}. This implies that the determinization $\cA^n_{\oracleFull}$ by $\disDFA$ can be of exponential size with respect to $\cA^n_{\cO_{\CF}}$. On the other hand, in Figure~\ref{fig:UFA-DFA-blowup} we show an example of a family of UFA with $\Omega(2^n)$ states whose Rabin-Scott determinization is of size $O(n)$.
Finding an optimal UFA-disambiguation scheme that minimizes the number of explored states is an interesting and relevant open problem that we leave for future work.
\begin{figure}[t]
	\centering
	\begin{tikzpicture}[every label/.style={black!60}, auto, initial text=, initial distance= {3mm}, >=stealth', node distance=0.8cm, 
		mystate/.style={state, initial where=left,inner sep=1.2mm, minimum size=0pt}]
		\node[mystate, initial]  (0) {};
		\node[mystate] (00) at ($(0)+(1.1,1.2)$)  {};
		\node[mystate] (01) at ($(0)+(1.1,-1.2)$) {};
		
		\node[mystate, draw=none,inner sep=1mm] (000) at ($(00)+(1.1,0.5)$) {$\cdots$};
		\node[mystate, draw=none,inner sep=1mm] (001) at ($(00)+(1.1,-0.5)$) {$\cdots$};
		
		\node[mystate, draw=none,inner sep=1mm] (010) at ($(01)+(1.1,0.5)$) {$\cdots$};
		\node[mystate, draw=none,inner sep=1mm] (011) at ($(01)+(1.1,-0.5)$) {$\cdots$};
		
		\node[mystate] (0000) at ($(000)+(1.1,0.2)$) {};
		\node[mystate] (0001) at ($(000)+(1.1,-0.2)$) {};
		
		\node[mystate] (0010) at ($(001)+(1.1,0.2)$) {};
		\node[mystate] (0011) at ($(001)+(1.1,-0.2)$) {};
		
		\node (vsdso) at ($(0011)+(0,-0.4)$) {$\vdots$};
		
		\node[mystate] (0100) at ($(010)+(1.1,0.2)$) {};
		\node[mystate] (0101) at ($(010)+(1.1,-0.2)$) {};
		
		\node[mystate] (0110) at ($(011)+(1.1,0.2)$) {};
		\node[mystate] (0111) at ($(011)+(1.1,-0.2)$) {};
		
		\node[mystate, draw=none,inner sep=1mm] (b000) at ($(0000)+(1.1,-0.2)$) {$\cdots$};
		\node[mystate, draw=none,inner sep=1mm] (b001) at ($(0010)+(1.1,-0.2)$) {$\cdots$};
		
		\node[mystate, draw=none,inner sep=1mm] (b010) at ($(0100)+(1.1,-0.2)$) {$\cdots$};
		\node[mystate, draw=none,inner sep=1mm] (b011) at ($(0110)+(1.1,-0.2)$) {$\cdots$};
		
		\node[mystate] (b00) at ($(b000)+(1.1,-0.5)$)  {};
		\node[mystate] (b01) at ($(b010)+(1.1,-0.5)$) {};
		
		\node[mystate, accepting] (b0) at ($(b00)+(1.1,-1)$)  {};

		\draw[->] (0) edge node {${ a}$} (00);
		\draw[->] (0) edge node[swap] {${ a}$} (01);
		
		\draw[->] (00) edge node {${ a}$} (000.west);
		\draw[->] (00) edge node[swap] {${ a}$} (001.west);
		
		\draw[->] (01) edge node {${ a}$} (010.west);
		\draw[->] (01) edge node[swap] {${ a}$} (011.west);
		
		\draw[->] (000) edge node[above] {${ a}$} (0000);
		\draw[->] (000) edge node[below,swap] {${ a}$} (0001);
		\draw[->] (001) edge node[above] {${ a}$} (0010);
		\draw[->] (001) edge node[below,swap] {${ a}$} (0011);
		
		\draw[->] (010) edge node[above] {${ a}$} (0100);
		\draw[->] (010) edge node[below,swap] {${ a}$} (0101);
		\draw[->] (011) edge node[above] {${ a}$} (0110);
		\draw[->] (011) edge node[below,swap] {${ a}$} (0111);
		
		\draw[->] (0000) edge node[above] {${ b}$} (b000);
		\draw[->] (0001) edge node[below,swap] {${ c}$} (b000);
		\draw[->] (0010) edge node[above] {${ b}$} (b001);
		\draw[->] (0011) edge node[below,swap] {${ c}$} (b001);
		
		\draw[->] (0100) edge node[above] {${ b}$} (b010);
		\draw[->] (0101) edge node[below,swap] {${ c}$} (b010);
		\draw[->] (0110) edge node[above] {${ b}$} (b011);
		\draw[->] (0111) edge node[below,swap] {${ c}$} (b011);
		
		\draw[<-] (b00) edge node[swap] {${ b}$} (b000.east);
		\draw[<-] (b00) edge node {${ c}$} (b001.east);
		
		\draw[<-] (b01) edge node[swap] {${ b}$} (b010.east);
		\draw[<-] (b01) edge node {${ c}$} (b011.east);
		
		\draw[->] (b00) edge node {${ b}$} (b0);
		\draw[->] (b01) edge node[swap] {${ c}$} (b0);
		
		\draw[decorate,decoration={brace, amplitude=5}] ($(0000)+(-3.3,0.4)$) -- node[above=2mm] {$n$ levels} ($(0000)+(0,0.4)$);
	\end{tikzpicture}
	\caption{Example of an UFA with $\Omega(2^n)$ states whose determinization is of size $O(n)$.}
	\label{fig:UFA-DFA-blowup}
\end{figure}
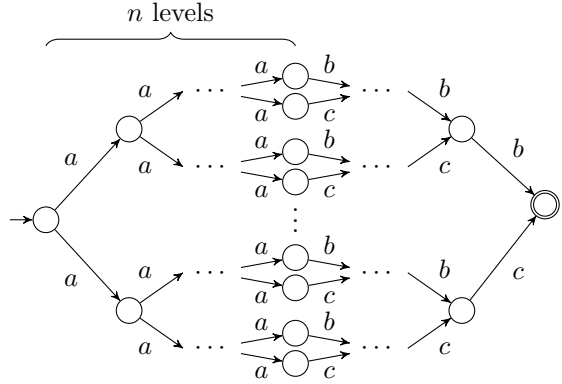

\paragraph{finFA and polyFA disambiguation} Similar to UFA, we provide relations to define partition oracles for the finFA and polyFA disambiguation. %
Let $\cA = (Q, \Sigma, \Delta, I, F)$ be~an~NFA.
\begin{itemize}
	\item \emph{Infinite common future (ICF)}: We say that a pair of states $(p, q)$ share an \emph{infinite common future} in $\cA$ iff there exists $r,s \in Q$ and $u, v \in \Sigma^*$ such that $p, q \in \Delta(r, u)$, $r \in \Delta(p, v)$, $s \in \Delta(q, v)$, and $s \in \Delta(s, uv)$. We define the relation $\ICF_{\cA} \subseteq Q \times Q$ where $(p, q) \in \ICF_{\cA}$ if $(p, q)$ or $(q, p)$ share an infinite common future.
	
	\item \emph{Exponentially infinite common future (ECF)}: We say that two states $p, q \in Q$ share an \emph{exponentially infinite common future} in $\cA$ iff there exists $r \in Q$ and $u, v \in \Sigma^*$ such that $r \in \Delta(p, v)$, $r \in \Delta(q, v)$, and $p,q \in \Delta(r, u)$. 
	Similarly, we define the relation $\ECF_{\cA} \subseteq Q \times Q$ where $(p, q) \in \ECF_{\cA}$ if $p, q$ share a exponentially infinite common future. 
\end{itemize}
We illustrate the infinite and exponentially~infinite common future conditions in Figure \ref{fig:ICF-ECF} (a) and (b), respectively. Note that ICF and ECF are stronger conditions than regular common future and, indeed, ECF implies ICF. 
Further, $\ICF_{\cA}$ and $\ECF_{\cA}$ are reflexive and symmetric. 
Thus, they induce relational partition oracles that we denote by $\cO_{\ICF}$ and $\cO_{\ECF}$, respectively.

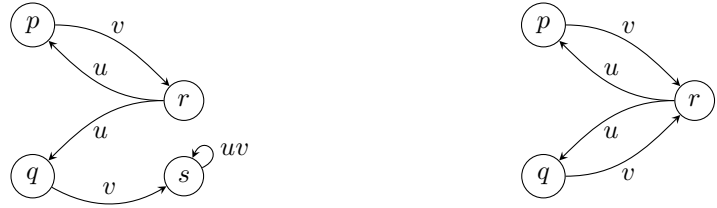
\begin{figure}[t]
	\centering
\begin{tikzpicture}[->,>=stealth,roundnode/.style={circle,draw,inner sep=3pt},node distance=0.8cm]
	\node[roundnode] (p) at (0, 3) {$p$};
	\node[roundnode] (q) at (0, 1) {$q$};
	\node[roundnode] (r) at (2, 2) {$r$};
	\node[roundnode] (s) at (2, 1) {$s$};

	\node[] at (1, 0) {(a) Infinite common future};

	\draw[out=0,in=135] (p) to[above] node {$v$} (r);
	\draw[out=180,in=-45] (r) to[above] node {$u$} (p);
	\draw[out=180,in=45] (r) to[below] node {$u$} (q);
	\draw[bend right] (q) to[above] node {$v$} (s);
	\draw[out=25,in=65,looseness=5] (s) to[right] node {$uv$} (s);
	
\end{tikzpicture}
\hspace{3em}
\begin{tikzpicture}[->,>=stealth,roundnode/.style={circle,draw,inner sep=3pt}]

	\node[roundnode] (p) at (0, 3) {$p$};
	\node[roundnode] (q) at (0, 1) {$q$};
	\node[roundnode] (r) at (2, 2) {$r$};

	\node[] at (1, 0) {(a) Exponentially infinite common future};

	\draw[out=0,in=135] (p) to[above] node {$v$} (r);
	\draw[out=180,in=-45] (r) to[above] node {$u$} (p);
	\draw[out=180,in=45] (r) to[below] node {$u$} (q);
	\draw[out=0,in=-135] (q) to[below] node {$v$} (r);
	
\end{tikzpicture}
	\caption{Graphical representation of infinite and exponentially infinite common future conditions.}
	\label{fig:ICF-ECF}
\end{figure}

Similar to the common future condition, ICF and ECF identify pairs of states in $\cA$ that can lead to infinitely and exponentially ambiguous behaviour, respectively. Indeed, they lead to efficient finFA- and polyFA-disambiguation schemes as expected.

\begin{restatable}{theorem}{finFApolyFAdisambiguation}\label{theo:finFA-polyFA-disambiguation}
		For every NFA $\cA = (Q, \Sigma, \Delta, I, F)$, the following properties hold: (1) $\cA_{\cO_{\ICF}}$ and $\cA_{\cO_{\ECF}}$ are always finitely ambiguous and polynomially ambiguous, respectively; (2) if $\cA$ is finitely ambiguous, then $\cA_{\cO_{\ICF}}$ is isomorphic to $\cA$; if $\cA$ is polynomially ambiguous then $\cA_{\cO_{\ECF}}$ isomorphic to $\cA$; and (3) the relations $\ICF_\cA$ and $\ECF_\cA$ can be computed in time $O(|\cA|^5)$ and $O(|\cA|^3)$, respectively. 
	In other words, the disambiguation schemes $\cD_{\ICF}$ and $\cD_{\ECF}$ that assign the relational partition oracles $\cO_{\ICF}$ and $\cO_{\ECF}$ for each NFA are efficient finFA- and polyFA-disambiguation schemes, respectively.
\end{restatable} 

An example of finFA and polyFA disambiguation of an expFA $\cA'$ is shown in Figure~\ref{fig:conditionExamples}.
The infinite and exponentially infinite common future relations on the NFA $\cA'$ are $\ICF_{\cA'} = \{(2, 5), (5, 2), (3, 5), (5, 3), (1,1), \ldots\}$ and $\ECF_{\cA'} = \{(3, 5), (5, 3), (1,1), \ldots\}$, respectively. 
Note that $\cA'_{\cO_{\ICF}}$ in Figure~\ref{fig:conditionExamples} (b) and $\cA'_{\cO_{\ECF}}$ in Figure~\ref{fig:conditionExamples} (c) are exactly finFA and polyFA, respectively, achieving partial disambiguation depending on the chosen relation.

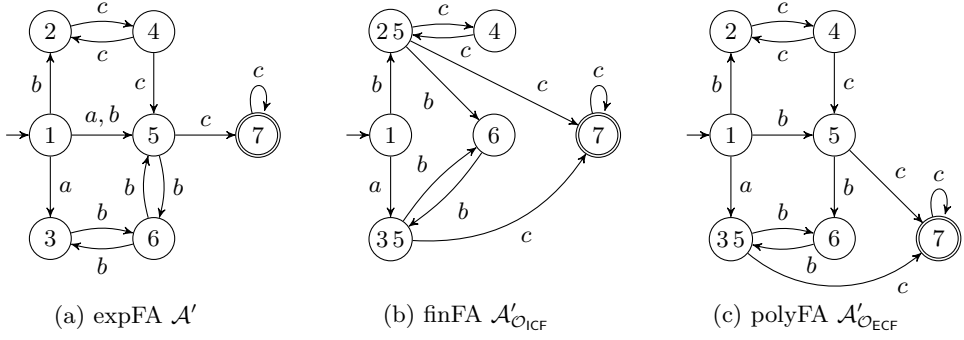
\begin{figure}[t]
	\centering
	 \input{./figures/kDisambiguation/example/original.tex}%
	    \caption{
	    (a) An example of an exponentially ambiguous NFA $\cA'$. (b) The finFA disambiguation $\cA'_{\cO_{\ICF}}$ of $\cA'$. (c) The polyFA disambiguation $\cA'_{\cO_{\ECF}}$ of $\cA'$.
	    }
	\label{fig:conditionExamples}
\end{figure}

Similar to the UFA case, one would like to show that 
$\cO_{\ICF}$ and $\cO_{\ECF}$ are minimal in the sense that they refine any other partition oracle that produces finitely ambiguous or polynomially ambiguous NFA, respectively. Unfortunately, this is not the case for both, since one can find partition oracles that produce the required level of ambiguity, but they are not refined by $\cO_{\ICF}$ or $\cO_{\ECF}$.%

\begin{example}
	Consider the NFA $\cA$ and its polynomial disambiguation $\cA_{\cO_{\ECF}}$ in Figure~\ref{fig:edaExamples}. Its exponentially infinite common future relation is $\ECF_{\cA} = \{(p, q), (q, p), (p, z), (z, p)\}$.
	There exist a relation $R=\{(p, z), (q, z), (z, p), (z, q)\}$ such that $\cO_{R}(\{p, q\})$ refines $\cO_{\ECF}(\{p, q\})$, while $\cA_{\cO_{R}}$ is polynomially ambiguous.
	In $\cA_{\cO_{R}}$, states $p$ and $q$ are only merged in the set of states that contains $z$, effectively preventing exponential degree of ambiguity.
	
	\begin{figure}
		\vspace{5em}
		\hspace{-0.5em}
		\subfloat[NFA $\cA$]{%
			\input{./figures/EDAexample/original.tex}%
		}
		\hspace{0.5em}
		\begin{minipage}{25mm}
			\vspace{-10em}
			\subfloat[The polynomial disambiguation $\cA_{\cO_{\ECF}}$ of $\cA$]{%
				\input{./figures/EDAexample/dis.tex}%
			}
			\vspace{0.5em}
			\subfloat[NFA $\cA_{\cO_{R}}$]{%
				\input{./figures/EDAexample/alternative.tex}%
			}
		\end{minipage}
		
		\caption{Counterexample for the minimality of $\ECF$.}
		\label{fig:edaExamples}
	\end{figure}
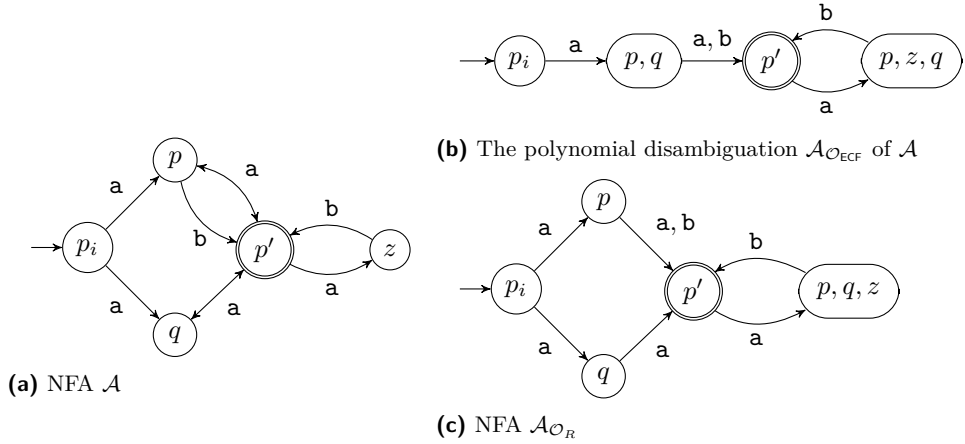
	
	To determine that states $p$ and $q$ should not be merged, the procedure needs knowledge of their future behavior on $\cA_{\cO_{R}}$  to verify they are not part of an (EDA) condition.
	Consequently, no disambiguation procedure that cannot be computed on-the-fly can refine the states formed by oracle $\ECF$ and define a polynomially ambiguous automaton.
	Note that in the example, the relation $\ICF_{\cA}$ is equivalent to $\ECF_{\cA}$, then the same NFA serves as a counter example for the minimality of $\ICF$.
\end{example}

Nevertheless, we can show that among all {\em relational} partition oracles, both $\cO_{\ICF}$ and $\cO_{\ECF}$ are minimal in terms of refinement; namely, $\ICF_{\cA}$ and $\ECF_{\cA}$ are the minimal relations for which the disambiguation schemes can reach the required level of ambiguity.

\begin{restatable}{theorem}{minimalityICFECF}\label{theo:minimalityICF-ECF}
	Let $\cA$ be an NFA and $R$ be a reflexive and symmetric relation between states of~$\cA$. (1) If
	$\cA_{\cO_R}$ is finitely ambiguous and $R \subseteq \ICF_{\cA}$, then $\cO_{\ICF}$ refines $\cO_R$ over $\cA$; and (2)~if
	$\cA_{\cO_R}$ is polynomially ambiguous and $R \subseteq \ECF_{\cA}$, then $\cO_{\ECF}$ refines $\cO_R$ over $\cA$.
\end{restatable} 

We end this section by recalling that the notion of a common future between states was previously used in~\cite{mohri13}, and the disambiguation scheme $\cD_{\CF}$ was introduced in~\cite{WeberEconomy}, but with a different presentation.%
The novelty of Theorem~\ref{theo:UFA-disambiguation} relies on understanding this result in the context of a larger disambiguation framework (the one introduced in this paper) and studying its minimality (e.g., Theorem~\ref{theo:UFA-minimality}). Moreover, to the best of our knowledge, disambiguation procedures for finding finFA and polyFA have not been studied before.

%% file: figures/kDisambiguation/example/original.tex
\small
\begin{tikzpicture}[every label/.style={black!60}, 
	initial text=, 
	initial distance= {3mm},
	->,>=stealth', 
	auto,
	node distance=0.8cm,
	mystate/.style={state, inner sep=3pt, minimum size=0mm},
	longstate/.style={rectangle,inner sep=8pt,draw,rounded corners=11pt}]
  \node[mystate, initial]          (p_i)                {$1$};
  \node[mystate]          (p) [below=of p_i] {$3$};
  \node[mystate]          (r) [right=of p] {$6$};
  \node[mystate]          (g) [above=of p_i] {$2$};
    \node[mystate]          (q) [right=of p_i] {$5$};
  \node[mystate]          (r_2) [above=of q] {$4$};
  \node[mystate, accepting]          (s) [right=of q] {$7$};

  \path[->] (p_i) edge      node   {$b$} (g)
            (p_i) edge      node   {$a$} (p)
            (p_i) edge      node  {$a, b$} (q)
            (r_2) edge node[swap]  {$c$} (q)
            (q) edge node  {$c$} (s)
            (s) edge [loop above] node  {$c$} (s)
            (p) edge[bend left=15] node  {$b$} (r)
            (r) edge[bend left=15] node {$b$} (p)
            (r) edge[bend left=15] node  {$b$} (q)
            (q) edge[bend left=15] node {$b$} (r)
            (g) edge[bend left=15] node  {$c$} (r_2)
            (r_2) edge[bend left=15] node {$c$} (g);
            
          \node at (1, -2.4) {(a) expFA $\cA'$};
            
    \begin{scope}[xshift=9cm]
    	  \node[mystate, initial]          (p_i)                {$1$};
    	\node[mystate]          (g) [above=of p_i] {$2$};
    	\node[mystate]          (q) [right=of p_i] {$5$};
    	\node[mystate, inner sep=1.5pt]          (pq) [below=of p_i] {$3\,5$};
    	\node[mystate]          (r) [below=of q] {$6$};
    	\node[mystate]          (r_2) [right=of g] {$4$};
    	\node[mystate, accepting]          (s) [right=of r] {$7$};
    	
    	\path[->] (p_i) edge      node   {$b$} (g)
    	(p_i) edge      node   {$b$} (q)
    	(p_i) edge      node {$a$} (pq)
    	(r_2) edge node {$c$} (q)
    	(q) edge node  {$c$} (s)
    	(pq) edge [bend right=40,pos=0.8] node[swap] {$c$} (s)
    	(s) edge [loop above] node  {$c$} (s)
    	(q) edge node  {$b$} (r)
    	(g) edge[bend left=15] node  {$c$} (r_2)
    	(r_2) edge[bend left=15] node {$c$} (g)
    	(pq) edge[bend left=15] node  {$b$} (r)
    	(r) edge[bend left=15, pos=0.25] node {$b$} (pq);
    	
    	    \node at (1, -2.4) {(c) polyFA $\cA'_{\cO_{\ECF}}$};
    	
    \end{scope}
    
    \begin{scope}[xshift=4.5cm]
    	  \node[mystate, initial]          (p_i)                {$1$};
    	\node[mystate, inner sep=1.5pt]          (gq) [above=of p_i] {$2\, 5$};
    	\node[mystate, inner sep=1.5pt]          (pq) [below=of p_i] {$3\, 5$};
    	\node[mystate]          (r) [right=of p_i] {$6$};
    	\node[mystate]          (r_2) [right=of gq] {$4$};
    	\node[mystate, accepting]          (s) [right=of r] {$7$};
    	
    	\path[->] (p_i) edge      node    {$b$} (gq)
    	(p_i) edge      node[swap] {$a$} (pq)
    	(gq) edge[pos=0.7]  node  {$c$} (s)
    	(pq) edge[bend right] node[swap] {$c$} (s)
    	(s) edge [loop above] node {$c$} (s)
    	(gq) edge node[swap]  {$b$} (r)
    	(gq) edge[bend left=10] node  {$c$} (r_2)
    	(r_2) edge[bend left=10,pos=0.35] node {$c$} (gq)
    	(pq) edge[bend left=10] node  {$b$} (r)
    	(r) edge[bend left=10] node {$b$} (pq);
    	
    	\node at (1, -2.4) {(b) finFA $\cA'_{\cO_{\ICF}}$};
  
    \end{scope}
\end{tikzpicture}

%% file: figures/EDAexample/original.tex
\begin{tikzpicture}[scale=0.9,every label/.style={black!60}, initial text=, ->,>=stealth', node distance=1cm,
		mystate/.style={state, inner sep=3pt, minimum size=0mm}]
  \node[mystate, initial]          (p_i)                {$p_i$};
  \node[mystate]          (p) [above right=of p_i] {$p$};
  \node[mystate]          (q) [below right=of p_i] {$q$};
  \node[mystate, accepting]          (p') [below right=of p] {$p'$};
  \node[mystate]          (z) [right=of p'] {$z$};

  \path[<->] (p) edge [bend left]     node [above right]   {${\tt a}$} (p')
            (q) edge      node [below right]   {${\tt a}$} (p');
  
  \path[->] (p_i) edge      node [above left]   {${\tt a}$} (p)
            (p_i) edge      node [below left]   {${\tt a}$} (q)
            (p') edge [bend right] node [below] {${\tt a}$} (z)
            (z) edge [bend right] node [above] {${\tt b}$} (p')
            (p) edge [bend right] node [below] {${\tt b}$} (p');
\end{tikzpicture}

%% file: figures/EDAexample/dis.tex
\begin{tikzpicture}[every label/.style={black!60}, initial text=, ->,>=stealth',node distance=0.8cm,state/.style={circle,draw},longstate/.style={rectangle,inner sep=7pt,draw,rounded corners=11pt},
		mystate/.style={state, inner sep=3pt, minimum size=0mm}]
  \node[mystate, initial]          (p_i)                {$p_i$};
  \node[longstate]          (pq) [right=of p_i] {$p, q$};
  \node[mystate, accepting]          (p') [right=of pq] {$p'$};
  \node[longstate]          (pqz) [right=of p'] {$p, z, q$};

  \path[->] (p_i) edge      node [above]   {${\tt a}$} (pq)
            (pq) edge      node [above]   {${\tt a,b}$} (p')
            (p') edge [shorten >=-3.2pt,bend right] node [below] {${\tt a}$} (pqz.-156)
            (pqz.156) edge [shorten <=-3.2pt,bend right] node [above] {${\tt b}$} (p');
\end{tikzpicture}

%% file: figures/EDAexample/alternative.tex
\begin{tikzpicture}[scale=0.9,every label/.style={black!60}, initial text=, ->,>=stealth', node distance=1cm,state/.style={circle,draw},
longstate/.style={rectangle,inner sep=7pt,draw,rounded corners=11pt},
mystate/.style={state, inner sep=3pt, minimum size=0mm}
]
  \node[mystate, initial]          (p_i)                {$p_i$};
  \node[mystate]          (p) [above right=of p_i] {$p$};
  \node[mystate]          (q) [below right=of p_i] {$q$};
  \node[mystate, accepting]          (p') [below right=of p] {$p'$};
  \node[longstate]          (z) [right=of p'] {$p, q, z$};

  \path[->] (p_i) edge      node [above left]   {${\tt a}$} (p)
            (p_i) edge      node [below left]   {${\tt a}$} (q)
            (p) edge      node [above right]   {${\tt a, b}$} (p')
            (q) edge      node [below right]   {${\tt a}$} (p')
            (p') edge [shorten >=-3.2pt,bend right] node [below] {${\tt a}$} (z.-156)
            (z.156) edge [shorten <=-3.2pt,bend right] node [above] {${\tt b}$} (p');
\end{tikzpicture}

%% file: sections/weighted.tex
Weighted finite automata are more complex than previous automata models and have numerous use cases~\cite{droste2009handbook}.
Moreover, they do not always admit a deterministic 
equivalent automata. 
Therefore, generalizing our algorithmic framework for disambiguation to this model would be useful for practical applications.
In this section, we extend the algorithmic framework for disambiguation to the model of weighted automata. 
We start by recalling some algebraic notions and the model of weighted automata. We then present the framework and the disambiguation results over weighted~automata. 

\paragraph{Semirings} A \emph{monoid} is a triple $(\bbM, \oplus, \overline{0})$ where $\bbM$ is a non-empty set, $\oplus$ is an associative binary operation, and $\zero \in \bbM$ is the identity of $\oplus$ over $\bbM$ (i.e., $\zero \oplus m = m \oplus \zero = m$). We say that $(\bbM, \oplus, \overline{0})$ is \emph{commutative} if, additionally, $\oplus$ is commutative. A \emph{semiring} is a tuple $(\bbS, \oplus, \odot, \zero, \one)$ such that $(\bbS, \oplus, \zero)$ is a commutative monoid, $(\bbS, \odot, \one)$ is a  monoid, $\odot$ distributes over $\oplus$ 
and $\zero$ annihilates with $\odot$. %
For simplicity, we will refer to such a semiring by its underlying set $\bbS$. Examples of semirings are the \emph{boolean semiring} $(\{0,1\}, \vee, \wedge, 0, 1)$, the \emph{natural numbers} $(\bbN, +, \cdot, 0, 1)$, the \emph{tropical semiring} $(\bbZ \cup \{\infty\}, \min, +, \infty, 0)$, and the \emph{artic semiring} $(\bbZ \cup \{-\infty\}, \max, +, -\infty, 0)$, among others. %

Given non-empty sets $A$ and $B$, we denote by $\bbS^{A\times B}$ the set of all \emph{matrices} indexed by $A \times B$ over $\bbS$. 
Further, we denote by $\bbS^{A}$ the set of (column) \emph{vectors} (i.e., $\bbS^{A} = \bbS^{A \times 1}$). 
For $M \in \bbS^{A\times B}$ and $V \in \bbS^A$, we write $M[a,b]$ and $V[a]$ for the entries at positions $(a,b) \in A \times B$ and $a \in A$, respectively.
Given two matrices $M \in \bbS^{A \times B}$ and $N \in \bbS^{B \times C}$, their product $P = M\odot N \in \mathbb{S}^{A \times C}$ is defined by $
P[a,c] = \bigoplus_{b \in B} M[a,b] \odot N[b,c]
$
for all $a \in A$ and $c \in C$. 
We write $M^t \in \bbS^{B \times A}$ for the transpose of $M \in \bbS^{A \times B}$ where $V^t$ corresponds to a row vector when $V \in \bbS^A$.
For a value $s\in \bbS$ and vector $V\in\bbS^A$, we define $s\odot V \in \bbS^A$ ($V\odot s \in \bbS^A$) by $(s\odot V)[a] = s\odot V[a]$ ($(V\odot s)[a] = V[a]\odot s$, resp.) for each $a\in A$.
Given a vector $V \in \bbS^{A}$, we define its support as $\supp(V) = \set{a \in A \mid V[a] \neq \zero}$. Given a value $s \in \bbS$, we denote by $s^A$ the vector in $\bbS^A$ such that $s^A[a] = s$ for every $a \in A$. 

\paragraph{Weighted automata} A weighted finite automaton~\cite{schutzenberger1961definition} (WFA) over a semiring $\bbS$ is a tuple:
\[
\cW \ := \ (Q, \Sigma, \Delta, I, F) \tag{$\dagger$}
\] 
where $Q$ is a finite set of states, $\Sigma$ is a finite alphabet, $\Delta \subseteq Q \times \Sigma \times \mathbb{S} \times Q$ is a finite transition relation, and $I: Q \rightarrow \bbS$, $F: Q \rightarrow \bbS$ are the initial and final weighted functions, respectively. We assume that for every $p,q \in Q$ and $a \in \Sigma$ there exists at most one transition $(p,a,s,q) \in \Delta$ for some $s \in \bbS$ (i.e., there cannot be multiple transitions between states for the same letter). We say that $q \in Q$ is an initial (final) state if $I(q) \neq \zero$ ($F(q) \neq \zero$, resp.). 
A \emph{partial run} $\rho$ of $\cW$ over a word $w = a_1 \ldots a_n \in \Sigma^*$ is a sequence:
\[
\rho \ := \ p_0 \, \xrightarrow{a_1/s_1} \, p_1 \, \xrightarrow{a_2/s_2} \ \cdots \ \xrightarrow{a_n/s_n} \, p_n \tag{$\ddagger$}
\]
where $(p_{i-1}, a_i, s_i, p_i) \in \Delta$ for each $i \in [n]$. 
We say $\rho$ is a \emph{run} if $p_0$ is an initial~state, and it is \emph{accepting} if additionally $p_n$ is a final state. 
We denote by~$\runs_{\cW}(w)$~the~set~of~all accepting runs of $\cW$ over $w$. For a partial run $\rho$ like ($\ddagger$) we define $\omega(\rho) := I(p_0) \odot s_1 \odot \cdots \odot s_n \odot F(p_n)$ as the {\em weight of $\rho$}. 
Each weighted automaton $\cW$ defines a function $\sem{\cW}: \Sigma^* \rightarrow \bbS$ such that 
$
\sem{\cW}(w) := \bigoplus_{\rho \in \, \runs_\cW(w)} \omega(\rho)
$
for every $w \in \Sigma^*$. 
Similar to NFA, sometimes we also interpret $\Delta$ as a function $\Delta: Q \times \Sigma \rightarrow 2^{\bbS \times Q}$ such that $\Delta(p, a) = \{(s,q) \mid (p,a,s,q) \in \Delta\}$.

In the sequel, it will be useful to consider the \emph{matrix representation} of a weighted automaton $\cW$. For every $a \in \Sigma$, the transition relation naturally defines the matrix $\Delta_a \in \bbS^{Q \times Q}$ such that $\Delta_a[p,q] = s$ if $(p,a,s,q) \in \Delta$ and $\zero$, otherwise. We extend this matrix from $\Sigma$ to $\Sigma^*$ as $\Delta_w = \Delta_{a_1} \odot \cdots \odot \Delta_{a_n}$ for every $w = a_1 \ldots a_n \in \Sigma^*$. Further, we can represent the initial and final functions $I$ and $F$ as vectors $\vI, \vF \in \bbS^{Q}$ such that $\vI[q] = I(q)$ and $\vF[q] = F(q)$ for every $q \in Q$. Then, the function $\sem{\cW}$ can be equivalently defined as $\sem{\cW}(w) = \vI^t \odot \Delta_w \odot \vF$ for every~$w \in \Sigma^*$. 
We will also use the notation $\vec{\Delta}(V, a) := (V^t\odot \Delta_a)^t$ which intuitively is the aggregate vector reached by starting at $V$ and moving through $a$.

Given a WFA $\cW$ like ($\dagger$), we define its \emph{underlying NFA} $\cA_\cW = (Q, \Sigma, \DeltaFA, \IFA, \FFA)$ such that $\IFA = \{p \in Q \mid I(p) \neq \zero\}$, $\FFA = \{p \in Q \mid F(p) \neq \zero\}$, and $\DeltaFA = \{(p,a,q) \mid (p,a,s,q) \in \Delta\}$.
Similar to NFA, we assume that all WFA considered in this paper are \emph{trimmed}, namely, every state in its underlying NFA is reachable and co-reachable simultaneously. Further, we say that $\cW$ is \emph{deterministic} (DWFA), \emph{unambiguous} (UWFA), \emph{finitely ambiguous} (finWFA), or \emph{polynomially ambiguous} (polyWFA) if the underlying NFA $\cA_{\cW}$ is deterministic, unambiguous, finitely ambiguous, or polynomially ambiguous, respectively. On some semirings, these classes define a strict \emph{hierarchy of functions}~\cite{KlimannLMP04,ChattopadhyayMM21}, namely, for every pair of classes $\cC_1, \cC_2 \in \{\text{DWFA}, \text{UWFA}, \text{finWFA}, \text{polyWFA}, \text{WFA}\}$ with $\cC_1 \subsetneq \cC_2$ there exists $\cW \in \cC_2$ such that $\sem{\cW} \neq \sem{\cW'}$ for every $\cW' \in \cC_1$.

\paragraph{The weighted disambiguation framework} 
The new challenge in defining a disambiguation algorithm for weighted 
automata lies in handling weights during this process. 
Mohri's determinization algorithm~\cite{mohriTransducer} does this by adding a residual 
weight on subsets. 
Kirsten~and Mäurer~\cite{KirstenOnDeterminizationWeighted} later generalized this 
approach by introducing the weight factorization. 
We use these ideas to generalize the algorithmic framework of previous sections to weighted automata. 

Fix a semiring $\bbS$ and a set $Q$. A \emph{weight factorization}~\cite{KirstenOnDeterminizationWeighted} over $\bbS^Q$ is  a pair of functions $(\fac, \res)$, called the \emph{residual} and \emph{factor} functions, respectively, with $\fac: \bbS^Q \setminus \set{\zero^Q} \mapsto \bbS$ and $\res: \bbS^Q \setminus \set{\zero^Q} \mapsto \bbS^Q$ such that for every vector $V \in \bbS^Q \setminus \set{\zero^Q}$ it holds that $
V = \fac(V) \odot \res(V)$. 
In other words, $(\fac, \res)$ defines a strategy to factorize a vector $V$ into a single factor common to all components ($\fac(V)$) and a vector ($\res(V)$).

Let $\cW$ be an WFA like ($\dagger$) over $\bbS$. We define a \emph{partition-factorization oracle} $\cF$ for $\cW$ (called PF-oracle for short) as a triple $\cF = (\Pi, \fac, \res)$ where $(\fac, \res)$ is a weight factorization over $\bbS^Q$ and $\Pi: \bbS^Q \mapsto 2^{\bbS^Q \setminus \{\zero^Q\}}$ is a \emph{weight partition function} such that $\Pi(V)$ is a finite set and $\bigoplus_{V' \in \Pi(V)} V' = V$ for every $V \in \bbS^Q$. 
That is, $\Pi(V)$ distributes the values of vector $V$ over a set of vectors that has to $\oplus$-aggregate back to $V$. 
Indeed, one can see $\Pi$ as a generalization of the partition oracle $\cO$ from the boolean semiring to any semiring $\bbS$.
Given a PF-oracle $\cF = (\Pi, \fac, \res)$ for $\cW$, we define the (infinite) WFA:
\[
\cW_{\cF}^\infty \ := \ (\bbS^Q, \Sigma, \Delta_{\cF}, I_\cF, F_{\cF})
\]  
such that $\supp(\vec{I}_\cF) = \{\res(V) \mid V \in \Pi(\vI)\}$ and $I_{\cF}(\res(V)) = \fac(V)$ for all $V \in \Pi(\vI)$; $F_\cF(V) = V^t \odot \vF$ for every $V \in \bbS^Q$; and:
\[
\Delta_{\cF} \ := \  \big\{\, (V, a, \fac(V'), \res(V')) \, \mid \, V' \in \Pi(\vec{\Delta}(V, a)) \, \big\}.
\]
Note that $\vec{I}_\cF$ has finite support and, for every $V \in \bbS^Q$, $\Delta_{\cF}(V, a)$ is finite. 
So, although $\cW_{\cF}^\infty$ has an infinite number of states when $\bbS$ is infinite, there is always a finite number of initial states and every state/letter has a finite number of next configurations. We define $\cW_{\cF}$ as the \emph{trimmed version of $\cW_{\cF}^\infty$}. Notice that $\cW_{\cF}$ could lead to a finite WFA but not always.

\begin{restatable}{proposition}{wfaFrameworkEquivalence}\label{prop:wfaFrameworkEquivalence}
	For every PF-oracle $\cF$ of $\cW$, if $\cW_{\cF}$ is finite, then $\sem{\cW} = \sem{\cW_{\cF}}$.
\end{restatable} 

For a class $\cC$ of WFA, the notion of a \emph{$\cC$-disambiguation scheme $\cD$} naturally extends from NFA to WFA where $\cD(\cW)$ must output a PF-oracle for a WFA $\cW$. The notion of an \emph{efficient} disambiguation scheme $\cD$ also extends to WFA where now, for every vector $V$, $\Pi(V)$, $\fac(V)$, and $\res(V)$ must be computed in polynomial time over $|\cW|$ whenever~$\cD(\cW) = (\Pi, \fac, \res)$.

One can check that the construction $\cW_{\cF}$ is a generalization of Mohri's determinization algorithm~\cite{mohriTransducer}. 
Specifically, for the tropical semiring $\bbT = (\bbZ \cup \{\infty\}, \min, +, \infty, 0)$ and a WFA $\cW$ like ($\dagger$) over $\bbT$, consider the PF-oracle $\cFMohri = (\PiId, \facMohri, \resMohri)$ where $\PiId(V) =~\{V\}$, $\facMohri(V) = min_{q \in Q} V[q]$ and $\resMohri(V) = V - \facMohri(V)$ for every $V \in \bbT^Q$. Then $\cW_{\cFMohri}$ indeed corresponds to Mohri's determinization construction. Similarly, consider now any semiring $\bbS$ and a PF-oracle $\cFK = (\PiId, \fac, \res)$ for any weight factorization $(\fac, \res)$ over $\bbS^Q$. And now $\cW_{\cFK}$ is equivalent to Kirsten and Mäurer's determinization construction~\cite{KirstenOnDeterminizationWeighted}. 
What our approach adds is that we can consider different weight partition functions $\Pi$, which could lead to different levels of ambiguity.%

\paragraph{Weighted disambiguation}
Within this framework, we can combine various factorizations, partitions, and semirings to define disambiguation algorithms for WFA. However, the effectiveness of these algorithms depends on whether $\cW_{\cF}$ is finite for all input WFA~$\cW$. Furthermore, we aim to ensure the desired properties of idempotent, on-the-fly, and efficiency discussed in the previous sections. We address this by introducing a specific class of functions $\Pi$ that, combined with the oracles from Section~\ref{sec:disambiguation-conditions}, will guarantee the desired level of ambiguity. %

Given a WFA $\cW$ like ($\dagger$) over a semiring $\bbS$ and 
a PF-oracle $\cF = (\Pi, \fac, \res)$, we say that $\Pi$ has \emph{disjoint vector support} iff $\supp(V_1) \cap \supp(V_2) = \emptyset$ for every pair $V_1, V_2 \in \Pi(V)$ with $V_1 \neq V_2$ and $V \in \bbS^Q$. 
Further, we say that $(\fac, \res)$ has \emph{identity factorization} if $\fac(s \odot V) = s$ and $\res(s \odot V) = V$ whenever $V$ is a zero-one vector (i.e., $V \in \{\zero, \one\}^Q$). For example, the PF-oracle
 $\cFMohri = (\PiId, \facMohri, \resMohri)$ of Mohri's determinization satisfies both~properties.

As we will see, the identity factorization of $(\fac, \res)$ will ensure the idempotent property of disambiguation algorithms; whereas the disjoint vector support of $\Pi$ will allow to extend the partition oracles found for NFA to WFA. Specifically, let $\cO$ be a partition oracle for the underlying NFA $\cA_{\cW}$ of a WFA $\cW$. Then $\cO$ naturally defines a weight partition function $\Pi_{\cO}$ over $\bbS^Q$ given by $\Pi_{\cO}(V) = \set{ V_S \mid S \in \cO(\supp(V))}$ where $V_S \in \bbS^Q$ is defined as $V_S[q] = V[q]$ if $q \in S$, and $V_S[q] = \zero$ otherwise. 
Note that $\Pi_{\cO}$ has disjoint vector support.

\begin{restatable}{theorem}{WFAdisambiguation}\label{theo:WFA-disambiguation}
	Let $\cW$ be a WFA over $\mathbb{S}$ and $(\fac, \res)$ be a weight factorization. For a PF-oracle $\cF$ equal to $(\Pi_{\cO_{\CF}}, \fac, \res)$, $(\Pi_{\cO_{\ICF}}, \fac, \res)$, or $(\Pi_{\cO_{\ECF}}, \fac, \res)$, 
	if $\cW_{\cF}$ is finite, then $\cW_{\cF}$ is unambiguous, finitely ambiguous, or polynomially ambiguous, respectively. Furthermore, if $\cW$ is unambiguous, finitely ambiguous, or polynomially ambiguous and $(\fac,\res)$ is an identity factorization, then $\cW_{\cF}$ is isomorphic to $\cW$.
\end{restatable} 

The previous result shows that the partition oracles for disambiguation introduced in Section~\ref{sec:disambiguation-conditions} naturally extend to the weighted case---but only if the resulting disambiguation $\cW_{\cF}$ is finite. Furthermore, these PF-oracles define the corresponding efficient disambiguation schemes when the weighted factorization $(\fac,\res)$ can be computed in polynomial time. 

One example of disambiguation is shown in Figure~\ref{fig:generalWFAExamples} with WFA $\cW$ over the natural numbers semiring $\bbS=  (\bbN, +, \cdot, 0, 1)$  and PF-oracle $\cF = (\Pi_{\cO_{\CF}}, \fac, \res)$, where $\fac(V) =  \textsf{GCD}(V)$ and $\res(V) = \frac{V}{\textsf{GCD}(V)}$ for all $V \in \bbS^Q$.
$\textsf{GCD}(V)$ is the \emph{greatest common divisor} of all elements in $V$.
$\cW_{\cF}$ is unambiguous and equivalent to $\cW$.

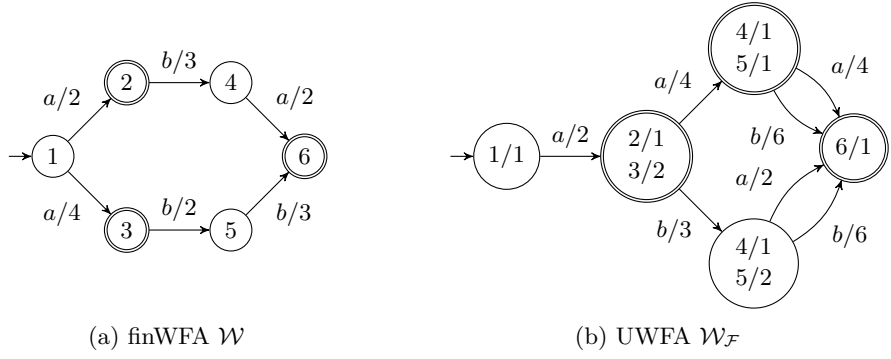
\begin{figure}[t]
	\centering
     \input{./figures/weightedAutomata/generalExample/original.tex}%
	    \caption{WFA $\cW$ and its disambiguation $\cW_{\cF}$. Every vector $V$ state in $\cW_{\cF}$ is shown as pairs $q/v$, where $V[q] = v$, states $q$ that are not part of the support of $V$ are omitted.}
	\label{fig:generalWFAExamples}
\end{figure}

\paragraph{Generalization of the {\em twins property}} For this subsection, fix the tropical semiring $\bbT = (\bbZ \cup \{\infty\}, \min, +, \infty, 0)$. In~\cite{mohriTransducer}, Mohri studied the \emph{twins property}, a condition towards characterizing termination for his algorithm for determinizing weighted automata over $\bbT$. 
This condition is sufficient in general, and also necessary when the weighted automata is unambiguous. 
Here, we demonstrate how to generalize the twins property to provide sufficient conditions for finite disambiguation across different levels of ambiguity.

Let $\cW$ be a WFA of the form ($\dagger$) over $\bbT$. %
Given a relation $R \subseteq Q \times Q$, we say that two states $p, q \in Q$ are $R$-\emph{twins} iff for every $u, w \in \Sigma^*$, whenever $p, q \in \DeltaFA(\IFA, u)$, $p \in \DeltaFA(p, w)$, $q \in \DeltaFA(q, w)$, and $(p,q) \in R$, then $\Delta_w[p,p] = \Delta_w[q,q]$. 
In other words, if $w$ forms a loop in $p$ and in $q$ and these nodes are related by $R$, then the weight of the loops must be the same. Then, we say that $\cW$ has the $R$-\emph{twins property} if every pair of states in $Q$ are $R$-twins. 

Note that the twins property in~\cite{mohriTransducer} is the special case when $R = Q \times Q$. 
We generalize this property for the relations introduced in Section~\ref{sec:disambiguation-conditions}, 
providing sufficient conditions for termination of the disambiguation algorithms.  

\begin{restatable}{theorem}{WFAtermination}\label{theo:WFAtermination}
		Let $\cW$ be a WFA over $\bbT$ and $R \in \{\CF_{\cW}, \ICF_{\cW}, \ECF_{\cW}\}$. If $\cW$ satisfies the $R$-twins property, then $\cW_{\cF}$ is finite when $\cF = (\Pi_{\cO_R}, \facMohri, \resMohri)$ is a PF-oracle.
\end{restatable} 

Similarly to~\cite{mohriTransducer}, the previous result provides sufficient conditions for termination when disambiguating WFA for each level of ambiguity. We leave as an open problem to understand when the $R$-twins property is a {\em necessary} condition in each case.

%% file: figures/weightedAutomata/generalExample/original.tex
\small
\begin{tikzpicture}[every label/.style={black!60}, initial text=, 
	initial distance= {3mm},
	->,>=stealth',
	auto,
	node distance=0.8cm,
	mystate/.style={state, inner sep=3pt, minimum size=0mm},
	longstate/.style={rectangle,inner sep=8pt,draw,rounded corners=11pt}]
  \node[mystate, initial]          (1)                {$1$};
  \node[mystate, accepting]          (2) [above right=of 1] {$2$};
  \node[mystate, accepting]          (3) [below right=of 1] {$3$};
  \node[mystate]          (4) [right=of 2] {$4$};
  \node[mystate]          (5) [right=of 3] {$5$};
  \node[mystate, accepting]          (6) [below right=of 4] {$6$};

  \path[->] (1) edge      node [above left]   {$a/2$} (2)
            (1) edge      node [below left]   {$a/4$} (3)
            (2) edge      node [above ]   {$b/3$} (4)
            (3) edge      node [above ]   {$b/2$} (5)
            (4) edge      node [above right]   {$a/2$} (6)
            (5) edge      node [below right]   {$b/3$} (6);
            
              \node at (1.5, -2.4) {(a) finWFA $\cW$};
            
     \begin{scope}[xshift=6cm]
     	  \node[mystate, initial]          (1)                {$1/1$};
     	\node[mystate, accepting]          (2) [right=of 1]  {\shortstack{2/1\\3/2}};
     	\node[mystate, accepting]          (3) [above right=of 2]  {\shortstack{4/1\\5/1}};
     	\node[mystate]          (4) [below right=of 2]  {\shortstack{4/1\\5/2}};
     	\node[mystate, accepting]          (5) [below right=of 3] {$6/1$};
     	
     	\path[->] (1) edge      node  {$a/2$} (2)
     	(2) edge      node   {$a/4$} (3)
     	(2) edge      node[swap]   {$b/3$} (4)
     	(3) edge  [bend left=20]    node    {$a/4$} (5)
     	(3) edge  [bend right=20]    node[swap]  {$b/6$} (5)
     	(4) edge  [bend left=20,pos=0.3]    node   {$a/2$} (5)
     	(4) edge  [bend right=20]    node[swap]  {$b/6$} (5);
     	
     	   \node at (2, -2.4) {(b) UWFA $\cW_{\cF}$};
     \end{scope}

\end{tikzpicture}

%% file: sections/conclusions-esa.tex
As an initial paper on this algorithmic framework, numerous approaches remain to be explored for further study. 
Besides the open problems we have stated throughout the paper, some natural follow-up work would be extending the framework onto tree automata~\cite{comon2008tree}, cost-register automata~\cite{AlurDDRY13}, or Büchi automata~\cite{thomas1990automata}. In the last case, it is an open problem to find good disambiguation algorithms, since deterministic Büchi automata are less expressive than unambiguous ones.  It will be interesting to see if the algorithmic framework presented here could lead to new insights on the~disambiguation~of~Büchi~automata.
We highlight one open problem, posed in Section~\ref{sec:disambiguation-conditions},  we find particularly interesting, which is to find optimal disambiguation schemes that minimize the number of explored states. %
Finally, our disambiguation framework is simple enough to be implemented and used in practice. 
We are interested in seeing it being explored as an alternative to determinization in real-life cases where an unambiguous, or a boundedly ambiguous automaton, is good~enough.

%% file: sections/app-mohri-counterexample.tex
A disambiguation algorithm for finite automata was proposed by Mehryar Mohri in~\cite{mohri13}. The main idea of Mohri's disambiguation is to work over an equivalent automaton $\cA_{\textsf{M}}$ that simulates the runs of the original automaton $\cA$, and the states keep track of competing runs (i.e., runs that have a common future). $\cA_{\textsf{M}}$ is equivalent to $\cA$ but it can still be ambiguous. For this reason, Mohri's disambiguation traverses the automaton $\cA_{\textsf{M}}$ in a depth-search manner, and prunes transitions that lead to ambiguous runs. The resulting unambiguous automaton depends on the order in which $\cA_{\textsf{M}}$ is traversed, that is, the order in which the transitions are chosen. For this reason, Mohri's disambiguation is not unique for each finite automaton $\cA$. Moreover, as we will show below, there exists a finite automaton $\cA$ such that for some order traversal of $\cA_{\textsf{M}}$, the resulting disambiguation is not equivalent to $\cA$. In other words, Mohri's disambiguation is not correct for all input. 

\begin{figure}[t]
    \centering
    \input{./figures/mohriCounterexample/original.tex}
    \caption{Example of an automaton $\cA$ such that is not equivalent to $\cA_{\textsf{M}}$ for some order of the states added to the queue.}
    \label{fig:counterexampleOriginal}
\end{figure}
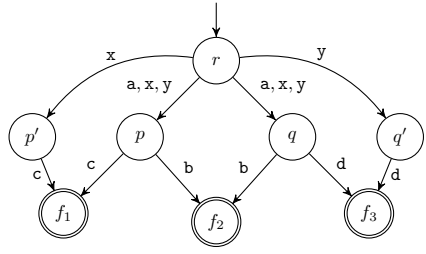

The counterexample is the NFA $\cA$ shown in Figure~\ref{fig:counterexampleOriginal}.
Next, we present the result of applying Mohri's disambiguation procedure over $\cA$. For a detailed presentation of Mohri's algorithm, we refer the reader to~\cite{mohri13}. 

In the following, we denote by $\cA_{\textsf{M}} = (Q_{\textsf{M}}, \Sigma, \Delta_{\textsf{M}}, I_{\textsf{M}}, F_{\textsf{M}})$ be the result of Mohri's procedure over $\cA$. 
During the construction, we assume that the states are enqueued in $\mathcal{Q}$ in the following order:
\begin{equation*}
\begin{aligned}
 &(r ,\{r\}), \ (p, \{p, p', q\}), \ (q , \{p, q, q'\}), \\[4pt]
 &(p, \{p, q\}), \ (q, \{p, q\}), \ (p', \{p, p'\}), \ (q', \{q, q'\}), \\[4pt]
 &(f_1, \{f_1\}), \ (f_2, \{f_2\}), \ (f_3, \{f_3\}).
\end{aligned}
\end{equation*}
\paragraph{Construction of $\cA_{\textsf{M}}$} The set of initial states and the queue are computed: 
\[
\mathcal{Q} \ = \ I_{\textsf{M}} \ = \  \{(r, \{r\})\}.
\] 
The set of transitions $\Delta_{\textsf{M}}$, the set of states $Q_{\textsf{M}}$ and set of final states $F_{\textsf{M}}$ starts empty. Next, the common future relation $\CF_{\cA}$ is computed. 
In this case, $\CF_{\cA}$ contains the pairs of states:
\[
(p', p), \quad (p, q), \quad (q, q').
\]
The first state $(r, \{r\})$ is dequeued, and all outgoing transitions from $r$ are examined. 
For each transition $(r, a, q) \in \Delta$, the algorithm constructs the state 
\[
(q, T) \quad \text{where} \quad 
T = \{ s \in \Delta(\{r\}, a) \mid (q, s) \in \CF_{\cA} \}.
\]
Then, the transition $((r, \{r\}), a, (q, T))$ is added to $\Delta_{\textsf{M}}$, 
unless there already exists a state $(p, S) \in Q_{\textsf{M}}$ such that 
$(p, S)$ and $(r, \{r\})$ are both reachable by the same word $w$ 
(in this case $w = \eps$), and the transition $((p, S), a, (q, T))$ 
is already in $\Delta_{\textsf{M}}$. 
This restriction avoids having two distinct runs over the same word 
from $I_{\textsf{M}}$ to $(q, T)$~\cite{mohri13}.

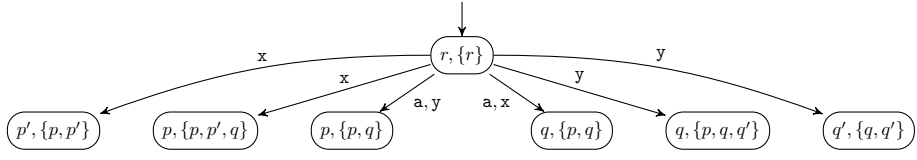
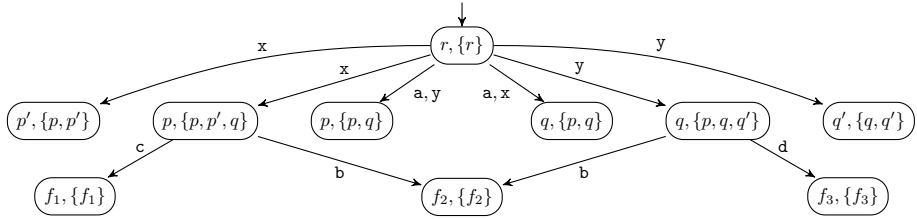
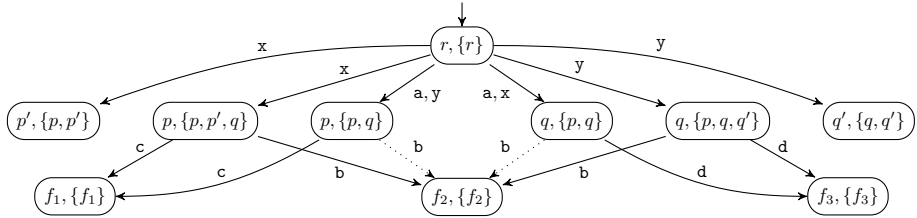
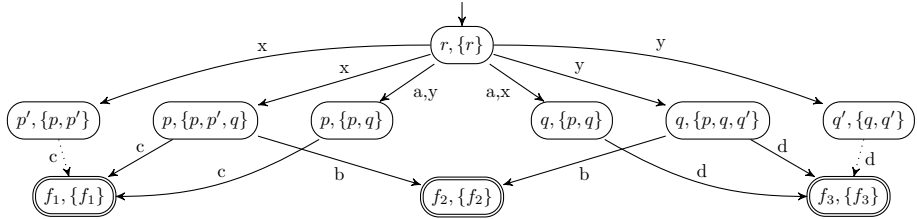
\begin{figure}[t]
	\centering
	
	\begin{subfigure}{0.99\textwidth}
		\centering
		\input{./figures/mohriCounterexample/partial.tex}
		\caption{Partial automaton $\cA_{\textsf{M}}$ after dequeing and processing the state $(r, \{r\})$.}
		\label{fig:counterexamplePartial}
	\end{subfigure}
	
	\vspace{2mm} 
	
	\begin{subfigure}{0.99\textwidth}
		\centering
		\input{./figures/mohriCounterexample/partial2.tex}
		\caption{Partial automaton $\cA_{\textsf{M}}$ after dequeing and processing the states $(p, \{p, p', q\})$ and $(q, \{p, q, q'\})$.}
		\label{fig:counterexamplePartial2}
	\end{subfigure}
	
	\vspace{2mm} 
	
	\begin{subfigure}{0.99\textwidth}
		\centering
		\input{./figures/mohriCounterexample/partial3.tex}
		\caption{Partial automaton $\cA_{\textsf{M}}$ after dequeing and processing the states $(p, \{p, q\})$ and $(q, \{p, q\})$, the dotted transitions were discarded by the algorithm.}
		\label{fig:counterexamplePartial3}
	\end{subfigure}
	
	\vspace{2mm} 
	
	\begin{subfigure}{0.99\textwidth}
		\centering
		\input{./figures/mohriCounterexample/disambiguation.tex}
		\caption{Automaton $\cA_{\textsf{M}}$, the dotted transitions were discarded by the algorithm.}
		\label{fig:counterexampleDisambiguation}
	\end{subfigure}
	
		\vspace{2mm}

	\caption{Mohri's disambiguation steps of the NFA $\cA$ shown in Figure~\ref{fig:counterexampleOriginal}.}
	\label{fig:vertical_stack}
\end{figure}

For example, consider the transition $(r, x, p) \in \Delta$. The algorithm creates the state $(p, \{p, p', q\})$, which is reachable by $x$. 
Here, the states $p, p', q$ are all reachable by $x$ and share a common future with $p$.  
Similarly, for the transition $(r, y, p) \in \Delta$, the state $(p, \{p, q\})$ is created. Although $p, q, q'$ are reachable by $y$, only $p$ and $q$ share a common future with $p$.

The state $(q, T)$ is then added to $Q_{\textsf{M}}$.  
In this first step, no transitions are discarded, and all newly created states are queued in $\mathcal{Q}$.  
The resulting partial automaton $\cA_{\textsf{M}}$ is shown in Figure~\ref{fig:counterexamplePartial}.

\paragraph{Processing $(p, \{p, p', q\})$ and $(q, \{p, q, q'\})$}  
These two states are dequeued next. All outgoing transitions from $p$ and $q$ are processed, producing new transitions in $\Delta_{\textsf{M}}$, none of which are discarded.  
As a result, the states $(f_1, \{f_1\}), (f_2, \{f_2\}), (f_3, \{f_3\})$ are queued.  
Since $f_1$, $f_2$, and $f_3$ do not share a common future with any other state, they remain as singleton states.  
The resulting automaton is shown in Figure~\ref{fig:counterexamplePartial2}.

\paragraph{Processing $(p, \{p, q\})$ and $(q, \{p, q\})$}  
For the state $(p, \{p, q\})$, $p$ has two outgoing transitions.  
The transition 
\[
((p, \{p, q\}), c, (f_1, \{f_1\}))
\]
is added to $\Delta_{\textsf{M}}$, 
but the transition 
\[
((p, \{p, q\}), b, (f_2, \{f_2\}))
\]
is discarded, since $(q, \{p, q, q'\})$ and $(p, \{p, q\})$ are both reachable by $y$, and $(q, \{p, q, q'\})$ already has a transition over $b$ to $(f_2, \{f_2\})$.  
Analogously, when processing $(q, \{p, q\})$, only one of its two transitions is added.  
The resulting automaton is shown in Figure~\ref{fig:counterexamplePartial3}.

\paragraph{Final steps}  
The states $(p', \{p, p'\})$ and $(q', \{q, q'\})$ are dequeued, but their outgoing transitions are discarded.  
The states $(f_1, \{f_1\})$ and $(f_2, \{f_2\})$ are already reachable by $xc$ from $(p, \{p, p', q\})$ and by $yd$ from $(q, \{p, q, q'\})$, respectively.  
Finally, the states $(f_1, \{f_1\})$, $(f_2, \{f_2\})$, and $(f_3, \{f_3\})$ are dequeued and added to $F_{\textsf{M}}$.  
The resulting automaton $\cA_{\textsf{M}}$ is shown in Figure~\ref{fig:counterexampleDisambiguation}.  
One can easily note that the word $ab \in \cL(\cA)$ is not accepted by $\cA_{\textsf{M}}$, thus $\cA_{\textsf{M}}$ is not equivalent to $\cA$.

In conclusion, Mohri's algorithm is not correct for all finite automata.
Specifically, the step of discarding outgoing transitions from states $(p, \{p, q\})$ correctly discards an already existing accepting run over $yb$, but also discards the accepting run over $ab$. 
Therefore, disambiguation procedures that are not computed on the fly needs to track a significant amount of information to ensure that no words in $\cL(\cA)$ are discarded.

%% file: figures/mohriCounterexample/original.tex
\begin{tikzpicture}[every label/.style={black!60}, initial text=, ->,>=stealth', node distance=0.8cm, every node/.style={inner sep=3pt, scale=0.7}]
  \node[state,initial, initial where=above]  (r)     {$r$};
  \node[state]          (p) [below left=of r] {$p$};
  \node[state]          (q) [below right =of r] {$q$};
  \node[state]          (p') [left =of p] {$p'$};
  \node[state]          (q') [right =of q] {$q'$};
  \node[state,accepting]          (f_1) [below left =of p] {$f_1$};
  \node[state,accepting]          (f_2) [below =1.5cm of r] {$f_2$};
  \node[state,accepting]          (f_3) [below right =of q] {$f_3$};

  \path[->] (r) edge              node [above left]   {${\tt a},{\tt x},{\tt y}$} (p)
            (r) edge              node [above right]   {${\tt a},{\tt x},{\tt y}$} (q)
            (r) edge [bend right]  node [above]   {${\tt x}$} (p')
            (r) edge [bend left] node [above]   {${\tt y}$} (q')
            (p) edge  node [above left]   {${\tt c}$} (f_1)
            (p') edge  node [left]   {${\tt c}$} (f_1)
            (q) edge  node [above right]   {${\tt d}$} (f_3)
            (q') edge  node [right]   {${\tt d}$} (f_3)
            (q) edge  node [above left]   {${\tt b}$} (f_2)
            (p) edge  node [above right]   {${\tt b}$} (f_2);

\end{tikzpicture}

%% file: figures/mohriCounterexample/partial.tex
\begin{tikzpicture}[every label/.style={black!60}, initial text=, ->,>=stealth', node distance=0.7cm, every node/.style={scale=0.7}, state/.style={rectangle,inner sep=5pt,draw,rounded corners=6.5pt}]
  \node[state,initial, initial where=above]  (r)     {$r, \{r\}$};
  \node[state]          (p_1) [below left=of r] {$p, \{p, q\}$};
  \node[state]          (q_1) [below right =of r] {$q, \{p, q\}$};
  \node[state]          (p_2) [left =of p_1] {$p, \{p, p', q\}$};
  \node[state]          (q_2) [right =of q_1] {$q, \{p, q, q'\}$};
  \node[state]          (p') [left =of p_2] {$p', \{p, p'\}$};
  \node[state]          (q') [right =of q_2] {$q', \{q, q'\}$};

  \path[->] (r) edge              node [below right]   {${\tt a},{\tt y}$} (p_1)
            (r) edge              node [below left]   {${\tt a},{\tt x}$} (q_1)
            (r) edge              node [above]   {${\tt x}$} (p_2)
            (r) edge              node [above]   {${\tt y}$} (q_2)
            (r) edge[out=180,in=20] node [above]   {${\tt x}$} (p')
            (r) edge[out=0,in=160]  node [above]   {${\tt y}$} (q');

\end{tikzpicture}

%% file: figures/mohriCounterexample/partial2.tex
\begin{tikzpicture}[every label/.style={black!60}, initial text=, ->,>=stealth', node distance=0.7cm, scale=0.7, every node/.style={scale=0.7}, state/.style={rectangle,inner sep=5pt,draw,rounded corners=6.5pt}]
  \node[state,initial, initial where=above]  (r)     {$r, \{r\}$};
  \node[state]          (p_1) [below left=of r] {$p, \{p, q\}$};
  \node[state]          (q_1) [below right =of r] {$q, \{p, q\}$};
  \node[state]          (p_2) [left =of p_1] {$p, \{p, p', q\}$};
  \node[state]          (q_2) [right =of q_1] {$q, \{p, q, q'\}$};
  \node[state]          (p') [left =of p_2] {$p', \{p, p'\}$};
  \node[state]          (q') [right =of q_2] {$q', \{q, q'\}$};
  \node[state]          (f_1) [below left=of p_2] {$f_1, \{f_1\}$};
  \node[state]          (f_2) [below=1.5cm of r] {$f_2, \{f_2\}$};
  \node[state]          (f_3) [below right=of q_2] {$f_3, \{f_3\}$};

  \path[->] (r) edge              node [below right]   {${\tt a},{\tt y}$} (p_1)
            (r) edge              node [below left]   {${\tt a},{\tt x}$} (q_1)
            (r) edge              node [above]   {${\tt x}$} (p_2)
            (r) edge              node [above]   {${\tt y}$} (q_2)
            (r) edge [out=180,in=20]  node [above]   {${\tt x}$} (p')
            (r) edge [out=0,in=160]  node [above]   {${\tt y}$} (q')
            (p_2) edge             node [above]   {${\tt c}$} (f_1)
            (q_2) edge             node [above]   {${\tt d}$} (f_3)
            (p_2) edge            node [below]   {${\tt b}$} (f_2)
            (q_2) edge            node [below]   {${\tt b}$} (f_2);

\end{tikzpicture}

%% file: figures/mohriCounterexample/partial3.tex
\begin{tikzpicture}[every label/.style={black!60}, initial text=, ->,>=stealth', node distance=0.7cm, scale=0.7, every node/.style={scale=0.7},state/.style={rectangle,inner sep=5pt,draw,rounded corners=6.5pt}]
  \node[state,initial, initial where=above]  (r)     {$r, \{r\}$};
  \node[state]          (p_1) [below left=of r] {$p, \{p, q\}$};
  \node[state]          (q_1) [below right =of r] {$q, \{p, q\}$};
  \node[state]          (p_2) [left =of p_1] {$p, \{p, p', q\}$};
  \node[state]          (q_2) [right =of q_1] {$q, \{p, q, q'\}$};
  \node[state]          (p') [left =of p_2] {$p', \{p, p'\}$};
  \node[state]          (q') [right =of q_2] {$q', \{q, q'\}$};
  \node[state]          (f_1) [below left=of p_2] {$f_1, \{f_1\}$};
  \node[state]          (f_2) [below =1.5cm of r] {$f_2, \{f_2\}$};
  \node[state]          (f_3) [below right=of q_2] {$f_3, \{f_3\}$};

  \path[->] (r) edge              node [below right]   {${\tt a},{\tt y}$} (p_1)
            (r) edge              node [below left]   {${\tt a},{\tt x}$} (q_1)
            (r) edge              node [above]   {${\tt x}$} (p_2)
            (r) edge              node [above]   {${\tt y}$} (q_2)
            (r) edge [out=180,in=20] node [above]   {${\tt x}$} (p')
            (r) edge [out=0,in=160]  node [above]   {${\tt y}$} (q')
            (p_2) edge             node [above]   {${\tt c}$} (f_1)
            (q_2) edge             node [above]   {${\tt d}$} (f_3)
            (p_2) edge            node [below]   {${\tt b}$} (f_2)
            (q_2) edge            node [below]   {${\tt b}$} (f_2)
            (p_1) edge [out=-150,in=0]   node [above]   {${\tt c}$} (f_1) 
            (p_1) edge [dotted]   node [above right]   {${\tt b}$} (f_2)
            (q_1) edge [out=-30,in=180] node [above]   {${\tt d}$} (f_3)
            (q_1) edge [dotted]   node [above left]   {${\tt b}$} (f_2);

\end{tikzpicture}

%% file: figures/mohriCounterexample/disambiguation.tex
\begin{tikzpicture}[every label/.style={black!60}, initial text=, ->,>=stealth', node distance=0.7cm, scale=0.7, every node/.style={scale=0.7},state/.style={rectangle,inner sep=5pt,draw,rounded corners=6.5pt}]
  \node[state,initial, initial where=above]  (r)     {$r, \{r\}$};
  \node[state]          (p_1) [below left=of r] {$p, \{p, q\}$};
  \node[state]          (q_1) [below right =of r] {$q, \{p, q\}$};
  \node[state]          (p_2) [left =of p_1] {$p, \{p, p', q\}$};
  \node[state]          (q_2) [right =of q_1] {$q, \{p, q, q'\}$};
  \node[state]          (p') [left =of p_2] {$p', \{p, p'\}$};
  \node[state]          (q') [right =of q_2] {$q', \{q, q'\}$};
  \node[state, accepting] (f_1) [below left=of p_2] {$f_1, \{f_1\}$};
  \node[state, accepting] (f_2) [below =1.5cm of r] {$f_2, \{f_2\}$};
  \node[state, accepting] (f_3) [below right=of q_2] {$f_3, \{f_3\}$};

  \path[->] (r) edge              node [below right]   {a,y} (p_1)
            (r) edge              node [below left]   {a,x} (q_1)
            (r) edge              node [above]   {x} (p_2)
            (r) edge              node [above]   {y} (q_2)
            (r) edge [out=180,in=20] node [above]   {x} (p')
            (r) edge [out=0,in=160] node [above]   {y} (q')
            (p_2) edge             node [above]   {c} (f_1)
            (q_2) edge             node [above]   {d} (f_3)
            (p_2) edge            node [below]   {b} (f_2)
            (q_2) edge            node [below]   {b} (f_2)
            (p_1) edge [out=-150,in=0]   node [above]   {c} (f_1)
            (q_1) edge [out=-30,in=180]  node [above]   {d} (f_3)
            (p') edge [dotted]           node [left]   {c} (f_1)
            (q') edge [dotted]           node [right]   {d} (f_3);

\end{tikzpicture}

%% file: sections/app-framework.tex
For the proofs that follow, we denote by $\Delta^*$ the transitive closure of $\Delta$, representing word transitions.
We also extend the notation $\cO$ with the operator $\cO(S,q)$ when $q\in S$ as the set in $\cO(S)$ that contains $q$.

We start with an auxiliary result that indicates that for every transition in the subset construction, each state in $S'$ has at least one predecessor state in $S$.

\begin{lemma}\label{lemma:predecessor}
    Let $\cA_{\cO} = (2^Q, \Sigma, \Delta_{\cO}, I', F')$ be the result automaton of $\cA = (Q, \Sigma, \Delta, I, F)$ using an oracle $\cO$. Then for every pair of states $S, S' \in 2^Q$:
    \begin{enumerate}
        \item If $(S, a, S') \in \Delta_{\cO}$, for every state $q \in S'$, there exists a state $p \in S$ such that $(p, a, q) \in \Delta$.
        \item More generally, if $(S, w, S') \in \Delta_{\cO}^*$, then for every state $q \in S'$, there exists a state $p \in S$ such that $(p, w, q) \in \Delta^*$.
    \end{enumerate}
\end{lemma}

\begin{proof}
    It follows directly from the definition of $\Delta_{\cO}$. The second part can be proved by induction on the length of $w$.
\end{proof}

\subsection{Proof of Proposition~\ref{prop:frameworkEquivalence}}

\frameworkEquivalence*

\begin{proof}

    First, we will prove that $\cL(\cA) \subseteq \cL(\cA_{\cO})$. Let $\rho$ be an accepting run of $\cA$ over a word $w = a_1\ldots a_n$:
    \[
       \rho :=  p_0 \xrightarrow{a_1} p_1 \xrightarrow{a_2}  \cdots \xrightarrow{a_n} p_n .
    \]
    We build a run $\rho'$ of $\cA_{\cO}$ as follows:
    \[
            \rho'  :=  S_0 \xrightarrow{a_1} S_1 \xrightarrow{a_2} \cdots \xrightarrow{a_n} S_n,
    \] 
    where $S_0 =   \cO(I, p_0)$, and 
    for each $i \in [n]$ we define $S_{i} = \cO(\Delta(S_{i-1}, a_i), p_i)$. This construction is valid since we can inductively see that at each step, $p_i\in \Delta(S_{i-1}, a_i)$.
    Then we have that $S_n \in F'$ since $p_n \in F$, so $\cA_{\cO}$ accepts $w$.
      
    For the reverse inclusion $\cL(\cA_{\cO}) \subseteq \cL(\cA)$, consider an accepting run $\rho'$ of $\cA_{\cO}$ over some word $w=a_1\ldots a_n$:
\[
            \rho' :=  S_0 \xrightarrow{a_1} S_1 \xrightarrow{a_2} \cdots \xrightarrow{a_n} S_n
    \]

	We can build a run $\rho$ of $\cA$ by starting with the last state set $S_n$ in $\rho'$, which contains some final state $q$, and working backwards. 
    Given Lemma~\ref{lemma:predecessor} we note that $q\in\Delta(S_{n-1},a_n)$, so there must be a transition $(p, a_n, q) \in \Delta$ for some state $p\in S_{n-1}$. We build an accepting run $\rho$ extending this idea for the rest of $w$.

\end{proof}

%% file: sections/app-disambiguationConditions.tex
We present several propositions and definitions that will be used in the subsequent proofs.
We start with the {\em product construction} for finite automata: 
given two finite automata $\cA_1 = (Q_1, \Sigma, \Delta_1, I_1, F_1)$ and 
$\cA_2 = (Q_2, \Sigma, \Delta_2, I_2, F_2)$, we say
$\cA_1\times\cA_2 := (Q_1\times Q_2, \Sigma, \Delta', I_1\times I_2, F')$
is a {\em product automaton} of $\cA_1$ and $\cA_2$ if $F'\subseteq Q_1\times Q_2$ and
\[
\Delta' := \{((p_1,p_2),a,(q_1,q_2))\mid (p_1,a,q_1)\in\Delta_1\text{ and }(p_2,a,q_2)\in\Delta_2\}.
\]
Creating $\Delta'$ takes time $O(|\Delta|^2)$ by looking at each tuple 
$((p_1,a_1,q_1),(p_2,a_2,q_2))\in \Delta\times\Delta$ and adding 
$((p_1,p_2),a_1,(q_1,q_2))$ to $\Delta'$ iff $a_1 = a_2$.
This is a classical construction, and it is used to build automata with language $\cL(\cA_1)\cap\cL(\cA_2)$ and $\cL(\cA_1)\cup\cL(\cA_2)$  (see~\cite{hopcroftlibro} for example).
The property we are interested in is that for any pair of states $p\in Q_1$ and $q\in Q_2$, there is a path in $\cA_1\times\cA_2$ that reaches $(p,q)$ iff there exists a word $u\in\Sigma^*$ for which there exist a path in $\cA_1$ that reaches $p$ after reading $u$, and a path in $\cA_2$ that reaches $q$ after reading the same $u$. 
In this section, we will work with product automata but not their languages, 
or their final sets $F'$, so we will define $\cA_1\times\cA_2$ using a random $F'$.

Two more constructions: given an NFA $\cA = (Q, \Sigma, \Delta, I, F)$, we define the automaton $\cA_q$ for some $q\in Q$ as $\cA_q := (Q, \Sigma, \Delta, \{q\}, F)$; also, we define $\cA^R:= (Q, \Sigma, \Delta^R, F, I)$ by $\Delta^R := \{(q,a,p)\mid(p,a,q)\in\Delta\}$, as the {\em reverse} of $\cA$, which accepts the language $\{w^R\mid w\in\cL(\cA)\}$, where $w^R$ is the reverse of a word $w\in\Sigma^*$.

The following result will be used for various idempotency results.

\begin{proposition}\label{prop:frameworkIsomorphic}
    If every state in $\cA_{\cO}$ is a singleton, then $\cA_{\cO}$ is isomorphic to $\cA$.
\end{proposition}

\begin{proof}
By definition, the state set of $\cA_{\cO}$ is a subset of $2^Q$, so every state in $\cA_{\cO}$ is of the form $\{q\}$ for some $q\in Q$. By contradiction, suppose that  $\cA_{\cO}$ is not isomorphic to $\cA$ via the function $f(q) = \{q\}$.

	We identify three possible cases: There is some state $q$ in $\cA$ such that $\{q\}$ is not a state in $\cA_{\cO}$; there is either some transition $(p,a,q)\in\Delta$ such that $(\{p\},a,\{q\})\not\in\Delta_{\cO}$, or there is some transition $(\{p\},a,\{q\})\in\Delta_{\cO}$ such that $(p,a,q)\not\in\Delta$.

We prove by induction on the execution the algorithm that for every state $q$ at distance $n$ from the set $I$, the induced subautomaton (analogous to an induced subgraph) of states at distance $n$ from $I$ in $\cA$ is isomorphic to the analogous subautomaton of the states at distance $n$ from $\cO(I)$ in $\cA_{\cO}$.

\textbf{Base case ($n = 0$):}  
For every initial state $p \in I$, the corresponding initial state in the result automaton $\cA_{\cO}$ is the singleton set $\{p\} \in I'$, since $\cO$ partitions $I$ into singletons. Hence, the initial states correspond one-to-one.

\textbf{Inductive step:}  
Assume that for a word $w$ of length $n$, every reachable state in $\cA_{\cO}$ over $w$ is a singleton set $\{p\}$, where $p \in Q$ is a state of the original automaton $\cA$.  

Consider a letter $a \in \Sigma$ and a reachable state $\{p\}$ by $w$.

By definition of the transition relation $\Delta_{\cO}$ in the disambiguation construction, the successor states after reading $a$ from $\{p\}$ are
\[
\Delta_{\cO}(\{p\}, a) = \cO(\Delta(p, a)),
\]

Since every state in $\cA_{\cO}$ is a singleton by hypothesis, the oracle $\cO$ partitions $\Delta(p, a)$ into singleton subsets.
Hence, for each transition $(p, a, q) \in \Delta$, there is a corresponding unique singleton state $\{q\}$ reachable by the word $wa$ and a transition $(\{p\}, a, \{q\}) \in \Delta_{\cO}$.

This establishes a one-to-one correspondence between states and transitions of $\cA_{\cO}$ and $\cA$, proving the isomorphism.
\end{proof}

Let us restate some definitions for convenience:
given a trimmed automaton $\cA$, the \emph{degree of ambiguity} $\daa(w)$ of a word $w$ is the number of different accepting runs of an automaton $\cA$ over $w$. Similarly, the \emph{degree of ambiguity} $\da(\cA)$ of an automaton $\cA$ is the maximum degree of ambiguity for some word $w \in L(\cA)$. If no such maximum exists, then $\da(\cA) = \infty$.

Next, we recall the criteria introduced by Weber and Seidl~\cite{weber1991degree}, which characterise the degree of ambiguity of NFA.

\paragraph{Infinite Degree of Ambiguity (IDA)} An NFA $\cA$ satisfies the (IDA) condition if there exist two states $r, s \in Q$ such that for some word $w \in \Sigma^*$ it holds that $(r, w, r), (r, w, s), (s, w, s) \in \Delta^*$. If $\cA$ satisfies (IDA), then it is (at least) polynomially ambiguous.
Finally, we say that the pair $(r, s)$ satisfies (IDA).

\paragraph{Exponential Degree of Ambiguity (EDA)} An NFA $\cA$ satisfies the (EDA) condition if there exists a state $s \in Q$ such that, for some word $w \in \Sigma^*$, $\daa(s, w, s) \geq 2$. If $\cA$ satisfies (EDA), then it is exponentially ambiguous.

Given an NFA~$\cA$, the relations $\ICF_{\cA}$ and $\ECF_{\cA}$ identify pairs of states that would satisfy the (IDA) or (EDA) conditions if not merged.

To refine these characterisations, we introduce (IDA)' and (EDA)', which will be useful in the proofs that follow. (Note that these criteria are different than the auxiliar (IDA)' and (EDA)' that appear in~\cite{weber1991degree}.)

The structure formed by (IDA) consists of two states $r, s \in Q$. From $r$, there exist two distinct runs over a word $w$, one that ends in $r$, and another that ends in $s$. 
We are interested in the state at which these runs diverge (Figure~\ref{fig:IDAequivalence}).

\paragraph{(IDA)'} There exist states $r, t, p, q, s \in Q$, words $u, v$, and letter $a$ such that:
\begin{itemize}
    \item $(r, u, t),\,(p, v, r),\,(q, v, s),\,(s, uav, s) \in \Delta^*$
    \item $(t, a, p),\,(t, a, q) \in \Delta$
\end{itemize}
Additionally, we say that the pair of states $(p, q)$ satisfies (IDA)' if such $r, t, s \in Q$ exist.
\begin{proposition}\label{app:idaprime}
    $\cA$ satisfies (IDA) iff $\cA$ satisfies (IDA)'.
\end{proposition}
\begin{proof}
    First, suppose that $\cA$ satisfies (IDA).  
    Let $(r, s)$ be the pair of states that satisfies (IDA) in $\cA$.  

    Then there exist two runs $\rho_1$ and $\rho_2$ over a word $w$ of length $n$ such that:
    \[
    \begin{array}{rcl}
        \rho_1 & := & r \xrightarrow{a_1} r^1_1 \xrightarrow{a_2} r^1_2 \xrightarrow{a_3} \ldots \xrightarrow{a_n} r \\
        \rho_2 & := & r \xrightarrow{a_1} r^2_1 \xrightarrow{a_2} r^2_2 \xrightarrow{a_3} \ldots \xrightarrow{a_n} s \\
    \end{array}
    \]
    Let $i$ be the maximum index such that $r^1_j = r^2_j$ for all $j \in [i]$. Then we can define the states $t = r^1_i$, $p = r^1_{i+1}$, and $q = r^2_{i+1}$.  

    Now we can rewrite the runs $\rho_1$ and $\rho_2$ using these new states:
    \[
    \begin{array}{rcl}
        \rho_1 & := & r \xrightarrow{u} t \xrightarrow{a} p \xrightarrow{v} r \\
        \rho_2 & := & r \xrightarrow{u} t \xrightarrow{a} q \xrightarrow{v} s \\
        w & = &  u \cdot a \cdot v \\
        u, v & \in & \Sigma^*, \quad a \in \Sigma
    \end{array}
    \]
    Hence, $\cA$ satisfies (IDA)'.  

    Conversely, if $\cA$ satisfies (IDA)', one can directly compose the two runs $\rho_1$ and $\rho_2$ to see that (IDA) is satisfied by the pair $(r, s)$.
\end{proof}

\begin{figure}[t]
    \centering
    \subfloat[(IDA)]{
        \input{./figures/kDisambiguation/IDA.tex}
    }
    \hspace{1cm}
    \subfloat[(IDA)']{
        \input{./figures/kDisambiguation/IDAalternative.tex}
    }
    \caption{Graphical representation of (IDA) and (IDA)'.}
    \label{fig:IDAequivalence}
\end{figure}
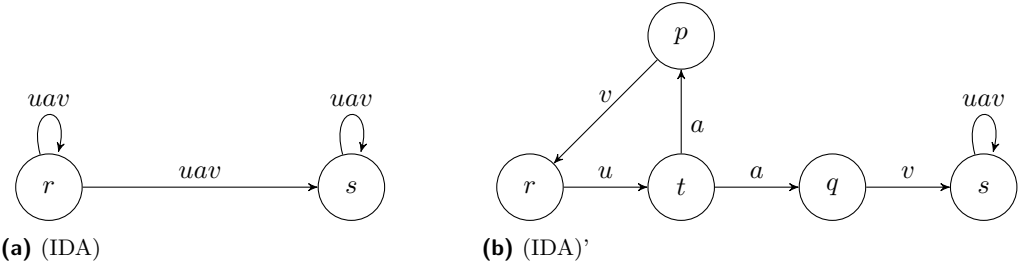
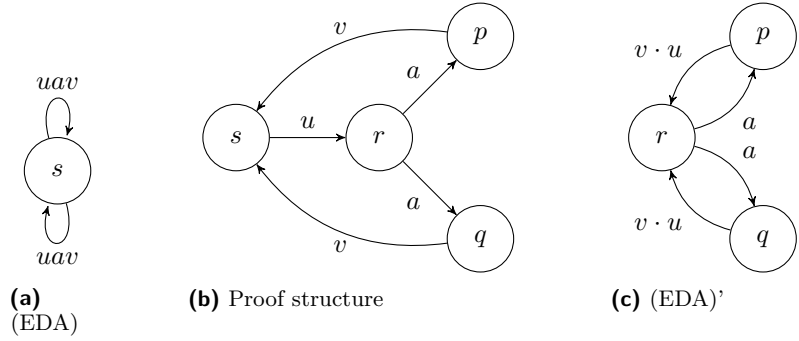
\begin{figure}[t]
    \centering
    \subfloat[(EDA)]{
        \input{./figures/kDisambiguation/EDA.tex}
    }
    \hspace{1cm}
    \subfloat[Proof structure]{
        \input{./figures/kDisambiguation/EDAintermediate.tex}
    }
    \hspace{1cm}
    \subfloat[(EDA)']{
        \input{./figures/kDisambiguation/EDAalternative.tex}
    }
    \caption{Graphical representation of (EDA) and (EDA)'.}
    \label{fig:EDAequivalence}
\end{figure}

As in the case of (IDA)', we focus on the state where the two runs from $s$ diverge in (EDA) (Figure~\ref{fig:EDAequivalence}).

\paragraph{(EDA)'} There exist distinct states $r, p, q \in Q$ such that $(r, a, p), (r, a, q), (p, w, r), (q, w, r) \in \Delta^*$ for some letter $a$ and word $w$.  
Additionally, we say that the pair of states $(p, q)$ satisfies (EDA)' if such an $r\in Q$ exists.
\begin{proposition}\label{app:edaprime}
    $\cA$ satisfies (EDA) iff $\cA$ satisfies (EDA)'.
\end{proposition}
\begin{proof}
    The proof is divided into the following equivalences:
    \[
    \begin{array}{rcl}
\cA \text{ satisfies } (EDA) & \iff & \exists s \in Q. \;\exists w \in \Sigma^*. \; \daa(s, w, s) \geq 2 \\
        & \iff & \exists s, r \in Q. \;\exists u, v \in \Sigma^*. \; a \in \Sigma. \; w = uav \;\land\; \\ 
        & & \qquad \daa(s, u, r) = 1 \;\land\; \daa(r, av, s) \geq 2 \\
        & \iff & \exists s, r, p, q \in Q. \;\exists u, v \in \Sigma^*. \; a \in \Sigma. \; p \neq q\;\land\;\\
        &  &\qquad (r, a, p),\,(r, a, q),\, (p, v, s),\, (q, v, s), (s, u, r) \in \Delta^* \\
        & \iff & \cA \text{ satisfies } (EDA)'  
    \end{array} \] \end{proof}
\subsection{Proof of Theorem~\ref{theo:UFA-disambiguation}}

\UFAdisambiguation*

\begin{proof}
    \textbf{1)} 
    Towards a contradiction, suppose $\cA_{\cO_{\CF}}$ is not unambiguous.
    Then there exist two different accepting runs $\rho_1$ and $\rho_2$ over the same word $w= a_1a_2\ldots a_n$ such that:

	\[
		\begin{array}{rcl}
			\rho_1 \ & := & \ S_0 \, \xrightarrow{a_1, v_1} \ldots \xrightarrow{a_{i-1}, v_{i-1}} S_{i-1} \xrightarrow{a_i, v_i} S_i \xrightarrow{a_{i+1}, v_{i+1}} \ldots  \xrightarrow{a_n, v_n} \, S_n 	\\
			\rho_2 \ & := & \ S_0' \, \xrightarrow{a_1, v_1'} \ldots \xrightarrow{a_{i-1}, v_{i-1}'} S_{i-1}' \xrightarrow{a_i, v_i'} S_i' \xrightarrow{a_{i+1}, v_{i+1}'} \ldots \xrightarrow{a_n, v_n'} \, S_n' 	
		\end{array}	
	\]

	Let $i$ be the index at which the runs differ (i.e., $S_j = S_j' \; \forall j < i$ and $S_i \neq  S_i'$). 

    There exists a state $p \in S_i$ from which there is a partial run over the suffix $a_{i+1} \ldots a_n$ that ends in a state $p_n \in S_n$. Likewise, there exists a state $p' \in S_i'$ that also reaches a state $p_n' \in S_n'$ over the same suffix.

	By the definition of $\Delta_{\cO_{\CF}}$ and the oracle $\cO_{\CF}$, this implies that $(p, p') \in \CF_\cA$ and thus $S_i = S_i'$, contradicting the assumption that the index $i$ exists.
	Therefore $\cA_{\cO_{\CF}}$ must be unambiguous.

    \medskip
    \textbf{2)} 
    If $\cA$ is already unambiguous, then no two distinct states reachable from $I$ by the same input word share a common future.  
    Consequently, the oracle $\cO_{\CF}$ partitions the state set into singletons.  
    By Proposition~\ref{prop:frameworkIsomorphic}, it follows that $\cA_{\cO_{\CF}}$ is isomorphic to $\cA$.
    
    \medskip
    \textbf{3)} 
    Consider the automaton $(\cA^R)\times (\cA^R)$ which has $F\times F$ as its initial set of states. For any state $(p,q)$ that is reachable from some state in $F\times F$ it must hold that $p$ and $q$ share a common future. This gives us an $O(|\cA|^2)$ construction to compute $\CF_\cA$.
\end{proof}

\subsection{Proof of Theorem~\ref{theo:UFA-minimality}}

\UFAminimality*

\begin{proof}
    Assume there exists an oracle $\cO$ and state set $S \in \{\Delta(S',a) \mid S' \in \reach(\cA_{\cO}) \wedge a \in \Sigma\}$ where $\cO_{\CF}(S)$ does not refine $\cO(S)$.
    We define the automaton $\cA_{\cO} = (2^Q, \Sigma, \Delta_{\cO}, I', F')$ given by the framework. 
    Assume that $\cA_{\cO}$ unambiguous.

    Since $\CF(S)$ does not refine $\cO(S)$, there exist states $p, q \in S$ such that $(p, q) \in \CF$, and there exist distinct sets $S_p, S_q \in \cO(S)$ with $p \in S_p$ and $q \in S_q$.

    Since $S_p$ and $S_q$ are reachable by a word $w$ and $\cA_{\cO}$ is trimmed, then there exist two accepting runs in $\cA_{\cO}$ such that:
    \[
        \begin{array}{rcl}
            \rho' & := & S_{p_0} \xrightarrow{w} S_p \xrightarrow{u} S_{p_n} \\
            \rho'_2 & := & S_{q_0} \xrightarrow{w} S_q \xrightarrow{u} S_{q_n} \\
        \end{array}
    \]
    This implies that $\cA_{\cO}$ is not unambiguous, contradicting our assumption that it is unambiguous.

    Therefore, no such partition oracle $\cO$ exists, and $\CF$ is minimal among partition oracles producing unambiguous automata.
\end{proof}

\subsection{Proof of Theorem~\ref{theo:finFA-polyFA-disambiguation}}

\finFApolyFAdisambiguation*

We separate the proof in three parts.

\begin{lemma}
	For every NFA $\cA = (Q, \Sigma, \Delta, I, F)$, the relations $\ICF_\cA$ and $\ECF_\cA$ can be computed in time $O(|\cA|^5)$ and $O(|\cA|^3)$, respectively.
\end{lemma}

\begin{proof}
To compute $\ECF_\cA$, consider the automaton $\cA_r\times\cA_r$ for some $r\in Q$. 
We note that the pairs of states $(p,q)$ that are reachable from $(r,r)$ are precisely those for which there is a word $u$ such that $p\in\Delta(r,u)$ and $q\in\Delta(r,u)$.
    Further, if we do the same reasoning for $(\cA^R)_r\times(\cA^R)_r$, the pairs of states $(p,q)$ that are reachable from $(r,r)$ are those for which there is a word $v$ such that $r\in\Delta(p,v)$ and $r\in\Delta(q,v)$.

The reasoning above gives us an $O(|\cA|^3)$ algorithm to compute $\ECF_\cA$:
first, build automata $\cA_r\times\cA_r$ and $(\cA^R)_r\times(\cA^R)_r$ for every $r\in Q$;
then, find all triples $(p,q,r)$ such that $(p,q)$ is reachable from $(r,r)$ both in  $\cA_r\times\cA_r$
and in $(\cA^R)_r\times(\cA^R)_r$, by a linear-time pass on each;
for any $p,q\in Q$ it holds that $(p,q)\in\ECF_\cA$ iff there exists such an $r$.

To see how to compute $\ICF_\cA$, let us fix two states $r,s\in Q$ and consider automaton $\cA_r\times\cA_r\times\cA_s$ (the product construction can be easily generalized to more than two automata).
We see that for any $p,q,t\in Q$ it holds that $(p,q,t)$ is reachable from $(r,r,s)$ iff $p\in\Delta(r,u)$, $q\in\Delta(r,u)$ and
$t\in\Delta(s,u)$ for some $u\in\Sigma^*$.
We do the same reasoning for $(\cA^R)_r\times(\cA^R)_s\times(\cA^R)_s$, and note that $(p,q,t)$ is reachable from $(r,s,s)$ iff  $r\in\Delta(p,v)$, $s\in\Delta(q,v)$, and $s\in\Delta(t,v)$ for some $v\in\Sigma^*$.
We see that if any triple $(p,q,t)$ satisfies both conditions, then, additionally, $s\in\Delta(s,u\cdot v)$ through the state $t$. We conclude that $(p,q)\in\ICF_\cA$ iff these states $r$, $s$ and $t$ exist. Using an analogous reasoning as for $\ECF_\cA$, this gives us an $O(|\cA|^5)$ algorithm to compute $\ICF_\cA$. 
\end{proof}

\begin{lemma}
	For every NFA $\cA = (Q, \Sigma, \Delta, I, F)$, the following properties hold. 
	\begin{enumerate}
		\item $\cA_{\cO_{\ICF}}$ is always finitely ambiguous.
		\item If $\cA$ is finitely ambiguous, then $\cA_{\cO_{\ICF}}$ is isomorphic to $\cA$.
	\end{enumerate}
\end{lemma}
\begin{proof}
    Suppose that the automaton $\cA_{\cO_{\ICF}}$ is not a finFA and satisfies (IDA)', then there exist states $R, P, T, \hat{Q}, S \in 2^Q. \; P \neq \hat{Q}$, a letter $a$ and words $u, v$ such that:

    \[
        \begin{array}{rcl}
            (R, u, T), (T, a, P), (P, v, R), (T, a, \hat{Q}), (\hat{Q}, v, S), (S, uav, S) & \in & \Delta_{\cO_{\ICF}}^* \\
        \end{array}
    \]
    The proof strategy is to identify two states $p \in P$ and $q \in \hat{Q}$ sharing an infinite common future, leading to a pair of states $(r, s)$ that satisfies (IDA).
    For that purpose, we will obtain the transitions that prove $(p, q) \in \ICF_{\cA}$ relying on Lemma~\ref{lemma:predecessor}.
    
    We will construct the necessary runs to demonstrate this.

    Starting from $S$, let $s_1 \in S$. Since $(S, w, S) \in \Delta_{\cO_{\ICF}}^*$, Lemma~\ref{lemma:predecessor} guarantees a state $s_2$ such that $(s_2, w, s_1) \in \Delta^*$.
    We call $s_2$ a preceding state of $s_1$ over $w$.
    
    Define a sequence $s_1, s_2, \ldots, s_{|S|+1}$ of successive preceding states over $w$.
    By the pigeonhole principle, there are two indices $i < j$ such that $s_i = s_j$.
    
    \begin{figure}
        \centering
        \input{./figures/kDisambiguation/IDAproof1.tex}
        \caption{Illustration of cycle in $S$.}
        \label{fig:IDAproof1}
    \end{figure}
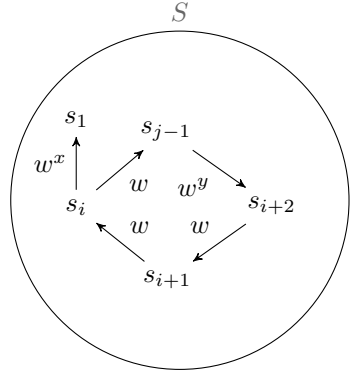

    Therefore lets define $s = s_i$, the following run exists for some integer $n$ (Figure~\ref{fig:IDAproof1}):
    \begin{align*}
        (s, w^n, s) \in \Delta^* \tag{1}\label{eq:fcf:run_1}
    \end{align*}

    Given $(R, u, T), (T, a, \hat{Q}), (\hat{Q}, v, S) \in \Delta_{\cO_{\ICF}}^*$, we can find states $r_1, t, q \in Q$ such that:
    
    \begin{align*}
        r_1 \ \xrightarrow{u} \ t \ \xrightarrow{a} \ q \ \xrightarrow{v} \ s \tag{2}\label{eq:fcf:run_2}
    \end{align*}
    
    For each state $r \in R$, there exist sequential preceding states $p^* \in P, t^* \in T, r^* \in R$ such that:

    \begin{align*}
        r^* \ \xrightarrow{u} \ t^* \ \xrightarrow{a} \ p^* \ \xrightarrow{v} \ r \tag{3}\label{eq:fcf:run_3}
    \end{align*}

    $r^*$ is a preceding state of $r$.
    We can then apply the same procedure used in $S$ to find a sequence of preceding states $r_1, r_2, \ldots, r_{|R|+1}$ and indices $i < j$ such that $r_i = r_j$.

    This yields the runs for some integer $k$ (Figure~\ref{fig:IDAproof2}): 

    \begin{align*}
        r_i \ \xrightarrow{w^{k}u} \ t \ \xrightarrow{a} \ q \ \xrightarrow{v}\ s \tag{4}\label{eq:fcf:run_4} \\
        r_i \ \xrightarrow{w^{m}} \ r_i \tag{5}\label{eq:fcf:run_5}
    \end{align*}
    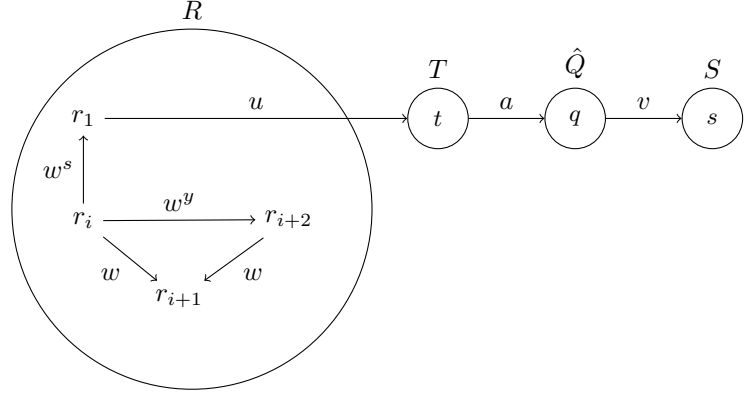
\begin{figure}
        \centering
        \input{./figures/kDisambiguation/IDAproof2.tex}
        \caption{Illustration of cycle in $R$ and the run from $R$ to $S$.}
        \label{fig:IDAproof2}
    \end{figure}

    Assume that $k+1 \leq n \cdot m$, if not we can replace run (\ref{eq:fcf:run_1}) to $(s, w^{n'}, s) \in \Delta^*$ where $n'$ is a multiple of $n$ and $k+1 \leq n' \cdot m$.

    Since we can extend the run (\ref{eq:fcf:run_4}) simply by adding the predecesor of $r_i$: 
    \[
        \begin{array}{rcl}
            r_{i+1} \ \xrightarrow{uav} \ r_i \ \xrightarrow{w^{k+1}} \ s
        \end{array}
    \]
    We can easily find a state $r_z$ with $z \in [i, j]$ where:
    \begin{align*}
        r_z \ \xrightarrow{w^{k'}u} \ t \ \xrightarrow{a} \ q \ \xrightarrow{v} \ s \tag{6}\label{eq:fcf:run_6} \\
        r_z \ \xrightarrow{w^m} \ r_z \tag{7}\label{eq:fcf:run_7}
    \end{align*}
    Such that $k'+1 = n \cdot m$.

    Let $r = r_z$.
    Additionally, we can extend runs (\ref{eq:fcf:run_2}) and (\ref{eq:fcf:run_7}) to make their words equal to $w^{n \cdot m}$ (Figure ~\ref{fig:IDAproof3}):
    \begin{align*}
        s \ \xrightarrow{w^{n \cdot m}} \ s \tag{8}\label{eq:fcf:run_8} \\
        r \ \xrightarrow{w^{n \cdot m}} \ r \tag{9}\label{eq:fcf:run_9}
    \end{align*}
    Since we extended using runs (\ref{eq:fcf:run_3}) then there exists a state $p \in P$ and a state $t_2 \in T$ such that:
    \begin{align*}
        r \ \xrightarrow{w^{n \cdot m - 1}u} \ t_2 \ \xrightarrow{a} \ p \ \xrightarrow{v} \ r \tag{10}\label{eq:fcf:run_10} \\
    \end{align*}
    With runs (\ref{eq:fcf:run_6}), (\ref{eq:fcf:run_8}) and (\ref{eq:fcf:run_10}) we show that $(p, q) \in \ICF_{\cA}$ (Figure ~\ref{fig:IDAproof4}) with words $u' = w^{n \cdot m-1}u$ and $v' = av$.

    This is a contradiction, since it implies that $\cA_{\cO_{\ICF}}$ satisfies the (IDA) condition and is a polyFA.
    When constructing the transitions from $T$ with letter $a$ in $\Delta_{\cO_{\ICF}}$, $p$ and $q$ are part of the same connected component in $\kappa(\ICF_{\cA} \ \bigcap \ (\Delta_{\cO_{\ICF}}(T, a) \times \Delta_{\cO_{\ICF}}(T, a)))$, therefore $P = \hat{Q}$, which contradicts the initial assumption.

    This concludes that $\cA_{\cO_{\ICF}}$ is a finFA.
    
     \begin{figure}
        \centering
        \input{./figures/kDisambiguation/IDAproof3.tex}
        \caption{Illustration of extended runs.}
        \label{fig:IDAproof3}
    \end{figure}
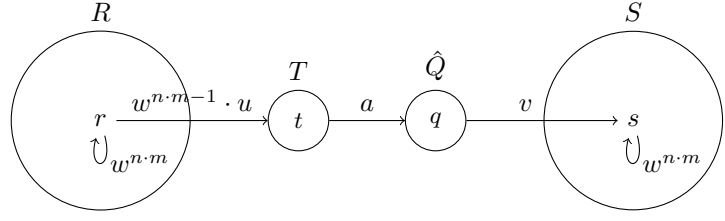

    \begin{figure}
        \centering
        \input{./figures/kDisambiguation/IDAproof4.tex}
        \caption{Illustration of states $p$ and $q$ having exponentially infinite common future.}
        \label{fig:IDAproof4}
    \end{figure}
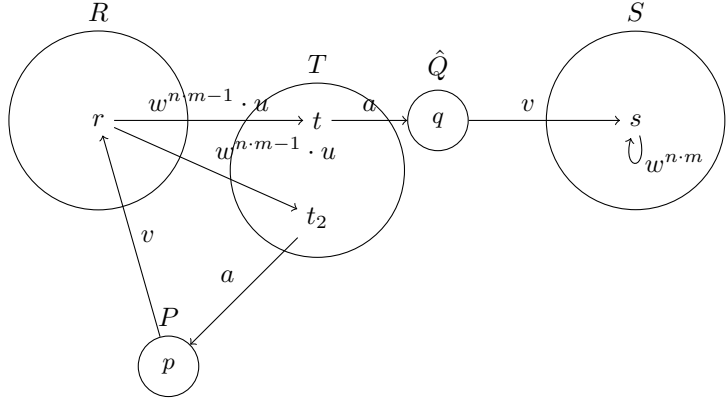

    Furthermore, if $\cA$ is already a finFA, then no two distinct states share an infinite common future.  
    Consequently, the oracle $\cO_{\ICF}$ partitions the state set into singletons.  
    By Proposition~\ref{prop:frameworkIsomorphic}, it follows that $\cA_{\cO_{\ICF}}$ is isomorphic to $\cA$.

\end{proof}
\begin{lemma}

	For every NFA $\cA = (Q, \Sigma, \Delta, I, F)$, the following properties hold. 
	\begin{enumerate}
		\item $\cA_{\cO_{\ECF}}$ is always polynomially ambiguous.
		\item if $\cA$ is polynomially ambiguous then $\cA_{\cO_{\ECF}}$ is isomorphic to $\cA$.
	\end{enumerate}

\end{lemma}

\begin{proof}
Suppose that the automaton $\cA_{\cO_{\ECF}}$ is not a polyFA and satisfies (EDA)', then there exist states $R, \hat{Q}, P \in 2^Q. \; \hat{Q} \neq P$, a letter $a$ and a word $w$ such that:

$$
    (R, a, P), (R, a, \hat{Q}), (P, w, R), (\hat{Q}, w, R) \in \Delta_{\cO_{\ECF}}^*
$$

Then by applying Lemma~\ref{lemma:predecessor}:
\begin{enumerate}
    \item For every state $p \in P \cup \hat{Q}$, it has a preceding state $r \in R$ such that $(r, a, p) \in \Delta$.
    \item For every state $r \in R$, it has preceding states $p \in P$ and $q \in \hat{Q}$ such that $(p, w, r), (q, w, r) \in \Delta^*$.
\end{enumerate}

Furthermore, these properties imply that any state $p \in R \cup P \cup \hat{Q}$ has at least one preceding state with a transition to $p$. This characteristic allows us to construct an infinite tree of states.

\paragraph{Infinite Run Tree} An \emph{infinite run tree} is a tuple $(t, T, q)$, where:

\begin{itemize}
    \item $t$ is an infinite set of nodes
    \item $T: t \mapsto Q$ is a labeling function that assigns states to nodes
    \item $q \in Q$ is the label of the root node, i.e., $T(k_0) = q$, where $k_0$ denotes the root at depth 0
\end{itemize}

The tree structure is defined recursively starting from the root as follows (Shown in Figure~\ref{fig:ECFtree}):

\begin{enumerate}
    \item If a node $k$ is at an even depth:
    \begin{itemize}
        \item $T(k) \in R$.
        \item $k$ has two children, $k_1$ and $k_2$, where:
        \begin{itemize}
            \item The left child $T(k_1) \in P$ is a preceding state of $T(k)$.
            \item The right child $T(k_2) \in \hat{Q}$ is a preceding state of $T(k)$.
        \end{itemize}
    \end{itemize}
    \item If a node $k$ is at an odd depth:
    \begin{itemize}
        \item $T(k) \in P \cup \hat{Q}$.
        \item $k$ has one child $k'$, where $T(k') \in R$ is a preceding state of $T(k)$.
    \end{itemize}
\end{enumerate}

The tree structure illustrates the infinite sequence of predecessors for a state in $R$, derived from transitions in $\cA$, highlighting the alternation between states in $R$ and $P \cup \hat{Q}$.
This construction enables us to identify pairs of states $p, q \in Q$ that share an exponentially infinite common future within the tree and $\cA$.

    Given a state $r_0 \in R$, we define an infinite run tree $(t, T, r_0)$ with root $k_0$ labeled $r_0$.

    Consider the depth level $2 \cdot \lceil \log_2(|S|) \rceil$ below the root. This level contains nodes $k_1, \ldots, k_K$ with labels $r_1, \ldots, r_K$, where $K \geq |R|$. By the pigeonhole principle, there exist two states $r_i, r_j \in Q$ that are equal.

    Let $k$ be the \emph{lowest common ancestor (LCA)} of $k_i$ and $k_j$ with label $r$. By the tree's structure, $k$ has an even depth and two children. Let $p \in P$ and $q \in \hat{Q}$ be the labels of these children.
    
    Define $A(k)$ as the set of labels of all ancestors of $k$ (including itself) with even depth.

    We consider two cases:

    \begin{enumerate}
        \item If $r_i \in A(k)$:
        We can find the following transitions in $\cA$:

        \begin{align*}
                r_i \ \xrightarrow{(aw)^{n_1}} \ r_1 \ \xrightarrow{a} \ p \ \xrightarrow{w(aw)^{n_2}} \ r_i  \\
                r_i \ \xrightarrow{(aw)^{n_1}} \ r_2 \ \xrightarrow{a} \ q \ \xrightarrow{w(aw)^{n_2}} \ r_i
        \end{align*}

    This implies $(p, q) \in \ECF_{\cA}$ (Figure~\ref{fig:ECFproof}). 
				
		This is a contradiction, since it leads to $\cA_{\cO_{\ECF}}$ being exponentially ambiguous.
    When constructing the transitions from $R$ that contains the states $r_1$ and $r_2$ on letter $a$ in $\Delta_{\cO_{\ECF}}$, $p$ and $q$ are part of the same connected component in $\kappa(\ECF_{\cA} \ \bigcap \ (\Delta_{\cO_{\ECF}}(R, a) \times \Delta_{\cO_{\ECF}}(R, a)))$, therefore $P = \hat{Q}$, which contradicts the initial assumption.
        
        \item If $r_i \notin A(k)$:
        
        We repeat the argument starting from node $k_i$. We find another pair of equal states $r_i', r_j'$ and examine the ancestors of their LCA $k'$.
        Note that $|A(k')| \geq |A(k)| + 1$, as $r_i$ is added to the set of ancestors.

        Therefore, if we don't find the state $r_i'$ among the ancestors, we can repeat the argument multiple times until the set of ancestors is equal to $Q$ for some $r_i^*$, where $r_i^*$ will be present among the ancestors. 
        In conclusion, the argument will always lead to case 1.
    
    \end{enumerate}

    Therefore, in both cases we will find that $P = \hat{Q}$, contradicting our initial assumption.
    This concludes that $\cA_{\cO_{\ECF}}$ is a polyFA.

    Furthermore, if $\cA$ is already a polyFA, then no two distinct states share an exponentially infinite common future.  
    Consequently, the oracle $\cO_{\ECF}$ partitions the state set into singletons.  
    By Proposition~\ref{prop:frameworkIsomorphic}, it follows that $\cA_{\cO_{\ECF}}$ is isomorphic to $\cA$.

\end{proof}

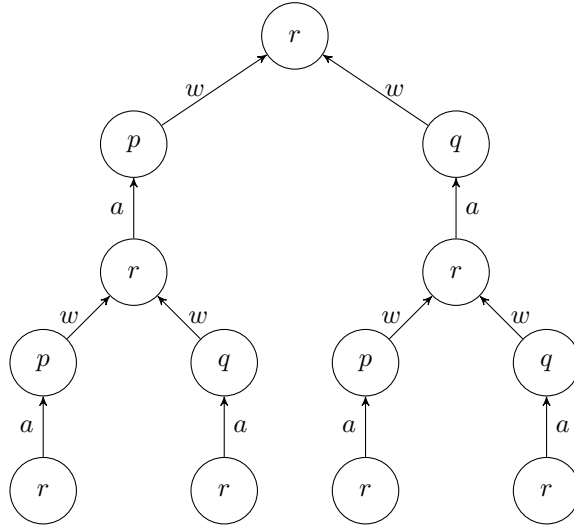
\begin{figure}
    \centering
    \input{./figures/kDisambiguation/ECFtree.tex}
    \caption{Illustration of an infinite run tree. States at even depths are labeled by $r \in R$, while states at odd depths are labeled alternately by $p \in P$ and $q \in \hat{Q}$. Edges are directed and labeled according to the corresponding transitions in $\Delta_{\cO_{\ECF}}^*$. The tree structure demonstrates the alternating pattern between states from $R$ and states from $P \cup \hat{Q}$.}
    \label{fig:ECFtree}
\end{figure}

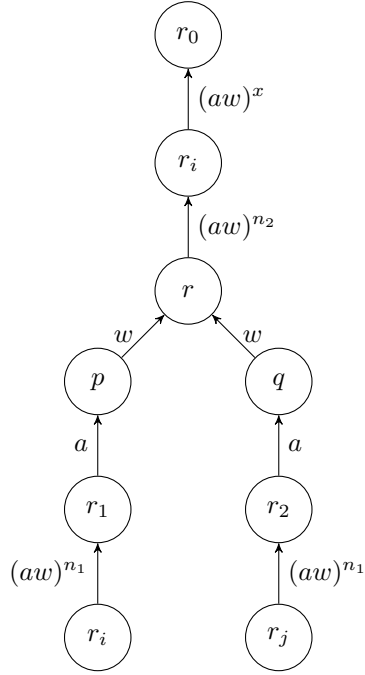
\begin{figure}
    \centering
    \input{./figures/kDisambiguation/ECFproof.tex}
    \caption{Illustration of the states used in the proof within the infinite run tree with $r_i = r_j$.}
    \label{fig:ECFproof}
\end{figure}

\subsection{Proof of Theorem~\ref{theo:minimalityICF-ECF}}

\minimalityICFECF*

We separate the proof in two parts.
\begin{lemma}
	Let $\cA$ be an NFA and $R$ be a reflexive and symmetric relation between states of~$\cA$. If
	$\cA_{\cO_R}$ is finitely ambiguous and $R \subseteq \ICF_{\cA}$, then $\cO_{\ICF}$ refines $\cO_R$ over $\cA$.\end{lemma}

\begin{proof}
    Let $R \subseteq \ICF_{\cA}$ be a relation whose corresponding oracle $\cO_{R}$ defines a finFA $\cA_{\cO_{R}} = (2^Q, \Sigma, \Delta_{\cO_{R}}, I', F')$.
    Assume, for contradiction, that the partition $\cO_{\ICF}$ does not refine $\cO_{R}$.

    If $\kappa(R) = \kappa(\ICF_{\cA})$, then $\cO_{\ICF} = \cO_{R}$.

    Otherwise, there exist states $p, q \in Q$ such that $(p, q) \in \ICF_{\cA}$ and for all subsets $S \in \{\Delta(S',a) \mid S' \in \reach(\cA_{\cO_{R}}) \wedge a \in \Sigma\}$ containing both $p$ and $q$ we have that they are contained in distinct sets $S_{p}, S_{q} \in \cO_{R}(S)$ with $p \in S_{p}$ and $q \in S_{q}$.

    By the definition of the (IDA)' property, there exist words $u,v \in \Sigma^*$ and states $s,r \in Q$ such that the following transitions hold in $\cA$:

    \[
    (r, u, p), \quad (r, u, q), \quad (p, v, r), \quad (q, v, s), \quad (s, uv, s) \in \Delta^*.
    \]

    We will use the notation for sets $S_z$, indicating that $z \in S_z$.

    Let $w := uv$. There exist a partial run in $\cA_{\cO_{R}}$ such that:

    \[
        \begin{array}{rcl}
            \rho^* & := &  S_r \ \xrightarrow{w} \ S'_{r} \\
        \end{array}
    \]

    Applying the pigeonhole principle to the iterated run over $w$, we find a state set $S_r^*$ and an integer $k > 0$ such that:
    \[
        \begin{array}{rcl}
             S_r^* \ \xrightarrow{w^k} \ S_{r}^* \\
        \end{array}
    \]

    Since $(r, w, s) \in \Delta$, we obtain an infinite sequence of states:
    \[
        \begin{array}{rcl}
             S_r^* \ \xrightarrow{w^k} \ S1^{+}_{s} \ \xrightarrow{w^k} \ S2^+_{s} \ \xrightarrow{w^k} \ldots  \\
        \end{array}
    \]

    Again, by pigeonhole principle, there exist integers $z_1, z_2$ and a set of states $S_s^{+}$ such that:
    \[
        \begin{array}{rcl}
             S_r^* \ \xrightarrow{w^{k \cdot z_1}} \ S^{+}_{s} \ \xrightarrow{w^{k \cdot z_2}} \ S^+_{s}  \\
        \end{array}
    \]

    Without loss of generality, assume $z_1 \leq z_2$ (otherwise, replace $z_2$ by one of its multiples).
    Since $(r,w,s) \in \Delta$, there exists a set of states $P_s$ such that:

    \[
        \begin{array}{rcl}
             S_r^* \ \xrightarrow{w^{k \cdot z_1}} \ S^{+}_{s} \ \xrightarrow{w^k} \ P_s  \ \xrightarrow{w^{k \cdot (z_2-1)}} \ S^+_{s}  \ \xrightarrow{w^k} \ P_s  \\
        \end{array}
    \]
    Iterating this argument extends $z_1$, obtaining a state set $S_s^*$ satisfying:

    \[
        \begin{array}{rcl}
            S_r^* \ \xrightarrow{w^{k \cdot z_2}} \ S^{*}_{s} \ \xrightarrow{w^{k \cdot z_2}} \ S^{*}_{s} \\
        \end{array}
    \]

    Because the runs $r \xrightarrow{u} p$ and $r \xrightarrow{u} q$ exist in $\cA$, there exist runs $\rho'$ and $\rho'_2$ such that:

    \[
        \begin{array}{rcl}
            \rho' & := & S_r^* \ \xrightarrow{u} \ S_{p} \ \xrightarrow{vw^{k \cdot z_2 - 1}} \ S_r^* \\
            \rho_2' & := & S_r^* \ \xrightarrow{u} \ S_{q} \ \xrightarrow{vw^{k \cdot z_2 - 1}} \ S^{*}_{s} \ \xrightarrow{w^{k \cdot z_2}} \ S^{*}_{s} \\
        \end{array}
    \]

    By the initial assumption, $S_{p} \neq S_{q}$, thus the pair $(S_{p}, S_{q})$ satisfies (IDA)', implying that $\cA_{\cO_{R}}$ is a polyFA.

    This contradicts the assumption that $\cA_{\cO_{R}}$ is a finFA, and therefore it cannot satisfy (IDA)'.
    Hence, no finer oracle $\cO_{R}$ exists, and $\cO_{\ICF}$ is minimal.   

\end{proof}

\begin{lemma}
	Let $\cA$ be an NFA and $R$ be a reflexive and symmetric relation between states of~$\cA$. If 
	$\cA_{\cO_R}$ is polynomially ambiguous and $R \subseteq \ECF_{\cA}$, then $\cO_{\ECF}$ refines $\cO_R$ over $\cA$.
\end{lemma}

\begin{proof}
    Let $R \subseteq \ECF_{\cA}$ be a relation whose corresponding oracle $\cO_{R}$ defines a polyFA $\cA_{\cO_{R}} = (2^Q, \Sigma, \Delta_{\cO_{R}}, I', F')$.
    Assume, for contradiction, that the partition $\cO_{\ECF}$ does not refine $\cO_{R}$.

    If $\kappa(R) = \kappa(\ECF_{\cA})$, then $\cO_{\ECF} = \cO_{R}$.

    Otherwise, there exist states $p, q \in Q$ such that $(p, q) \in \ECF_{\cA}$ and for all subsets $S \in \{\Delta(S',a) \mid S' \in \reach(\cA_{\cO_{R}}) \wedge a \in \Sigma\}$ containing both $p$ and $q$ we have that they are contained in distinct sets $S_{p}, S_{q} \in \cO_{R}(S)$ with $p \in S_{p}$ and $q \in S_{q}$.

    By definition of $\ECF_{\cA}$, there exists a word $w \in \Sigma^*$, a letter $a \in \Sigma$ and a state $r \in Q$ such that the following transitions hold in $\cA$.

    \[
    (r, a, p), \quad (r, a, q), \quad (p, w, r), \quad (q, w, r) \in \Delta^*.
    \]

    Moreover, by Lemma~\ref{lemma:predecessor}, the following hold:

    \begin{enumerate}
        \item For any state $S_p$ containing $p$, there exists a preceding state set $S_{r}$ such that $(S_{r}, a, S_p) \in \Delta_{\cO_{\ECF}}^*$ and $r \in S_{r}$.
        \item For any state $S_q$ containing $q$, there exists a preceding state set $S_{r}$ such that $(S_{r}, a, S_q) \in \Delta_{\cO_{\ECF}}^*$ and $r \in S_{r}$.
        \item For any state $S_{r}$ containing $r$, there exist preceding state sets $S_p$ and $S_q$ such that
        \[
            (S_p, w, S_{r}), \quad (S_q, w, S_{r}) \in \Delta_{\cO_{\ECF}}^*,
        \]
        with $p \in S_p$ and $q \in S_q$.
    \end{enumerate}

    Let $S_{r} \in \textsf{Reach}(\cA)$ be a reachable state set containing $r$. 
    Similarly to Proposition~\ref{theo:finFA-polyFA-disambiguation}, define an infinite run tree $(t, T, S_{r})$ with labeling function $T: t \mapsto 2^Q$ as follows:

    \begin{enumerate}
        \item If a node $k$ is at an even depth:
        \begin{itemize}
            \item $T(k)$ contains $r$.
            \item $k$ has two children, $k_1$ and $k_2$, where:
            \begin{itemize}
                \item $T(k_1)$ contains $p$ and is a preceding state set of $T(k)$.
                \item $T(k_2)$ contains $q$ and is a preceding state set of $T(k)$.
            \end{itemize}
        \end{itemize}
        \item If a node $k$ is at an odd depth:
        \begin{itemize}
            \item $T(k)$ contains either $p$ or $q$.
            \item $k$ has one child $k'$, such that $T(k')$ contains $r$ and is a preceding state set of $T(k)$.
        \end{itemize}
    \end{enumerate}

    Using the same argument as in Theorem~\ref{theo:finFA-polyFA-disambiguation}, there exist integers $n_1, n_2$ and states $S_p, S_q \in 2^Q$ such that:

    \begin{align*}
        S_{r} \ \xrightarrow{a(wa)^{n_1}} \ S_p \ \xrightarrow{w(aw)^{n_2}} \ S_{r}  \\
        S_{r} \ \xrightarrow{a(wa)^{n_1}} \ S_q \ \xrightarrow{w(aw)^{n_2}} \ S_{r} \\
    \end{align*}

    By the initial assumption, $S_p \neq S_q$, thus the pair $(S_p, S_q)$ satisfies (EDA)', implying that $\cA_{\cO_{R}}$ is expFA.

    This contradicts the assumption that $\cA_{\cO_{R}}$ is a polyFA, and therefore it cannot satisfy (EDA)'.
    Hence, no finer oracle $\cO_{R}$ exists, and $\cO_{\ECF}$ is minimal.   

\end{proof}

\subsection{Minimality counterexample}

Theorem~\ref{theo:UFA-minimality} does not hold for relations $\ICF$ and $\ECF$.
Consider the NFA $\cA$ in Figure~\ref{fig:edaExamples} and its polynomial disambiguation $\cA_{\cO_{\ECF}}$, its exponentially infinite common future relation is $\ECF_{\cA} = \{(p, q), (q, p), (p, z), (z, p)\}$.
There exist a relation $R=\{(p, z), (q, z), (z, p), (z, q)\}$ such that $\cO_{R}(\{p, q\})$ refines $\cO_{\ECF}(\{p, q\})$, while $\cA_{\cO_{R}}$ is polynomially ambiguous.

In $\cA_{\cO_{R}}$, states $p$ and $q$ are only merged in the set of states that contains $z$, effectively preventing exponetial degree of ambiguity.

\begin{figure}
\vspace{5em}
    \hspace{-0.5em}
    \subfloat[$\cA$]{%
        \input{./figures/EDAexample/original.tex}%
    }
    \hspace{0.5em}
    \begin{minipage}{25mm}
    \vspace{-10em}
    \subfloat[$\cA_{\cO_{\ECF}}$]{%
        \input{./figures/EDAexample/dis.tex}%
    }
    \vspace{0.5em}
    \subfloat[$\cA_{\cO_{R}}$]{%
        \input{./figures/EDAexample/alternative.tex}%
    }
    \end{minipage}
    
    \caption{counterexample for minimality of $\ECF$.}
    \label{fig:edaExamples}
\end{figure}
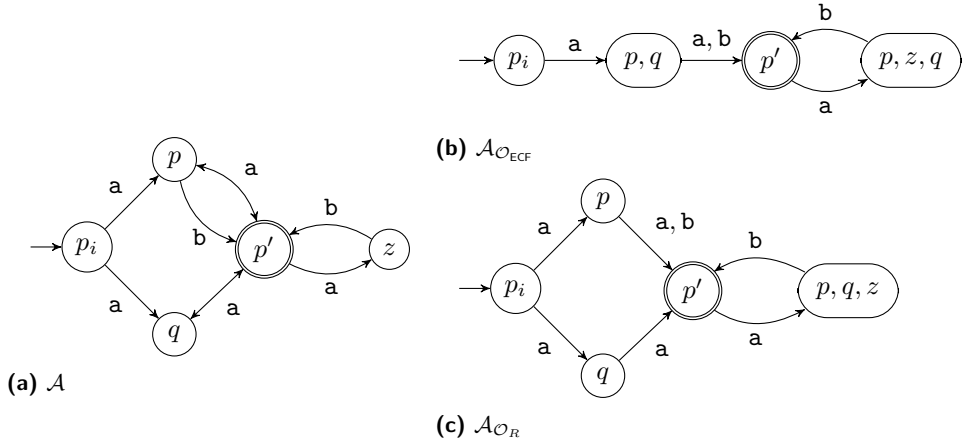

To determine that states $p$ and $q$ should not be merged, the procedure needs knowledge of their future behavior on $\cA_{\cO_{R}}$  to verify they are not part of an (EDA) condition.
Consequently, no disambiguation procedure that cannot be computed on-the-fly can refine the states formed by oracle $\ECF$ and define a polynomially ambiguous automaton.

Note that in the example, the relation $\ICF_{\cA}$ is equivalent to $\ECF_{\cA}$, then the same NFA serves as a counter example for the minimality of $\ICF$.

%% file: figures/kDisambiguation/IDA.tex
\begin{tikzpicture}[every label/.style={black!60}, initial text=, ->,>=stealth', node distance=4cm]
  \node[state]          (p)                {$r$};
  \node[state]          (q)    [right of=p]            {$s$};
  
  \path[->] (p) edge [loop above] node [above]   {$uav$} (p)
            (p) edge  node [above]   {$uav$} (q)
            (q) edge [loop above]  node [above]   {$uav$} (q);
\end{tikzpicture}

%% file: figures/kDisambiguation/IDAalternative.tex
\begin{tikzpicture}[every label/.style={black!60}, initial text=, ->,>=stealth', node distance=2cm]
  \node[state]          (p)                {$r$};
  \node[state]          (r)    [right of=p]            {$t$};
  \node[state]          (p')    [above of=r]            {$p$};
  \node[state]          (q')    [right of=r]            {$q$};
  \node[state]          (q)    [right of=q]            {$s$};
  
  \path[->] (p) edge  node [above]   {$u$} (r)
            (r) edge  node [below right]   {$a$} (p')
            (r) edge  node [above]   {$a$} (q')
            (p') edge  node [above]   {$v$} (p)
            (q') edge  node [above]   {$v$} (q)
            (q) edge [loop above]  node [above]   {$uav$} (q);
\end{tikzpicture}

%% file: figures/kDisambiguation/EDA.tex
\begin{tikzpicture}[every label/.style={black!60}, initial text=, ->,>=stealth', node distance=1cm]
  \node[state]          (p)                {$s$};
  
  \path[->] (p) edge [loop above] node [above]   {$uav$} (p)
            (p) edge [loop below]  node [below]   {$uav$} (p);
\end{tikzpicture}

%% file: figures/kDisambiguation/EDAintermediate.tex
\begin{tikzpicture}[every label/.style={black!60}, initial text=, ->,>=stealth', node distance=1cm]
  \node[state]          (p)                {$s$};
  \node[state]          (p') [right=of p] {$r$};
  \node[state]          (q_1) [above right=of p'] {$p$};
  \node[state]          (q_2) [below right=of p'] {$q$};
  
  \path[->] (p) edge                 node [above]   {$u$} (p')
            (p') edge                node [above left]   {$a$} (q_1)
            (p') edge                node [below left]   {$a$} (q_2)
            (q_1) edge [bend right]  node [above]   {$v$} (p)
            (q_2) edge [bend left]  node [below]   {$v$} (p);
\end{tikzpicture}

%% file: figures/kDisambiguation/EDAalternative.tex
\begin{tikzpicture}[every label/.style={black!60}, initial text=, ->,>=stealth', node distance=1cm]
  \node[state]          (p')                {$r$};
  \node[state]          (q_1) [above right=of p'] {$p$};
  \node[state]          (q_2) [below right=of p'] {$q$};
  
  \path[->] (p') edge  [bend right]    node [below right]   {$a$} (q_1)
            (p') edge  [bend left]     node [above right]   {$a$} (q_2)
            (q_1) edge [bend right]    node [above left]   {$v \cdot u$} (p')
            (q_2) edge [bend left]     node [below left]   {$v \cdot u$} (p');
\end{tikzpicture}

%% file: figures/kDisambiguation/IDAproof1.tex
\begin{tikzpicture}[every label/.style={black!60}, initial text=, ->,>=stealth', node distance=0.7cm]

  \node (s1)                       {$s_1$};
  \node (si)  [below =of s1]   {$s_i$};
  \node (si1) [below right=of si]  {$s_{i+1}$};
  \node (si2) [above right=of si1]       {$s_{i+2}$};
  \node (sj1) [above left=of si2] {$s_{j-1}$};

  \node[draw,circle,fit=(s1)(si)(si1)(si2)(sj1),inner sep=0.1cm,label=above:$S$] {};

  \path (si)  edge                 node [left]       {$w^x$} (s1)
        (si1)  edge                 node [above right] {$w$}   (si)
        (si2) edge                 node [above left]      {$w$}   (si1)
        (sj1) edge         node [below left]{$w^y$} (si2)
        (si) edge         node  [below right] {$w$} (sj1);

\end{tikzpicture}

%% file: figures/kDisambiguation/IDAproof2.tex
\tikzset{
  nomark/.style   = {inner sep=0pt,font=\small},            
  state/.style    = {draw,circle,minimum size=8mm,font=\small},
  lab/.style      = {font=\scriptsize,inner sep=1pt},
}

\begin{tikzpicture}[->]

\node (r1)                       {$r_1$};
\node (ri)  [below =0.9cm of r1] {$r_i$};
\node (ri1) [below right =0.8cm of ri] {$r_{i+1}$};
\node (ri2) [above right =0.8cm of ri1] {$r_{i+2}$};

\node[draw,circle,fit=(r1)(ri)(ri1)(ri2),
      inner sep=3pt,label=above:$R$] (R) {};

\node[state] (t) [right=4cm of r1] {$t$};
\node[state] (q) [right=of t]        {$q$};
\node[state] (s) [right=of q]        {$s$};

\node[above=0pt of t] {$T$};
\node[above=0pt of q] {$\hat{Q}$};
\node[above=0pt of s] {$S$};

\path (ri)  edge[              ] node[left]        {$w^{s}$} (r1)    
      (ri)  edge[        ] node[above] {$w^y$    } (ri2)   
      (ri2) edge[              ] node[below right ] {$w$    } (ri1)   
      (ri)  edge[              ] node[below left ] {$w$   } (ri1);

\path (r1) edge node[above] {$u$} (t)
              (q)  edge node[above] {$v$} (s);
\path         (t)  edge node[above] {$a$} (q);
\end{tikzpicture}

%% file: figures/kDisambiguation/IDAproof3.tex
\tikzset{
  nomark/.style   = {inner sep=0pt,font=\small},            
  state/.style    = {draw,circle,minimum size=8mm,font=\small},
  lab/.style      = {font=\scriptsize,inner sep=1pt},
}

\begin{tikzpicture}[->]

\node (r1)                       {$r$};

\node[draw,circle,fit=(r1),
      inner sep=18pt,label=above:$R$] (R) {};

\node[state] (t) [right=2cm of r1] {$t$};
\node[state] (q) [right=of t]        {$q$};
\node (s) [right=2cm of q]        {$s$};

\node[draw,circle,fit=(s),
      inner sep=18pt,label=above:$S$] ( S) {};

\node[above=0pt of t] {$T$};
\node[above=0pt of q] {$\hat{Q}$};

\path (r1) edge node[above] {$w^{n \cdot m - 1} \cdot u$} (t)
              (r1) edge [loop below] node[right=] {$w^{n \cdot m}$} (r1)
              (s) edge [loop below] node[right] {$w^{n \cdot m}$} (s)
              (q)  edge node[above left] {$v$} (s);
\path         (t)  edge node[above] {$a$} (q);
\end{tikzpicture}

%% file: figures/kDisambiguation/IDAproof4.tex
\tikzset{        
  state/.style    = {draw,circle,minimum size=8mm,font=\small},
}

\begin{tikzpicture}[->]

\node (r1)                       {$r$};

\node[draw,circle,fit=(r1),
      inner sep=18pt,label=above:$R$] (R) {};

\node (t) [right=2.5cm of r1] {$t$};

\node (t_2) [below=0.8cm of t] {$t_2$};
\node[state] (q) [right=of t]        {$q$};
\node[state] (p) [below left=2cm of t_2] {$p$};
\node (s) [right=2cm of q]        {$s$};

\node[draw,circle,fit=(s),
      inner sep=18pt,label=above:$S$] ( S) {};

\node[draw,circle,fit=(t)(t_2),
      inner sep=5pt,label=above:$T$] ( T) {};

\node[above=0pt of q] {$\hat{Q}$};
\node[above=0pt of p] {$P$};

\path (r1) edge node[above] {$w^{n \cdot m - 1} \cdot u$} (t)
              (r1) edge node[above right] {$w^{n \cdot m - 1} \cdot u$} (t_2)
              (s) edge [loop below] node[right] {$w^{n \cdot m}$} (s)
              (q)  edge node[above left] {$v$} (s)
              (p)  edge node[right] {$v$} (r1);
\path         (t)  edge node[above] {$a$} (q)
              (t_2)  edge node[above left] {$a$} (p);
\end{tikzpicture}

%% file: figures/kDisambiguation/ECFtree.tex
\begin{tikzpicture}[every label/.style={black!60}, initial text=, ->,>=stealth', node distance=0.8cm]

  \node[state] (q0) {$r$};
  \node[state] (p1) [below left=0.8cm and 1.5cm of q0] {$p$};
  \node[state] (p2) [below right=0.8cm and 1.5cm of q0] {$q$};

  \node[state] (q1) [below=of p1] {$r$};
  \node[state] (q2) [below=of p2] {$r$};

  \node[state] (p3) [below left= of q1] {$p$};
  \node[state] (p4) [below right= of q1] {$q$};

  \node[state] (p5) [below left= of q2] {$p$};
  \node[state] (p6) [below right= of q2] {$q$};

  \node[state] (q3) [below=of p3] {$r$};
  \node[state] (q4) [below=of p4] {$r$};

  \node[state] (q5) [below=of p5] {$r$};
  \node[state] (q6) [below=of p6] {$r$};

  \path[<-] (q0) edge node[left] {$w$} (p1)
            (q0) edge node[right] {$w$} (p2);

  \path[<-] (p1) edge node[left] {$a$} (q1)
            (p2) edge node[right] {$a$} (q2);

  \path[<-] (q1) edge node[left] {$w$} (p3)
            (q1) edge node[right] {$w$} (p4);

  \path[<-] (q2) edge node[left] {$w$} (p5)
            (q2) edge node[right] {$w$} (p6);

  \path[<-] (p3) edge node[left] {$a$} (q3)
            (p4) edge node[right] {$a$} (q4)
            (p5) edge node[left] {$a$} (q5)
            (p6) edge node[right] {$a$} (q6);

\end{tikzpicture}

%% file: figures/kDisambiguation/ECFproof.tex
\begin{tikzpicture}[every label/.style={black!60}, initial text=, ->,>=stealth', node distance=0.8cm]

  \node[state] (q0) {$r_0$};
  \node[state] (qi) [below=of q0] {$r_i$};
  \node[state] (qf) [below=of qi] {$r$};

  \node[state] (p1) [below left=of qf] {$p$};
  \node[state] (p2) [below right=of qf] {$q$};

  \node[state] (r1) [below =of p1] {$r_1$};
  \node[state] (r2) [below =of p2] {$r_2$};

  \node[state] (qi1) [below=of r1] {$r_{i}$};
  \node[state] (qj1) [below=of r2] {$r_{j}$};

  \path[<-] (q0) edge node[right] {$(aw)^x$} (qi)
            (qi) edge node[right] {$(aw)^{n_2}$} (qf);

  \path[<-] (qf) edge node[left] {$w$} (p1)
            (qf) edge node[right] {$w$} (p2);

  \path[<-] (r1) edge node[left] {$(aw)^{n_1}$} (qi1)
            (r2) edge node[right] {$(aw)^{n_1}$} (qj1)
            (p1) edge node[left] {$a$} (r1)
            (p2) edge node[right] {$a$} (r2);

\end{tikzpicture}

%% file: sections/app-weighted.tex
We present several propositions and definitions that will be used in the subsequent proofs.

We denote by $\Delta^*$ the transitive closure of $\Delta$, representing word transitions.

\begin{proposition}\label{prop:wfaFrameworkIsomorphic}
  Given a WFA $\cW \ = \ (Q, \Sigma, \Delta, I, F)$ over $\mathbb{S}$ and ($\fac$, $\res$) be an identity factorization.
  Let $\cF = (\Pi, \fac, \res)$ be a PF-oracle where $\Pi$ have disjoint vector support.
  If the support of every vector in $\cW_{\cF} \ = \ (\bbS^Q, \Sigma, \Delta_{\cF}, I_\cF, F_{\cF})$ is a singleton, then $\cW_{\cF}$ is isomorphic to $\cW$.

\end{proposition}

\begin{proof}

Define, for each state $p \in Q$, the zero-one vector $V_p \in \mathbb{S}^Q$ such that:
\[
    V_p[q] = 
    \begin{cases}
    \overline{1} & \text{if } q = p, \\
    \overline{0} & \text{otherwise},
    \end{cases}
\]

We prove the isomorphism by induction on the length of input words $w$ using the bijection $f: Q \rightarrow \mathbb{S^Q}$ such that $f(q) = q \; \forall q \in Q$. 

\textbf{Base case ($w = \varepsilon$):}  
By definition of $I_{\cF}$, for every initial state $p \in \supp(\vI)$ of $\cW$ there exists a unique vector $V_p \in \supp(I_{\cF})$ with singleton support $\{p\}$.

\textbf{Inductive step:}  
Assume that for any word $w$ of length $n$, every vector reachable by reading $w$ in $\cW_{\cF}$ corresponds to some vector $V_p$ for a state $p \in Q$ with singleton support $\{p\}$.

Consider a letter $a \in \Sigma$ and a vector $V_p$ reachable by $w$ to some state $p$.  

By the definition of the transition relation $\Delta_{\cF}$ in the disambiguation construction, the set of successor vectors after reading $a$ from $V_p$ is
\[
\vec{\Delta_{\cF}}(V_p, a) = \{ \fac(V) \mid V \in \Pi((V_p^t \odot \Delta_a)^t) \}
\]

Note that $(V_p^t \odot \Delta_a)^t$ corresponds to the weight vector of all possible successor states of $p$ via $a$ in $\cW$. Since $V_p$ has singleton support $\{p\}$, the support of $(V_p^t \odot \Delta_a)^t$ matches the set $\{q \mid (s,q) \in \Delta(p, a)\}$.

Because $\Pi$ has disjoint vector support and $\cW_{\cF}$ has only vectors with singleton support by hypothesis, the partition $\Pi((V_p^t \odot \Delta_a)^t)$ splits $(V_p^t \odot \Delta_a)^t$ into a set of vectors, each supported on a unique state $q$ reachable from $p$ by $a$. 
Applying the identity factorization, we obtain: 

\[
\vec{\Delta_{\cF}}(V_p, a) = \{ V_q \mid (q, s) \in \Delta(p, a) \}
\]

Hence, for each transition $(p, a, s, q) \in \Delta$, there is a unique vector $V_q$ reachable by $wa$ and a corresponding transition $(V_p, a, s, V_q) \in \Delta_{\cF}$.

This establishes a one-to-one correspondence between states and transitions of $\cW_{\cF}$ and $\cW$, proving that $\cW_{\cF}$ is isomorphic to $\cW$.
\end{proof}

\begin{lemma}\label{lem:wfaDeltaFinite}
  Given a WFA $\cW = (Q, \Sigma, \Delta, I, F)$ over $\bbS$ and ($\fac$, $\res$) be any factorization.
  Let $\cF = (\Pi, \fac, \res)$ where $\Pi$ have disjoint vector support and $\cW_{\cF} = (\bbS, \Sigma, \Delta_{\cF}, I_\cF, F_{\cF})$.
  For all words $w$, initial states $V_i \in \supp(I_{\cF})$, $|\vec{\Delta}(V_i, w)|$ is finite. 
\end{lemma}

\begin{proof}
	By induction on $|w|$. For the empty word $\varepsilon$, $\vec{\Delta}(V_i, \varepsilon)$ is finite since $\Pi$ has disjoint vector support.  
	Assuming $\vec{\Delta}(V_i, w)$ is finite for a word $w$, then
	\[
	|\vec{\Delta}(V_i, wa)| \leq \sum_{V \in \vec{\Delta}(V_i, w)} |\Pi((V^t \odot \Delta_a)^t)|
	\]
	is a finite union of finite sets, therefore it is finite.
\end{proof}

We extend the definition of $\runs_{\cW}$ by setting $\runs_{\cW}(w,p)$ to be the set of all runs that end in $p$ after reading $w$, and $\runs_{\cW}(q,w,p)$ to be the set of all partial runs that start in $q$ and end in $p$ after reading $w$.

\paragraph{Victorious run} A run $\rho := p_0 \xrightarrow{w, v} p$ is \emph{victorious} if $v = \Delta_w[p_0, p]$.
For all tropical WFA $\cW$, it is known that for all words $w$ and states $p_0, p \in Q$, if $\Delta_w[p_0, p] \neq 0$, there exists a victorious run in $\runs_{\cW}(p_0, w, p)$ (see~\cite{KirstenOnDeterminizationWeighted} for example).
We extend this definition for $v = \Delta_w[I, p]$, there always exists a victorious run in $\runs_{\cW}(w, p)$.

Given a partial run $\rho$ we define $\omega_I(\rho) := I(p_0) \odot s_1 \odot \cdots \odot s_n$ as the weight of $\rho$ without the final function.

\begin{proposition}\label{prop:wfaResidual}
  Given a tropical WFA $\cW \ = \ (Q, \Sigma, \Delta, I, F)$.
  Let $\cF = (\Pi, \facMohri, \resMohri)$ where $\Pi$ have disjoint vector support and $\cW_{\cF} = (\bbT, \Sigma, \Delta_{\cF}, I_\cF, F_{\cF})$.
  For every word $w$ and vector $V \in \bbT^Q$ reachable by $w$, and for every run $\rho' := V_0 \xrightarrow{w/ v} V$:
\end{proposition}

$V[q] = \min_{q_i \in \supp(\vI)}(I(q_i)+ \Delta_w[q_i, q]) - v - I_{\cF}(V_0)$

\begin{proof}
	The following holds by induction on the length of \( w \):

	Base case with $w = \varepsilon $, $q_i = q$ and $V_0 = V$:
	\begin{align*}
			(I(q) + 0) - 0 - I_{\cF}(V) 
	\end{align*}

  There exists $V^* \in \Pi(\vI)$ such that $V = \res(V^*)$ and $V^*[q] = I(q)$ since $\Pi$ has disjoint vector support:
\[
			V^*[q] - I_{\cF}(\res(V^*)) = V^*[q] - \fac(V^*) = V[q]
\]

	By induction, suppose the proposition holds for $|w| \leq n$. 
	Let $|w \cdot a| = n+1$ and $V \in \vec{\Delta}(I, w)$ with a run $\rho' := V_0 \xrightarrow{w, v'} V' \xrightarrow{a, v} V$, then:

	For \( V \in \Pi(I) \),
	\begin{align*}
  &\!\!\!\!\!\!\!\!\!\!\!\!\!\!\!\!\!\!\min_{q_i \in \supp(\vI)}(I(q_i)+ \Delta_{wa}[q_i, q]) - v' - v - I_{\cF}(V_0)\\
  &=\min_{q_i\in \supp(\vI)}\min_{q' \in \supp(V')}(I(q_i)+ \Delta_{w}[q_i, q'] + \Delta_a[q', q]) - v' -v - I_{\cF}(V_0),\\
  &=\min_{q' \in \supp(V')} (\Delta_a[q', q] + \min_{q_i \in \supp(\vI)}(I(q_i)+ \Delta_{w}[q_i, q'])) - v' -v - I_{\cF}(V_0),\\
  &=\min_{q' \in \supp(V')} (\Delta_a[q', q] + \min_{q_i \in \supp(\vI)}(I(q_i)+ \Delta_{w}[q_i, q']) -v' - I_{\cF}(V_0)) - v.\\
  \intertext{By inductive hypothesis:}
  &=\min_{q' \in \supp(V')} (\Delta_a[q', q] + V'(q')) - v.\\
  \intertext{By definition of $\Delta_a$:}
  &=(V' + \Delta_a)[q] - v.\\
  \intertext{Because $(V', a, v, V) \in \Delta_{\cF}$, we know that $V = \res(V' + \Delta_a) = (V' + \Delta_a) - v$:}
  &=V[q].
	\end{align*}
\end{proof}

\begin{corollary}\label{cor:wfaResidualVictorious}
	For every words $w$ and vector $V \in \bbT^Q$ reachable by $w$, and for every run $\rho' := V_0 \xrightarrow{w, v} V$ and the victorious run $\rho$ from $I$ to $q$:
	$V[q]$ = $\omega_{I}(\rho) - v - I_{\cF}(V_0)$
\end{corollary}

\subsection{Proof of Proposition~\ref{prop:wfaFrameworkEquivalence}}

\wfaFrameworkEquivalence*

\begin{proof}
  Given a WFA $\cW = (Q, \Sigma, \Delta, I, F)$ over $\bbS$ and a PF-oracle $\cF = (\Pi, \fac, \res)$, we define $\cW_{\cF} \ = \ (\bbS^Q, \Sigma, \Delta_{\cF}, I_\cF, F_{\cF})$.

We will use the following result to prove the equivalence for all $w \in \Sigma^*$:
	\[
			\vI^t \odot \Delta_w =  \bigoplus_{V_i \in \supp(I_{\cF})} \bigoplus_{V \in \vec{\Delta_\cF}(V_i, w)} I_{\cF}(V_i) \odot (\Delta_{\cF})_w[V_i, V] \odot V^t
	\]
We will prove this result by induction. Let $w = \eps$:
	\begin{align*}
		\vI^t &= \bigoplus_{V_i \in \supp(I_{\cF})} \bigoplus_{V \in \vec{\Delta_\cF}(V_i, \eps)} I_{\cF}(V_i) \odot V^t \\
\intertext{Since $V = V_i$:} %
		&= \bigoplus_{V_i \in \supp(I_{\cF})}  I_{\cF}(V_i) \odot V_i^t \\
\intertext{Because for every vector $V_i \in \supp(I_{\cF})$ there exist a vector $V \in \Pi(\vec{I})$ such that $V_i = \res(V)$:}%
		&= \bigoplus_{V \in \Pi(\vI)}  (I_{\cF}(\res(V)) \odot \res(V))^t \\
		&= \bigoplus_{V \in \Pi(\vI)}  (\fac(V) \odot \res(V))^t\\
		&= \bigoplus_{V \in \Pi(\vI)}  V^t \\
\intertext{Since $\bigoplus_{V' \in \Pi(V)} V' = V$ for every $V \in \bbS^Q$:} %
		\bigoplus_{V \in \Pi(\vI)}  V^t &= \vI^t %
	\end{align*}
	By induction, suppose the result holds for $|w| \leq n$. Let $|w| = n$, then:
	\begin{align*}
		\vI^t \odot \Delta_{wa} &= \bigoplus_{V_i \in \supp(I_{\cF})} \bigoplus_{V \in \vec{\Delta_{\cF}}(V_i, wa)} I_{\cF}(V_i) \odot (\Delta_{\cF})_{wa}[V_i, V] \odot V^t  \\
		 &= \bigoplus_{V_i \in \supp(I_{\cF})} \bigoplus_{V' \in \vec{\Delta_{\cF}}(V_i, w)} \bigoplus_{V^* \in \Pi(\vec{\Delta}(V', a))}  I_{\cF}(V_i) \odot (\Delta_{\cF})_w[V_i, V'] \odot (\fac(V^*) \odot \res(V^*))^t, \\
		 &= \bigoplus_{V_i \in \supp(I_{\cF})} \bigoplus_{V' \in \vec{\Delta_{\cF}}(V_i, w)} \bigoplus_{V^* \in \Pi(\vec{\Delta}(V', a))}  I_{\cF}(V_i) \odot (\Delta_{\cF})_w[V_i, V'] \odot (V^*)^t. \\
    \intertext{Since $\bigoplus_{V' \in \Pi(V)} V' = V$ for every $V \in \bbS^Q$:} %
		 &= \bigoplus_{V_i \in \supp(I_{\cF})} \bigoplus_{V' \in \vec{\Delta_{\cF}}(V_i, w)} I_{\cF}(V_i) \odot (\Delta_{\cF})_w[V_i, V'] \odot \vec{\Delta}(V', a)^t, \\
		 &= \bigoplus_{V_i \in \supp(I_{\cF})} \bigoplus_{V' \in \vec{\Delta_{\cF}}(V_i, w)} I_{\cF}(V_i) \odot (\Delta_{\cF})_w[V_i, V'] \odot V'^t \odot \Delta_a, \\
\intertext{By distributability of $\odot$ over $\oplus$:} %
		&= \left( \bigoplus_{V_i \in \supp(I_{\cF})} \bigoplus_{V' \in \vec{\Delta_{\cF}}(V_i, w)} I_{\cF}(V_i) \odot (\Delta_{\cF})_w[V_i, V'] \odot V'^t \right) \odot  \Delta_a. \\
		\intertext{By inductive hypothesis:} %
		&= \vI^t \odot \Delta_w \odot  \Delta_a = \vI^t \odot \Delta_{wa}. %
	\end{align*}

	Finally, we can prove the equivalence of $\cW$ and $\cW_{\cF}$ for all $w \in \Sigma^*$:

	\begin{align*}
		\sem{\cW_{\cF}}(w)&= \bigoplus_{V_i \in \supp(I_{\cF})} \bigoplus_{V \in \vec{\Delta_{\cF}}(V_i, w)} I_{\cF}(V_i) \odot (\Delta_{\cF})_w[V_i, V] \odot F_{\cF}(V). \\		
		\intertext{By definition of $F_{\cF}$:} 
		&= \bigoplus_{V_i \in \supp(I_{\cF})} \bigoplus_{V \in \vec{\Delta_{\cF}}(V_i, w)} I_{\cF}(V_i) \odot (\Delta_{\cF})_w[V_i, V] \odot V \odot \vec{F}. \\
		\intertext{Using the previous result:} 
		&= \vI^t \odot \Delta_{w} \odot \vec{F},\\
		&= \sem{\cW}(w).
	\end{align*}\end{proof}

\paragraph{On degree of ambiguity conditions} 
The (IDA)' and (EDA)' conditions can also be used to characterize the degree of ambiguity of a WFA, since this depends only on the number of accepting runs, while the weights do not matter.
For the following theorem, we will adapt those conditions with weights.

\subsection{Proof of Theorem~\ref{theo:WFA-disambiguation}}

\WFAdisambiguation*

We start by proving the correctness of the disambiguation construction for each ambiguity class and then we will prove the isomorphism property.

\subsubsection{Case UWFA}

\begin{proof}
  Given a WFA $\cW = (Q, \Sigma, \Delta, I, F)$ over $\bbS$ and any factorization ($\fac$, $\res$).
  Let $\cF = (\Pi_{\cO_{\CF}}, \fac, \res)$ and $\cW_{\cF} \ = \ (\bbS^Q, \Sigma, \Delta_{\cF}, I_\cF, F_{\cF})$.
 
	Assume that $\cW_{\cF}$ is finite.

  Suppose, for contradiction, that $\cW_{\cF}$ is not unambiguous.
	Then there exist two different accepting runs $\rho_1$ and $\rho_2$ over the same word $w= a_1a_2\ldots a_n$ such that:

	\[
		\begin{array}{rcl}
			\rho_1 \ & := & \ V_0 \, \xrightarrow{a_1/ v_1} \ldots \xrightarrow{a_{i-1}/ v_{i-1}} V_{i-1} \xrightarrow{a_i/ v_i} V_i \xrightarrow{a_{i+1}/ v_{i+1}} \ldots \xrightarrow{a_n/ v_n} \, V_n 	\\
			\rho_2 \ & := & \ V_0' \, \xrightarrow{a_1/ v_1'} \ldots \xrightarrow{a_{i-1}/ v_{i-1}'} V_{i-1}' \xrightarrow{a_i/ v_i'} V_i' \xrightarrow{a_{i+1}/ v_{i+1}'}  \ldots \xrightarrow{a_n/ v_n'} \, V_n' 	
		\end{array}	
	\]

	Let $i$ be the largest index such that $V_j = V_j'$ for all $j < i$, and $V_i \neq V_i'$.
	Define $S_i = \textsf{supp}(V_i)$ and $S_i' = \textsf{supp}(V_i')$.

	\textbf{Case 1:} $S_i = S_i'$. Since $\Pi_{\cO_{\CF}}$ has disjoint vector support then $V_i = V_i'$, and it contradicts the assumption that $V_i \neq V_i'$.

  \textbf{Case 2:} $S_i \neq S_i'$. This case is analogous to the proof of Theorem~\ref{theo:UFA-disambiguation}, and similarly implies that $V_i = V_i'$, again contradicting the assumption.

	Therefore $\cW_{\cF}$ must be unambiguous.

\end{proof}

\subsubsection{Case finWFA}
\begin{proof}
  Given a WFA $\cW = (Q, \Sigma, \Delta, I, F)$ over $\bbS$ and any factorization ($\fac$, $\res$).
  Let $\cF = (\Pi_{\cO_{\ICF}}, \fac, \res)$ and $\cW_{\cF} \ = \ (\bbS^Q, \Sigma, \Delta_{\cF}, I_\cF, F_{\cF})$.
 
	Assume that $\cW_{\cF}$ is finite.

	Suppose that the WFA $\cW_{\cO_{\ICF}}$ is not finWFA and satisfies (IDA)', then there exist vectors $R, P, T, \hat{Q}, S \in \bbS^Q. \; P \neq \hat{Q}$, weights $s_1, s_2, s_3, s_4, s_5, s_6 \in \bbS$, a letter $a$ and words $u, v$ such that:

    \[
        \begin{array}{rcl}
            (R, u, s_1, T), (T, a, s_2, P), (P, v, s_3, R), (T, a, s_4, \hat{Q}), (\hat{Q}, v, s_5, S), (S, uav, s_6, S) & \in & \Delta_{\cF}^* \\
        \end{array}
    \]

	We can use the same proof strategy defined in the proof of Theorem~\ref{theo:UFA-disambiguation} to obtain a pair of states $p' \in \supp(P')$ and $q' \in \supp(S')$ that share an infinite common future.
	This implies that $P' = S'$ because $\Pi_{\cO_{\ICF}}$ partitions vectors using the oracle $\cO_{\ICF}$.
	However, this is a contradiction to the initial assumption.

	In conclusion, $\cW_{\cF}$ is finWFA.

\end{proof}

\subsubsection{Case polyWFA}

\begin{proof}
  Given a WFA $\cW = (Q, \Sigma, \Delta, I, F)$ over $\bbS$ and any factorization ($\fac$, $\res$).
  Let $\cF = (\Pi_{\cO_{\ECF}}, \fac, \res)$ and $\cW_{\cF} \ = \ (\bbS^Q, \Sigma, \Delta_{\cF}, I_\cF, F_{\cF})$.
  Assume that $\cW_{\cF}$ is finite.

	Suppose that the WFA $\cW_{\cO_{\ECF}}$ is not polyWFA and satisfies (EDA)', then there exist states $R, \hat{Q}, P \in \bbS^Q. \; \hat{Q} \neq P$, weights $s_1, s_2, s_3, s_4 \in \bbS$, a letter $a$ and a word $w$ such that:

	$$
			(R, a, s_1, P), (R, a, s_2, \hat{Q}), (P, w, s_3, R), (\hat{Q}, w, s_4, R) \in \Delta_{\cF}^*
	$$

	We can use the Infinite Run Tree defined in the proof of Theorem~\ref{theo:UFA-disambiguation} to obtain a pair of states $p \in \supp(S_1)$ and $p \in \supp(S_2)$ that share an exponentially infinite common future.
	This implies that $S_1 = S_2$ because $\Pi_{\cO_{\ECF}}$ partitions vectors using the oracle $\cO_{ECF}$.
	However, this is a contradiction to the initial assumption.

	In conclusion, $\cW_{\cF}$ is polyWFA.

	\end{proof}

\subsubsection{Isomorphism}

\begin{proof}
Given a WFA $\cW = (Q, \Sigma, \Delta, I, F)$ over $\bbS$, any factorization ($\fac$, $\res$) and $R \in \{\CF_{\cW}, \ICF_{\cW}, \ECF_{\cW}\}$.
Let $\cF = (\Pi_{\cO_{R}}, \fac, \res)$ and $\cW_{\cF} \ = \ (\bbS^Q, \Sigma, \Delta_{\cF}, I_\cF, F_{\cF})$.
  
If $\cW$ is already UWFA, finWFA or polyWFA respectively and ($\fac$, $\res$) is an identity factorization, then no two distinct states $p, q \in Q$ exist such that $(p, q) \in R$.  
Consequently, $\Pi_{\cO_{R}}$ partitions every vector into vectors with singleton support.  
By Proposition~\ref{prop:wfaFrameworkIsomorphic}, it follows that $\cW_{\cF}$ is isomorphic to $\cA$.

\end{proof}

\subsection{Proof of Theorem~\ref{theo:WFAtermination}}

\WFAtermination*

\begin{proof}
  Given a tropical WFA $\cW = (Q, \Sigma, \Delta, I, F)$ and $R \in \{\CF_{\cW}, \ICF_{\cW}, \ECF_{\cW}\}$.
  Let $\cF = (\Pi_{\cO_R}, \facMohri, \resMohri)$ and $\cW_{\cF} \ = \ (Q_{\cF}, \Sigma, \Delta_{\cF}, I_\cF, F_{\cF})$.

	Suppose that $Q_{\cF}$ is infinite. Then, there exists a subset of states $S = \{q_0, \ldots, q_m\} $ such that there are infinitely many states $V \in Q_{\cF}$ with $\textsf{supp}(V) = S$.
	Note that $m \geq 1$. If $S = \{q_0\}$, there is only a single element, then for any state $V$, we have $V[q_0] = v$, $\fac(V) = v$ and $\res(V)[q_0] = 0$, so there is only one possible $V$.

	Let $K_1$ be the set of words $w$ such that:

	\[
		\exists V \in \vec{\Delta}(\vI, w).\; \textsf{supp(V)} = S
	\]

	For each $w$, the set $|\vec{\Delta}(\vI, w)|$ is finite (Lemma~\ref{lem:wfaDeltaFinite}). Therefore, $K_1$ must be infinite.

	Additionally, for each state $V$, at least one state $q$ exists such that $V[q] = 0$. 
	This is because the factorization is defined as the minimum of the residuals.
	Since $K_1$ is infinite and there exists a residual with value $0$, there exists $i \in [0, m]$ such that for an infinite number of words $w \in K_1$, we have that $\exists V \in \vec{\Delta}(\vI, w). \; V[q_{i}] = 0$.
	Let $K_2 \subseteq K_1$ be the set of such words and without loss of generality, assume that $i = 0$.

	For every word $w \in K_2$, since $\vec{\Delta}(\vI, w)$ has disjoint vector support, there exists a unique vector $V$ such that $q_0 \in \supp(V)$.
	Let $V^w$ be such a vector.

	There exists $j \in [m]$ such that the residual $V^w[q_j]$ is different for an infinite number of words $w \in K_2$.
	Then $q_0$ and $q_j$ are part of the same connected component of $\kappa(R \ \bigcap \ (\supp(V^w) \ \times \ \supp(V^w))$, let $q_0, q_1, \ldots, q_j$ be consecutive states such that $(q_i, q_{i+1}) \in R$.
	Let $z$ be the minimum index such that $\{V^w[q_{z}] \mid w \in K_2 \}$ is finite and  $\{V^w[q_{z+1}] \mid w \in K_2 \}$ is infinite.

	Let $K_3$ be an infinite set of words $w$ with all residuals $V^w[q_{z+1}]$ different.

	Define the difference run weight set $\textsf{diff}(q_z, q_{z+1})$ as:

	\begin{align*}
		\textsf{diff}(q_z, q_{z+1}) = \{\, & \omega_I(\rho_{z+1}) - \omega_I(\rho_z) \;\mid \\
		& \rho_z \in \runs_{\cW}(w, q_z) \land \rho_{z+1} \in \runs_{\cW}(w, q_{z+1}) \land |w| \leq |Q|^2-1 \}
	\end{align*}

	In other words, $\textsf{diff}(q_z, q_{z+1})$ contains the weight differences between all pair of runs that ends in $q_z$ and $q_{z+1}$ respectively.
	Note that $\textsf{diff}(q_z, q_{z+1})$ is finite.

	We will show that the set $\{V^w[q_{z+1}] - V^w[q_z] \mid w \in K_3\} \subseteq \textsf{diff}(q_z, q_{z+1})$.

	For a word $w \in K_3$, let $\rho_z$ and $\rho_{z+1}$ be two victorious runs of $\cW$ over $w$ that end in $q_z$ and $q_{z+1}$ respectively.
	By Corollary~\ref{cor:wfaResidualVictorious}, we know that:

	\[
		\begin{array}{rcl}
			V^w[q_{z+1}] - V^w[q_z] & = & \omega_I(\rho_{z+1}) - \omega_I(\rho_z)\\				
		\end{array}	
	\]

	We have two cases:

	\begin{enumerate}
		\item If $|w| \leq |Q|^2-1$, then $V^w \in \textsf{diff}(q_z, q_{z+1})$.
		\item If $|w| > |Q|^2-1$, then using the pigeonhole principle, it is known that we can factorize the runs as follows (See Lemma 1 in~\cite{mohriTransducer}):
		
		\[
			\begin{array}{rcl}
				\rho_z \ & := & \ q_z^I \, \xrightarrow{w_1/ v_1} \, p \, \xrightarrow{w_2/ v_2} \ p \ \xrightarrow{w_3/v_3} \, q_z \\
				\rho_{z+1} \ & := & \ q_{z+1}^I \, \xrightarrow{w_1/ v_1} \, q \, \xrightarrow{w_2/ v_2} \ q \ \xrightarrow{w_3/ v_3} \, q_{z+1} \\
				|w_2| & > &  0 \\
				|w_1| + |w_3| & \leq&  |Q|^2 - 1 \\
				 
			\end{array}	
		\]

		Then we can define the following runs:

		\[
			\begin{array}{rcl}
				\rho_z' \ & := & \ q_z^I \, \xrightarrow{w_1/ v_1} \, \ p \ \xrightarrow{w_3/ v_3} \, q_z \\
				\rho_{z+1}' \ & := & \ q_{z+1}^I \, \xrightarrow{w_1/ v_1} \, \ q \ \xrightarrow{w_3/ v_3} \, q_{z+1} \\
			\end{array}	
		\]

		Now we have that $\omega(\rho_z) = \omega(\rho_z') + \Delta_{w_2}[p, p]$ and $\omega(\rho_{z+1}) = \omega(\rho_{z+1}') + \Delta_{w_2}[q, q]$. 
		
		Since $\cW$ has the future twins property, $\Delta_{w_2}[p, p] = \Delta_{w_2}[q, q]$, then $\omega_I(\rho_{z+1}') - \omega_I(\rho_z') = V^w[q_{z+1}] - V^w[q_z]$.
		This results in $V^w[q_{z+1}] - V^w[q_z] \in \textsf{diff}(q_z, q_{z+1})$.
	\end{enumerate}

	We proved that the set $\{V^w[q_{z+1}] - V^w[q_z] \mid w \in K_3\}$ is finite. However, this contradicts the assumption that all residuals $V^w[q_{z+1}]$ are different.
	Hence $\cW_{\cF}$ is finite.
\end{proof}